\documentclass[11pt]{article}
\usepackage{fullpage}
\usepackage[T1]{fontenc}
\usepackage{lmodern}
\usepackage[utf8]{inputenc}
\usepackage{bibentry}
\usepackage{amsmath}
\usepackage{amsfonts}
\usepackage{amsthm}
\usepackage{appendix}
\usepackage{color}
\usepackage{algorithm}
\usepackage{algpseudocode}
\usepackage{amssymb}
\usepackage{setspace}
\usepackage{xspace}
\usepackage{graphicx}
\usepackage{relsize}
\usepackage{dsfont}
\usepackage{bbm}
\usepackage{graphicx}
\usepackage{subcaption}
\usepackage{enumitem}
\usepackage[pagebackref]{hyperref}
\usepackage{tikz}
\usetikzlibrary{arrows.meta, positioning, decorations.pathreplacing}

\newcommand{\mcalc}{{\mathcal{C}}}

\newcommand{\mcald}{{\mathcal{D}}}
\newcommand{\mcale}{{\mathcal{E}}}

\newcommand{\mcalg}{{\mathcal{G}}}
\newcommand{\DKL}{{D_{\textnormal{KL}}}}

\newcommand{\mcalx}{{\mathcal{X}}}

\newcommand{\mcala}{{\mathcal{A}}}
\newcommand{\mcalp}{{\mathcal{P}}}

\newcommand{\mcals}{{\mathcal{S}}}
\newcommand{\mbbe}{{\mathbb{E}}}
\newcommand{\costg}{{\textnormal{cost}_G}}
\newcommand{\cutg}{{\textnormal{cut}_G}}

\newcommand{\mbbone}{{\mathbbm{1}}}

\newcommand{\ssuv}{{\mathcal{S}_{\sigma,u,v}}}
\newcommand{\ssuvp}{{\mathcal{S}'_{\sigma,u,v}}}

\newcommand{\simc}{{\sim_{\mcalc}}}
\newcommand{\nsimc}{{\not\sim_{\mcalc}}}
\newcommand{\simcp}{{\sim_{\mcalc'}}}
\newcommand{\nsimcp}{{\not\sim_{\mcalc'}}}

\newcommand{\veps}{\varepsilon}
\newcommand{\poly}{\mathrm{poly}}

\newtheorem{theorem}{Theorem}
\newtheorem{lemma}[theorem]{Lemma}
\newtheorem{claim}[theorem]{Claim}
\newtheorem{definition}[theorem]{Definition}
\newtheorem{corollary}[theorem]{Corollary} 
\newtheorem{fact}[theorem]{Fact}
\newtheorem{proposition}[theorem]{Proposition}

\numberwithin{theorem}{section}

\newtheoremstyle{named}{}{}{\itshape}{}{\bfseries}{.}{.5em}{\thmnote{#3}}
\theoremstyle{named}
\newtheorem*{namedtheorem}{Theorem}

\title{Query Lower Bounds for Correlation Clustering under Memory Constraints}
\author{Sumegha Garg\\Rutgers University\\sumegha.garg@rutgers.edu \and Songhua He\thanks{The author was partially supported by the Rutgers University startup grant of Sumegha Garg and by Karthik C. S. through the National Science Foundation under Grant CCF-2443697.}\\Rutgers University\\sh1511@scarletmail.rutgers.edu \and Periklis A. Papakonstantinou\\Rutgers University\\periklis.research@gmail.com}
\date{}

\begin{document}

\maketitle

\begin{abstract}
This work initiates the study of memory–query tradeoffs for graph problems, with a focus on correlation clustering. Correlation clustering asks for a partition of the vertices that minimizes disagreements: non‑edges inside clusters plus edges across clusters. Our first result is a tight query lower bound: to output a partition whose cost approximates the optimum up to an additive error of $\varepsilon n^2$, any algorithm requires $\Omega(n/\varepsilon^2)$ adjacency-matrix queries. Under memory constraints, we show that even for the seemingly easier task of approximating the optimal clustering cost (without producing a partition), any algorithm in the random query model must make $\gg n/\varepsilon^2$ adjacency-matrix queries. Finally, we prove the first general graph model query lower bound for correlation clustering, where algorithms are allowed adjacency-matrix, neighbor, and degree queries. The latter two bounds are not yet tight, leaving room for sharper results.
\end{abstract}

%The correlation clustering problem has been intensively studied in recent years in several computational models. The best-known polynomial time approximation factor is 1.73 by Cohen-Addad, Lee, Li, and Newman \cite{cohen2023handling}. For sublinear algorithms, Cohen-Addad, Lolck, Pilipczuk, Thorup, Yan, and Zhang \cite{cohen2024combinatorial} gave an $\tilde O(n)$ time algorithm with approximation factor 1.847. From the lower bound side, it is known by Charikar, Guruswami, and Wirth \cite{charikar2005clustering} that the minimum disagreement correlation clustering is APX-hard. However, the constructed graph in their reduction is sparse, where the number of edges is $O(n)$. It remains an open question whether it is still hard to approximate the correlation clustering to within a super-linear in $n$ additive error.

\section{Introduction}

Sublinear-time algorithms for processing massive graphs aim to infer global graph properties while examining only a small fraction of the input. Broadly, the goal is to determine how much information an algorithm must gather in order to solve a given problem; additionally, we ask whether this cost increases when the algorithm operates under memory constraints. In particular, this work investigates these questions for Correlation Clustering (CC), a cornerstone problem in machine learning and network analysis.
Given an input graph for the correlation clustering problem, we interpret edges as pairs of ``similar'' (+) items and non-edges as pairs of ``dissimilar'' ($-$) items. The goal is to partition the vertices into clusters to minimize the total number of ``disagreements'': the number of similar pairs that are cut apart plus the number of dissimilar pairs placed in the same cluster. This problem is closely related to other classic cut problems, such as Max-Cut and Minimum Bisection, and our first two hardness results also apply directly to these.

The study of Correlation Clustering was initiated by \cite{bansal2004correlation}. Since then, the focus has been design of efficient algorithms that produce clusterings with cost is close to optimal. A clustering can be evaluated under two natural objectives: minimizing disagreements or maximizing agreements. These objectives exhibit different behavior with respect to multiplicative approximation. While the maximization version admits a PTAS \cite{bansal2004correlation}, the minimization version is NP-hard to approximate within some constant factor $c>1$ \cite{charikar2005clustering}.
On the algorithmic side, there has been a long sequence of works progressively improving the approximation ratio for minimizing disagreements: from the original 8-approximation of \cite{bansal2004correlation}, to 4 \cite{charikar2005clustering}, 3 \cite{ailon2008aggregating}, 2.5 using LP-based methods \cite{ailon2008aggregating}, 2.06 \cite{chawla2015near}, $(1.994+\epsilon)$  \cite{cohen2022correlation}, 1.73 \cite{cohen2023handling}, and most recently $1.485+\epsilon$ \cite{cao2024understanding,cao2025solving}. Notably, the result of \cite{cao2025solving} achieves this approximation guarantee using a sublinear-time algorithm.
Efficient algorithms and lower bounds for this problem have been studied in a variety of computational models, including the Massively Parallel Computation (MPC) model \cite{behnezhad2022almost, cao2024breaking, cohen2024combinatorial}, the dynamic setting \cite{behnezhad2019fully}, the streaming model \cite{bhaskara2018sublinear, cohen2021correlation, assadi2022sublinear, assadi2023streaming, behnezhad2023single, makarychev2023single, cambus20243+}, and the query model with unbounded computational time \cite{bressan2019correlation, garcia2020query}.

We study the feasibility of approximately solving correlation clustering (with additive error) under various query models and memory constraints.
Our main result establishes a time–space (query–memory) trade-off for correlation clustering in the random query model (to be defined shortly). \ In addition, we strengthen the query lower bounds of \cite{bressan2019correlation}, proving a tight lower bound in the adjacency matrix model and a strong lower bound in the general graph model. We focus on clusterings that minimize disagreements, which is equivalent to maximizing agreements in the additive-error regime.  Throughout, an additive error of $\varepsilon n^2$ means the algorithm's solution cost may be at most $\varepsilon n^2$ worse than the true optimum.

%Our work makes three primary contributions towards understanding the query complexity of correlation clustering. We establish the precise information-theoretic requirements of solving CC in sublinear time (i.e., in terms of query complexity) and, importantly, initiate the study of query-space tradeoffs for graph problems. Throughout, an additive error of $\varepsilon n^2$ means the algorithm's solution cost may be at most $\varepsilon n^2$ worse than the true optimum. We remark that our main conceptual and technical contribution is the memory-query tradeoff for approximating the optimal clustering cost (in the random query model). However, to put things in proper context, we will present this as our second result in the introduction. 

\subsubsection*{A tight query bound in the standard model}

We first consider the classic setting for studying sublinear-time algorithms: an adjacency-matrix query model in which an algorithm has unbounded memory and may adaptively query any pair of vertices to determine whether an edge exists between them. In this model, we prove the following asymptotically tight query lower bound for finding a nearly optimal clustering.

\paragraph{Theorem 1 (Informal -- restated as Theorem~\ref{thm:strong_tight}).} \emph{Every randomized algorithm that finds a clustering within an additive error of $O(\varepsilon n^2)$ must make $\Omega(n/\varepsilon^2)$ adaptive adjacency-matrix queries.}

~

This result improves a previous $\Omega(n/\varepsilon)$ lower bound given in \cite{bressan2019correlation} and matches the known $O(n/\varepsilon^2)$ upper bound for the pure-additive guarantee $\mathrm{OPT}+\varepsilon n^2$ in the same model \cite{bressan2019correlation}, settling the query complexity of CC in the adjacency-matrix query model for this objective.\footnote{A related work \cite{garcia2020query} gives better query complexity for a weaker mixed multiplicative-additive guarantee, gives expected cost $3\mathrm{OPT}+O(n^3/Q)$ using $Q$ queries. Setting $Q=\Theta(n/\varepsilon)$ gives $3\mathrm{OPT}+O(\varepsilon n^2)$, but this is not a pure-additive $\mathrm{OPT}+\varepsilon n^2$ guarantee. Thus our lower bound is compared with the pure-additive upper bound of \cite{bressan2019correlation}.} A proof overview is provided in Section \ref{sec:poadjacency}.
Next, we show that for the seemingly easier problem of approximating the optimal clustering cost, the query complexity increases dramatically under memory constraints, albeit in the random query model.

\subsubsection*{A query-space tradeoff for approximating the optimal clustering cost}

What happens when memory is scarce? To understand the additional query cost of limited memory algorithms, we turn to a model where the algorithm receives a stream of uniformly random vertex pairs and indicators of whether these pairs connect or not in the underlying graph. This is the random-query model. 
For computing boolean functions, time-space tradeoffs in the random query model have been introduced and studied in~\cite{raz2020random, dinur2024time}.
%which were extremely successful in studying time-space tradeoffs \cite{raz2017time,raz2018fast,garg2019time,garg2020time,raz2020random,garg2021memory}.
%\shnote{Is the above accurate?}
In our second result, the task is to estimate the \emph{value} of the optimal clustering, a seemingly easier goal than finding the partition.
%{thm:loose_bound_main}

\paragraph{Theorem 2 (Informal -- restated as Theorem~\ref{thm:streaming_main}).} \emph{In the random-query model, an algorithm using only $\gamma\sqrt{n}$ bits of memory to estimate the optimal CC cost, within additive error of $\varepsilon n^2$, needs $q$ queries where $$q=\begin{cases}
    \Omega\left(\min\left\{\frac{n}{\varepsilon^2\sqrt{\gamma}},\frac{n^{3/2}}{\gamma}\right\}\right)&\textnormal{if }\gamma<1\\
    \Omega\left(\frac{n}{\varepsilon^2}\right)&\textnormal{if }\gamma\ge 1
\end{cases}$$}

~

To the best of our knowledge, this is the first non-trivial query–space trade-off for any approximation problem, and more broadly for any graph problem in the random query model. Our result demonstrates that memory is not free: when the available space is subpolynomial (i.e., $\gamma = o(1)$), the query complexity necessarily exceeds the $\tilde{O}(n/\varepsilon^2)$ baseline -- for example, becoming $\tilde{\Omega}\left(\min\left(\frac{n^{1+1/4}}{\varepsilon^2}, n^2\right)\right)$ for polylogarithmic-space algorithms. This shows that memory and query complexity are provably and intrinsically intertwined.\footnote{Here is a word of caution on interpretation. Relative to Theorem~1, in Theorem~2 we change three parameters of the setting. 
(i) We impose a memory cap. 
(ii) We switch to the random–query model~\cite{raz2020random}. 
(iii) We study the \emph{value} (cost) problem rather than the search problem. 
Our tradeoff should be read and contrasted with the other results in light of all three changes.}
The proof of the above theorem faces various technical challenges, which we elaborate on in Section \ref{sec:techniques_time_space}.
%Technically, our proof incorporates information theory and a boolean Fourier analytic framework for random-query, low-space algorithms, non-trivially adapting \cite{kapralov2014streaming}. Unlike \cite{kapralov2014streaming}, which gives only streaming lower bounds, our method yields query–space tradeoffs (see Section~\ref{sec:techniques_time_space} for a detailed discussion of our proof techniques). The first tradeoff of this flavor is \cite{raz2020random} for random-query Boolean functions (not graph problems) via branching programs. Theorem~2 is the graph analogue of \cite{raz2020random}, extending the Fourier spine of \cite{kapralov2014streaming} to dense graphs with repeated queries.

The setting of Theorem~2 differs from that of Theorem~1 in two ways: we use the random-query model instead of arbitrary queries, and we consider the cost (value) variant of the problem. We briefly explain these choices.
First, allowing arbitrary queries would yield lower bounds that translate into non-trivial space lower bounds for space-bounded Turing machines on a natural graph problem.  Thus, restricting how we access the input circumvents this difficulty and makes the problem amenable to analysis.
Second, we use the value formulation for two reasons: simplicity and, more importantly, because under $O(\sqrt n)$ space the search version becomes vacuous -- storing a near-optimal clustering already requires $\Omega(n)$ bits, even for two clusters. Thus any lower bound would be dominated by output size rather than computation. The value formulation removes this artifact and helps in exhibiting the first memory–query trade-off for graph problems. We view Correlation Clustering as a canonical testbed, and see this work as a first step toward a broader study of query–space trade-offs across graph primitives.

\subsubsection*{A lower bound in the general graph model}

Finally, we consider a stronger query model, the \emph{general graph model}. In this model an algorithm can make pair, degree, and neighbor queries. This model was introduced and studied in \cite{kaufman2004tight, alon2008testing}. This model is significantly stronger, as neighbor queries can reveal global graph structure and create complex statistical dependencies that are challenging to analyze. Specifically, it is known from \cite{kaufman2004tight} that testing bipartiteness in the general graph model is strictly stronger than algorithms in the adjacency list or (exclusively)  adjacency matrix models alone. However, proving lower bounds in the general model is conceptually difficult.

\paragraph{Theorem 3 (Informal -- restated as Theorem~\ref{thm:loose_bound_main}).} \emph{In the general graph model, any algorithm that finds a clustering with an additive error of $\varepsilon n^2$ requires $\Omega(n/\varepsilon)$ queries.}

~

While this bound is not as strong as Theorem~1, to the best of our knowledge it establishes the first query lower bounds for correlation clustering in the general graph model. In the proof, we deal with the challenges of this model by analyzing query interactions on carefully constructed regular graphs (graphs where symmetry in some sense turns neighbor queries useless).

\subsection{Our techniques}
\label{sec:techniques}

Below, we outline the main ideas behind our lower bound proofs. 

\subsubsection{Memory-query tradeoffs for approximating the clustering cost in the random query model}
\label{sec:techniques_time_space}

To illustrate the proof techniques, we fix $\veps = n^{-1/4}$ and consider $(\poly\log n)$-memory algorithms that approximate the clustering cost within an additive error of $\veps n^2$. Given a graph $G$ on $n$ vertices, in the random query model, at each time step $t$, the algorithm receives a uniformly random vertex pair $(u,v)$ along with the indicator bit specifying whether $(u,v)$ is an edge of $G$. Using a VC-dimension argument \cite{bressan2019correlation}, it follows that $O(n/\veps^2) = O(n^{1.5})$ random queries suffice to produce a partition whose cost is within an additive error of $\veps n^2=n^{7/4}$ from optimal. Although representing such a partition requires $\Omega(n)$ bits of memory, it is natural to expect that just approximating the optimal clustering cost might be possible with significantly less memory\footnote{Indeed, low-space streaming algorithms are known for certain approximation parameters~\cite{bhaskara2018sublinear,assadi2023streaming}.}. Before our work, it was not known whether even $O(\log n)$-memory algorithms could achieve an additive $O(n^{7/4})$-approximation using only $O(n^{1.5})$ random queries\footnote{While \cite{assadi2023streaming} established $\poly\log n$-space hardness for approximating the clustering cost within an additive error of $o(n^2)$ in the worst-case streaming model, proving hardness in the random query model introduces new technical challenges, even relative to the random-order streaming setting.}. We show that any $\poly\log n$-memory algorithm achieving this guarantee must in fact use $\tilde{\Omega}(n^{7/4})\gg n^{1.5}$ random queries.

The random-query model was introduced in \cite{raz2020random} to study memory-query tradeoffs for computing $n$-bit Boolean functions, when the algorithm receives a random input bit at each time-step. This and the subsequent work of~\cite{dinur2024time} established non-trivial tradeoffs for functions with high sensitivity and total influence, respectively. However, for promise or approximation problems, no small set of edges has high influence, and therefore these prior techniques do not apply to our setting. Instead, our result builds on the approach of \cite{kapralov2014streaming}, which established hardness of approximating MAX-CUT (and the clustering cost) within a multiplicative factor of $(2-\delta)$ over sparse graphs, in the one-pass streaming model. Specifically, the paper leveraged tight lower bounds for a two-player one-way communication problem -- called the Distributional Boolean Hidden Partition (D-BHP) problem -- to show that any (random-order) streaming algorithm that distinguishes between a bipartite graph and a random graph must use $\Omega(\sqrt{n})$ bits of memory. 

There are two main challenges in extending this approach to prove an $\tilde{\Omega}(n^{7/4})$ random-query lower bound for additive approximation of the optimal clustering cost (under memory constraints). First, \cite{kapralov2014streaming} establishes lower bounds only for sparse graphs with $O(n)$ edges; for such sparse graphs, an additive approximation of $O(n^{1.5})$ to the clustering cost is trivial. Second, in the random-query model, after $\Omega(n)$ queries, edge repetitions occur with high probability, unlike in the random-order one-pass streaming model, where each edge is guaranteed to appear exactly once. The ability to go over the same edge twice significantly increases the technical difficulties in  proving hardness results\footnote{Only recently, the breakthrough work of \cite{fei2025multi} showed that any constant-pass streaming algorithm for distinguishing between a sparse bipartite graph and a random graph requires $\Omega(n^{1/3})$ bits of space. Since multi-pass streaming algorithms read the stream in the same order, this model is incomparable to the random query model.}.

To overcome the first challenge, we introduce and study a noisy variant of D-BHP, which we call the Perturbed Distributional Boolean Hidden Partition (PD-BHP) problem. This formulation allows us to work with dense graphs that satisfy the bipartite partition ``approximately''. Starting with the Boolean Hidden Matching (BHM) problem introduced by \cite{gavinsky2007exponential}, BHM, D-BHP, and their variants have been extensively used to prove streaming and sketching lower bounds (see, e.g., \cite{verbin2011streaming, kapralov2014approximating, kapralov2014streaming, kogan2015sketching, guruswami2017streaming, kapralov20171+, kallaugher2018sketching, guruswami2019streaming, chou2020optimal, assadi2022hierarchical}). The tight bounds we establish for PD-BHP may therefore be of independent interest. Similar to \cite{kapralov2014streaming}, we use a hybrid argument to lift the two-player one-way communication lower bound to a multi-player one-way communication bound, where each player receives an independent graph obeying the hidden partition. In the random-query model, however, there is a single underlying graph $G$ dictating the inputs to all players; in particular, repeated queries to the same edge must be consistent with previous observations. This brings in the second challenge. To overcome it, we show that any algorithm capable of distinguishing between the case where all edge queries are consistent and the case where each queried edge is independently resampled must use either $\poly\log n$ bits of memory or $\Omega(n^2 / \poly\log n)$ random queries. We refer to this as the same vector problem, and we establish tight memory–query tradeoffs for it.

\paragraph{Tight communication complexity for PD-BHP problem:} Consider the following two-player communication problem. Alice receives a uniformly random $n$-bit vector $x\in\{0,1\}^n$. Bob receives a graph $G=(V,E)$ on $n$ vertices with uniformly random set of $r$ edges, and a vector $w \in \{0,1\}^r$. In the NO case, each $w_i$ is an independent uniformly random bit. In the YES case, for the $i$th edge $(u_i,v_i) \in E$, we have $$w_i=\begin{cases}
        x_{u_i} \oplus x_{v_i}, & \text{with probability } \tfrac{1}{2} - 10\varepsilon, \\[6pt]
        1 - (x_{u_i} \oplus x_{v_i}), & \text{with probability } \tfrac{1}{2} + 10\varepsilon.
    \end{cases}$$
In contrast, in the D-BHP problem studied by \cite{kapralov2014streaming}, the YES case is noiseless, that is $w_{i}=x_{u_i}\oplus x_{v_i}$. The goal is for Bob to distinguish between the two cases with as little communication from Alice as possible. We show that when Bob receives $r = \alpha n/\varepsilon^2$ edges (for some $0<\alpha<1$), and Alice sends at most $\gamma\sqrt{n}$ bits (for some $\gamma>0$), Bob’s distinguishing advantage is at most $O((\gamma+\alpha)\alpha^{1/2})$. This bound is tight: if Alice simply sends the first $O(\sqrt{n/\alpha})$ coordinates of $x$ to Bob, then with high probability, the graph $G$ will contain at least $\Omega(1/\varepsilon^2)$ edges whose endpoints lie entirely within these coordinates, enabling Bob to distinguish the two cases with constant advantage. To lift the two-player communication bound to the multi-player setting, we parametrize $0 < \alpha = \gamma < 1$. Our proof builds on the Fourier-analytic approach of \cite{kapralov2014streaming} for the D-BHP problem. A key challenge in extending their proof is that, since Bob now receives $\gg n$ edges, the graph $G$ necessarily contains many cycles (of length at least $3$). In contrast, the argument of \cite{kapralov2014streaming} relies crucially on the fact that when Bob receives $\ll n$ edges, the graph is cycle-free. By carefully balancing the number of cycles against the distinguishing advantage they provide, we are able to extend the Fourier-analytic proof to handle the noisy case on \emph{denser} graphs.

\paragraph{Same vector problem:} Let $N = \binom{n}{2}$. Consider the following distinguishing problem, where the goal is to decide whether the random queries are ``consistent''. At each time step $t$, the algorithm receives a random index $i_t \in [N]$ and a bit $b_t$. In the NO case, $b_t$ is uniformly random, whereas in the YES case, $b_t = x_{i_t}$, where $x \in \{0,1\}^N$ is fixed at the start. We show that any algorithm that solves the distinguishing problem, in $\sqrt{N} < T < N$ time-steps, must use at least $\Omega(N/T)$ bits of memory.\footnote{Recall that by the Birthday Paradox, the probability of sampling the same index twice is negligible when $T\ll\sqrt N$. Our result covers the entire non-trivial range of $T$ where collisions are possible but not guaranteed.} This bound is tight (up to log factors): by storing the first $s$ queries, an algorithm will, with high probability, encounter a repeated index within $O(N/s)$ time steps, and thus distinguish with constant advantage. To obtain the tight result, we use an information-theoretic potential function that measures the maximum progress any $s$-memory algorithm can make towards distinguishing at a time-step. For our main theorem, we generalize the same vector problem to non-uniform bits.

\subsubsection{Tight query lower bound for approximate clustering in the adjacency-matrix query model}\label{sec:poadjacency}
%\sgnote{change the heading to only correlation clustering?}

For every fixed parameter $\varepsilon\in\left(\omega\left(\sqrt{\frac{\log n}{n}}\right),0.001\right)$,
we establish an $\Omega(n/\varepsilon^{2})$ query lower bound for approximating the correlation clustering partition within additive error $\varepsilon n^2$ (Theorem~\ref{thm:strong_tight}), in the adjacency matrix model. Our lower bound improves the existing $\Omega(n/\varepsilon)$ lower bound \cite{bressan2019correlation}, and matches the upper bound given in the same paper \cite{bressan2019correlation}. The same proof structure also yields the same query lower bound to searching a max cut or a minimum bisection with $\varepsilon n^2$ additive error (Theorem~\ref{thm:strong_tight_mcut_mbis}).

The prior $\Omega(n/\varepsilon)$ lower bound is built on a structured hard distribution. Specifically, the vertex set is partitioned into two parts, $A$ and $B$, where $|A|=0.9n$ and $|B|=0.1n$. The subgraph induced by $A$ consists of $1/\varepsilon$ disjoint cliques, $A_1,\dots,A_{1/\varepsilon}$, of equal size, while the subgraph induced by $B$ is an empty graph. For every vertex $v$ in $B$, $v$ is connected to all vertices in a single clique $A_i$, with $i$ chosen uniformly at random. With high probability, the optimal clustering groups each $A_i$ with the vertices from $B$ connected to it. An algorithm must query $\Omega(1/\varepsilon)$ edges for each vertex in $B$ to find the clique it connects to, thus a lower bound of $\Omega(n/\varepsilon)$ follows. This lower bound is optimal for this distribution. To improve upon this bound, we design a noisy distribution with less structure, where each query yields less information.

Let $\rho = 100\varepsilon$. We sample a random \emph{underlying partition} $P \in \{0,1\}^n$ that divides the vertex set into two parts $V = V_0 \cup V_1$. For each pair of vertices $u, v$ from the same part, we connect them independently with probability $1/2 + \rho$; for pairs from different parts, with probability $1/2 - \rho$.

For a balanced partition $P$ (i.e., $||V_0| - |V_1|| \le \tilde O(\sqrt{n})$), concentration bounds imply that the cut size of $P$ is $(1/2+\rho)n^2/4\pm O(n^{1.5})=n^2/8+25\varepsilon n^2\pm \tilde O(n^{1.5})$ with high probability. In contrast, the expected cut size of a uniformly random partition $P'\in\{0,1\}^n$ is $n^2/8$, leaving a gap much larger than $\varepsilon n^2$. The key insight is that, unless the algorithm recovers $\Omega(n)$ bits of the partition $P$, it cannot produce a clustering with $\le \varepsilon n^2$ additive error. In addition, each adjacency matrix query can at most reveal $O(\varepsilon^{2})$ information about $P$.

To formalize the intuition, our lower bound proof consists of three steps
\begin{enumerate}
    \item We apply a generalized version of Fano’s inequality tailored for approximation problems. (See Lemma~\ref{lem:strong_fano}.)
    \item We show that, with high probability, every output clustering that is far from $P$ incurs high additive error. (See Lemma~\ref{lem:strong_hamming_dis_cc}.)
    \item We bound the information each query reveals about $P$, showing that each adjacency-matrix query contributes at most $O(\varepsilon^{2})$ bits of information. (See Lemma~\ref{lem:strong_mutual_information})
\end{enumerate}

Fano’s inequality connects the success probability of a randomized algorithm to the information its queries reveal about the answer. We extend it to handle approximate answers. Specifically, let $P \in \{0,1\}^n$ be the underlying partition, $\sigma_\Pi$ the query history of algorithm $\Pi$, $\mcala_G$ the set of approximately correct answers (i.e., those within additive error $\le \varepsilon n^2$), and $p_e$ the error probability. Then our generalized Fano's inequality implies:
$$p_e\cdot n+(1-p_e)\cdot \log|\mcala_G|+H(p_e)\ge H(P|\sigma_\Pi)$$

However, the number of approximate answers for correlation clustering could be as large as $\exp(\Theta(n\log n))$. No non-trivial bound to $p_e$ could be obtained from Fano's inequality. We resolve it by showing that clusterings with low errors must have their two largest clusters closely aligned with $P$. Smaller clusters can be safely ignored. Therefore, we reduce the correlation clustering from the the problem that only requires outputting the largest two clusters, where the number of approximately correct answers can be bounded by $\exp(O(n))$.

Finally, we upper bound the mutual information between $P$ and the query history of the algorithm. For each query to a pair of vertices $(u,v)$, depending on whether $u$ and $v$ are in the same part of $P$, their connectivity follows an independent distribution either from $\text{Bern}(1/2 - \rho)$ or $\text{Bern}(1/2 + \rho)$. Therefore, conditioned on the query history, the probability of $(u,v)\in E$ follows a mixture distribution of $\text{Bern}(1/2 - \rho)$ and $\text{Bern}(1/2 + \rho)$. By bounding the KL divergence one can see that each query reveals at most $O(\varepsilon^2)$ information about $P$.

\subsubsection{Query lower bounds for approximate clustering in the general graph model}\label{sec:pogeneral}

The previous result gives a tight lower bound in a weak model, the adjacency matrix model. Here we give a weaker lower bound in the general graph model, which is a query model introduced by \cite{kaufman2004tight} in property testing. In this model, algorithms have both adjacency-matrix access and adjacency-list access to the input graph. Specifically, we show that every algorithm that approximates the correlation clustering within additive error $\varepsilon n^2$ in the general graph model requires $\Omega(n/\varepsilon)$ queries, for every $\varepsilon\in\left(\omega\left(\frac{\log n}{n}\right),10^{-6}\right)$.

Our proof to the tight query lower bound in the adjacency-matrix model relies on the fact that each query is very ``local''. However, when adjacency-list query access is allowed, neighbor queries and degree queries create statistical dependencies between different parts of the input graph that are hard to bound. In addition, the hard input distribution used in \cite{bressan2019correlation} fails when neighbor queries are introduced. Therefore, we will use a different proof strategy.

First, we construct an adversarial input distribution on regular graphs of fixed degree. In this way, we can ignore the information brought by degree queries. The input graph is sampled together with an \emph{underlying clustering} $\mcalc$. Specifically, let $\mcalc=(C_1,\dots,C_k)$ be a random clustering to the vertex set $V$ that contains $k=0.01/\varepsilon$ clusters, each of size exactly $n/k=100\varepsilon n$. Given $\mcalc$, each cluster is a clique in $G=(V,E)$. For every pair of clusters $C_\alpha,C_\beta$ in $\mcalc$, their induced bipartite graph $G[C_\alpha,C_\beta]$ is a uniformly random $(\varepsilon n)$-regular bipartite graph. We construct the input graph such that every vertex is incident to redundant edges connecting it to vertices in different clusters. And it becomes hard to find a neighbor from the same cluster. Yet the optimal clustering is still dominated by the underlying clustering. In addition, bipartite subgraphs induced by different pairs of clusters are generated independently. This decorrelates the dependencies between parts of the graph, and simplified our proof. Intuitively, one needs to make $\Omega(k)=\Omega(\varepsilon^{-1})$ queries to each vertex to find its cluster. Our proof formalizes this intuition.

For proving a lower bound, we assume that the algorithm, when making pair queries and neighbor queries, also get to know whether the two vertices are from the same cluster in $\mcalc$ or not. We call a vertex \emph{revealed} if a pair/neighbor query finds an edge between this vertex and another vertex from the same cluster. Fix a desired query lower bound $t=C\cdot n/\varepsilon$ for a small constant $C>0$. We show that every algorithm making at most $t$ queries will reveal fewer than $0.001n$ vertices with high probability. This is achieved by showing that under certain conditions, every possible query will reveal a new vertex with probability $\le O(\varepsilon)$. (See Lemma~\ref{lem:loose_bound_core_prob} for details.)

Given a query history of length $\le t$ that reveals $\le 0.001n$ vertices, we show that every possible output clustering has a high cost with high probability. To achieve it, we first show that with high probability \emph{every} clustering with low cost must be \emph{close} to the underlying clustering $\mcalc$, where we capture the closeness of two clusterings $\mcalc,\mcalc'$ by the number of pairs $(u,v)$ where they belong to the same cluster in one clustering, but belong to different clusters in another. (See Lemma~\ref{lem:loose_bound_case_i} for details.)

Lastly, by reusing our probability bound in Lemma~\ref{lem:loose_bound_core_prob}, we show in Lemma~\ref{lem:loose_bound_conditionalbd} that when no more than $0.001n$ vertices are revealed, for every clustering $\mcalc'$, its symmetric difference to $\mcalc$ is large with high probability. Therefore, we conclude that for every possible output clustering, with high probability, it has a high cost.

\paragraph{Toward a tight lower bound in the general graph model.} The central technical challenge towards a tight lower bound lies in constructing a regular graph distribution with a hidden bi-partition while successfully decorrelating the complex statistical dependencies introduced by adaptive neighbor queries, which allow algorithms to efficiently explore local structures. We conjecture that the tight lower bound is achievable via a locally uninformative construction and the analytical toolkit established in this work.

%\subsubsection{Future Directions}
%\sgnote{Add that it is interesting to go beyond our lower bound as well superlinear query?}

\section{Preliminaries}

We always denote $[k]$ as the set $\{1,\dots, k\}$. For every binary string $S$, we use $|S|$ to denote its Hamming weight. For every two binary strings $S,S'$ of the same length, we use $d_H(S,S')$ to denote their Hamming distance.

\paragraph{Concentration bounds.} We use the following standard forms of concentration bounds in this work.

\begin{proposition}[Markov's inequality]
    \label{prop:markov}
    For every non-negative random variable $X$ and any $a>0$,
    $$\Pr[X\ge a]\le \frac{\mbbe[X]}{a}$$
\end{proposition}

\begin{proposition}[Hoeffding's inequality]
    \label{prop:hoeffding}
    Let $X_1,X_2,\dots, X_n$ be independent random variables such that $a_i\le X_i\le b_i$. Let $X=\sum_{i=1}^n X_i$. Then, for every $t>0$,
    $$\Pr[|X-\mbbe[X]|\ge t]\le 2\cdot \exp\left(-\frac{2t^2}{\sum_{i=1}^n(b_i-a_i)^2}\right)$$
\end{proposition}

\begin{proposition}[Chernoff bound]
    \label{prop:chernoff}
    Let $X_1,X_2,\dots, X_n$ be independent random variables such that $X_i\in [0,1]$. Let $X=\sum_{i=1}^nX_i$. Then, for every $\delta>0$,
    $$\Pr[|X-\mbbe[X]|\ge \delta\cdot \mbbe[X]]\le 2\cdot \exp\left(-\frac{\delta^2}{2+\delta}\cdot \mbbe[X]\right)$$
\end{proposition}

\begin{proposition}[The generic Chernoff bound]
    \label{prop:generic_chernoff}
    Let $X$ be a random variable. For every $a$ and every $t>0$,
    $$\Pr[X\ge a]\le \mbbe[\exp(tX-ta)]$$
\end{proposition}

We also need the following Chernoff bound to random variables that are not independent, but with bounded conditional probability. The proof is deferred to Appendix~\ref{appendix:chernoff_bound}.

\begin{lemma}[Chernoff bound for conditionally bounded bits]
    \label{prop:adaptive}
    Fix $p\in(0,1)$ and $\delta>0$. Let $X_1,\dots, X_n\in \{0,1\}$ be random variables such that for every $i\in[n]$ and every $a_1,\dots, a_{i-1}\in\{0,1\}$ such that $\sum_{t=1}^{i-1}a_{t}<(1+\delta)\cdot pn$,
    $$\Pr[X_i=1|X_1=a_1,\dots,X_{i-1}=a_{i-1}]\le p$$
    Let $X=\sum_{i=1}^n X_i$. Then
    $$\Pr[X\ge (1+\delta)\cdot pn]\le 2\cdot \exp\left(-\frac{\delta^2}{2+\delta}\cdot pn\right)$$
\end{lemma}

\paragraph{Basics of information theory.} Given a random variable $Z$, we use $H(Z)$ to denote the Shannon entropy of $Z$, i.e., $H(Z)=\sum_z\Pr[Z=z]\log(1/\Pr[Z=z])$. For a scalar $p\in(0,1)$, $H(p)$ denotes the entropy of the Bernoulli random variable with $p$ probability to be $1$. We use $I(X;Y|Z)$ to denote the mutual information between $X$ and $Y$ conditioned on the random variable $Z$.  $I(X;Y|Z)=H(X|Z)-H(X|Y,Z)$, where $H(X|Y)=\mbbe_{y}H(X|Y=y)\le H(X)$. Next, we describe some of the properties of mutual information used in the paper.

\begin{enumerate}
\item \label{itemmi1} (Chain Rule) $I(X,Y;Z)=I(X;Z)+I(Y;Z|X)$.
\item \label{itemmi4} For discrete random variables, $X,Y,Z$, $I(X;Y|Z)=0\iff X\perp Y|Z$.
\item \label{itemmi2} If $I(Z_1;Z_2|X,Y) = 0$, then $I(X;Z_2|Y) \ge I(X;Z_2|Y,Z_1)$.
\item \label{itemmi3} If $I(Z_1;Z_2|Y) = 0$, then $I(X;Z_2|Y)\le I(X;Z_2|Y,Z_1)$.
\end{enumerate}

Property \ref{itemmi1} follows from the chain rule for Shannon entropy. For property \ref{itemmi4}, it is easy to see that if $X\perp Y|Z$, then $I(X;Y|Z)=0$; the other direction uses strict concavity of the $\log$ function. Properties \ref{itemmi2} and \ref{itemmi3} follow from the observation that 
\[I(X;Z_2|Y)+I(Z_1;Z_2|X,Y)=I(X,Z_1;Z_2|Y)=I(Z_1;Z_2|Y)+I(X;Z_2|Y,Z_1).\]
As mutual information is non-negative, if $I(Z_1;Z_2|X,Y)=0$, then $I(X;Z_2|Y)\ge I(X;Z_2|Y,Z_1)$ (because $I(Z_1;Z_2|Y)\ge 0$) and if $I(Z_1;Z_2|Y)=0$, then $I(X;Z_2|Y)\le I(X;Z_2|Y,Z_1)$.

\paragraph{Graph notations and problem definitions.}
Throughout this paper, we use $G=(V, E)$ to denote the input graph. We denote $n=|V|$ as the number of vertices. We use $\binom V2$ to denote the set of unordered pairs of vertices. Given two disjoint subsets $V_1,V_2$ of a graph $G$, we use $G[V_1,V_2]$ to denote the induced bipartite subgraph where $V_1$ and $V_2$ are the two parts. Given a vertex $u$ and a number $i$, we use $N(u,i)$ to denote the $i$-th neighbor of vertex $u$.

We study the following three problems in this work in both versions where the output is an approximately optimal cost or a partition that yields an approximately optimal cost: the correlation clustering problem, the max cut problem, and the minimum bisection problem.

\paragraph{Correlation clustering.} Given an undirected and unweighted graph $G=(V,E)$. A \emph{clustering} $\mcalc$ of the graph is a partition $C_1,C_2,\dots,C_k$ of its vertex set $V$ where $k$ is unfixed. A clustering also defines an equivalence between vertices: $u\simc v$ if $u$ and $v$ belong to the same cluster, and $u\nsimc v$ otherwise. The cost of a clustering is defined as the total number of missing edges inside each cluster plus the total number of edges crossing different clusters. Formally, we define it as
$$\costg(\mcalc):=|\{(u,v)\not\in E:u\simc v\}|+|\{(u,v)\in E:u\nsimc v\}|$$
Fix a parameter $\varepsilon\in(0,1/2)$, the \emph{correlation clustering cost} problem asks for an approximate value $A$ such that $|A-\min_{\mcalc}\costg(\mcalc)|\le \varepsilon n^2$ given $G$ as the input. The \emph{correlation clustering partition} problem asks for a clustering $\mcalc'$ such that $\costg(\mcalc')\le \min_{\mcalc}\costg(\mcalc)+\varepsilon n^2$. 

A popular and equivalent definition of correlation clustering is based on a signed graph, where each edge is associated with a positive ($+$) or a negative ($-$) sign. The goal is to find a clustering that minimizes the number of disagreements, which corresponds to our cost function $\costg(\mcalc)$. The problem can also be framed as maximizing the number of agreements. In the additive error setting, these two goals are equivalent, as the sum of agreements and disagreements is the fixed total number of pairs of vertices, $\binom{n}{2}$. In our paper, we model positive edges as existing edges and negative edges as non-edges, and we focus on the case of a complete signed input graph. Our lower bounds for this complete graph model also imply a lower bound for general graph cases.

\paragraph{Max cut.} Given an undirected graph $G=(V,E)$. A \emph{cut} of the graph is a partition of the vertex set $V$ into two disjoint parts $V_0,V_1$. We present the partition by a vector $P\in\{0,1\}^n$ where $P_v=0$ for $v\in V_0$ and $P_v=1$ for $v\in V_1$. The \emph{cut size} of $P$ in $G$ is defined as the number of edges crossing the two parts $\cutg(P):=\left|\{(u,v)\in E:u\in V_0,v\in V_1\}\right|$. Fix a parameter $\varepsilon\in(0,1/2)$, the \emph{max cut size} problem asks for an approximate value $A$ such that $|A-\max_P\cutg(P)|\le \varepsilon n^2$ given $G$ as the input. The \emph{max cut partition} problem asks for a partition $P$ such that $\cutg(P)\ge \max_{P'}\cutg(P')-\varepsilon n^2$. 

\paragraph{Minimum bisection.} Assume $n$ is even. Given an undirected and unweighted graph $G=(V,E)$. A \emph{bisection} of the graph is a partition of its vertex set into two disjoint parts of \emph{equal} size. We denote the partition as a vector $P\in\{0,1\}^n$ where $|P|=n/2$. Fix a parameter $\varepsilon\in(0,1/2)$, the \emph{minimum bisection size} problem asks for an approximate value $A$ such that $|A-\min_{P:|P|=n/2}\cutg(P)|\le \varepsilon n^2$ given $G$ as the input. The \emph{minimum bisection partition} problem asks for a bisection $P'$ such that $\cutg(P')\le \min_{P:|P|=n/2}\cutg(P)+\varepsilon n^2$.

\section{Memory-query tradeoffs in the random query model}

We study the memory-query tradeoffs of the correlation clustering cost problem in the \emph{random query model}. This model was introduced in \cite{raz2020random} and subsequently studied in \cite{dinur2024time}. Our main result is the first memory-query tradeoff lower bound for graph problems in the random query model. Different from the previous works, we adapt the proof structure from \cite{kapralov2014streaming}, which is based on Fourier analysis. 

\paragraph{The random query model.} 
In the random query model, an algorithm can make random queries to the input graph $G=(V,E)$. Each query returns an independent and identically distributed (i.i.d.) uniformly random pair of vertices $(u,v)$ from $\binom{V}{2}$ and an indicator $\mbbone_{(u,v)\in E}$ of whether or not $(u,v)\in E$. We denote the return of the $i$-th query as $(u_i,v_i,e_i)$, where $e_i=\mbbone_{(u_i,v_i)\in E}$.

\begin{theorem}
    \label{thm:streaming_main}
    Let $G=(V,E)$ be an undirected simple graph with $n$ vertices. Let $\Pi$ be any randomized algorithm that, in the random query model, approximates the correlation clustering cost of $G$ to within an additive error of $\varepsilon n^2$ with probability at least $99/100$. For this algorithm, if the worst-case query complexity is $q$ and the space used is at most $\gamma \sqrt n$ bits, then the following lower bound holds:
    $$q=\begin{cases}
        \Omega\left(\min\left(\frac{n}{\varepsilon^2\sqrt{\gamma}}, \frac{n\sqrt n}{\gamma}\right)\right)&\textnormal{if }\gamma<1\\
        \Omega\left(n/\varepsilon^2\right)&\textnormal{if }\gamma\ge 1
    \end{cases}$$
    for parameters $\varepsilon\in \left(\omega\left(\frac{1}{\sqrt n}\right),0.05\right)$ and $\gamma>\omega\left(\frac{\log n}{\sqrt n}\right)$.
\end{theorem}

We note that our lower bound, which holds for deterministic algorithms, also applies to randomized algorithms by a standard application of Yao's minimax principle. In addition, the same lower bound applies to the max cut and minimum bisection problems. The three problems share the same hard distribution. We defer the lower bound statement and its proof to the latter two problems to Appendix~\ref{appendix:streaming_distinguishability}, and focus on the correlation clustering problem in this section.

%\sgnote{Mention that the lower bound is for deterministic and randomized can be converted into deterministic by hardcoding the randomness. In particular lemma's, explicitly mention deterministic protocols and algorithms.}

Our proof extends the streaming lower bound for max cut in \cite{kapralov2014streaming}. They showed an $\tilde\Omega(\sqrt n)$ space lower bound for approximating max cut with $(2-\epsilon)$-multiplicative error for any constant $\epsilon>0$. We adapt their proof to establish lower bounds for a wide range of additive errors.

The core of our proof strategy is a reduction to two auxiliary problems. The first is a noisy variant of the \emph{distributional Boolean Hidden Partition (D-BHP)} problem introduced in \cite{gavinsky2007exponential} and extended in \cite{verbin2011streaming, kapralov2014streaming}, and the second is what we call the \emph{same vector problem}. We note that the query-space lower bound for the case $\gamma\ge 1$ is a direct corollary of the lower bound to the first problem. And the desired lower bound for the case $\gamma<1$ follows from establishing lower bounds for both problems. We will define these auxiliary problems formally in the following subsections.

\subsection{Input distribution}
\label{sec:streaming_input_dis}

We use an input distribution for the correlation clustering that differs from the one in \cite{kapralov2014streaming}. While they studied the indistinguishability between sparse random graphs, we focus on the indistinguishability between two types of dense graphs: a random dense Erd\H{o}s-R\'enyi graph and a random dense Erd\H{o}s-R\'enyi graph with $\Theta(\varepsilon n^2)$ edges perturbed in expectation.

Let $\mathcal{G}_{n,p}$ denote the distribution of Erd\H{o}s-R\'enyi graphs with parameters $n$ and $p$. In a random graph from this distribution, each of the $\binom{n}{2}$ possible edges is included with an independent probability $p$.

We define two input graph distributions, $\mathcal{G}^Y$ and $\mathcal{G}^N$.
\begin{itemize}
    \item \textbf{The YES distribution ($\mathcal{G}^Y$):} A random graph from $\mcalg^Y$ is sampled as follows. We first uniformly generate a partition $P \in \{0,1\}^n$. For every pair of vertices $(v_i, v_j)$, we independently include the edge $(v_i, v_j)$ with probability $1/2 + (-1)^{P_i \oplus P_j}\rho$, where $\rho = 10\varepsilon$. For a fixed partition $P$, we denote this random graph distribution by $\mathcal{G}^Y_P$.
    \item \textbf{The NO distribution ($\mathcal{G}^N$):} This is simply the standard Erd\H{o}s-R\'enyi graph distribution $\mathcal{G}_{n,1/2}$.
\end{itemize}

We will show later that the correlation clustering costs of graphs from these two distributions are $2\varepsilon n^2$-far from each other with high probability. Therefore, the lower bound reduces to the indistinguishability of graphs from the two distributions $\mcalg^Y$ and $\mcalg^N$.

We will also use the graph distribution $\mcalg_{n,r'}$ to define the perturbed distributional boolean hidden partition (PD-BHP) problem: we uniformly sample $r'$ edges over all pairs of vertices, allowing for repetition; the graph is obtained by taking the union of the sampled edges, which is a graph of $r\le r'$ edges.

\paragraph{Distribution of Queries.} We describe the distributions to the random queries the algorithm makes and their query answers. We define the distributions over query streams. Let a query stream be a sequence of $t$ triples $((u_1, v_1, e_1), \dots, (u_t, v_t, e_t))$, where each pair $(u_i, v_i)$ is sampled uniformly from $\binom{V}{2}$.
\begin{itemize}
    \item $\mathcal{D}^Y_t$: The distribution of a stream of $t$ random queries where the graph $G = (V,E)$ is sampled from $\mathcal{G}^Y$, and each $e_i$ is set to $\mathbbm{1}_{(u_i, v_i) \in E}$.
    \item $\mathcal{D}^N_t$: The distribution of a stream of $t$ random queries where the graph $G$ is sampled from $\mathcal{G}^N$.
\end{itemize}
Our lower bound applies to the mixed distribution $\mathcal{D}_t = \frac{1}{2}\mathcal{D}^Y_t + \frac{1}{2}\mathcal{D}^N_t$.

The indistinguishability of the two distributions is established on the indistinguishability between the two distributions and intermediate distributions. Specifically, fix a number $k\ge 1$ where we assume $k$ divides $t$. For every $l\in\{0,\dots, k\}$, we let $\mcald_{t,l}$ to denote the following distribution of query streams: (i) uniformly sample a partition $P$ from the $2^{n}$ possible partitions of the graph; (ii) let $G_1,\dots, G_l\leftarrow \mcalg^Y_{P}$, and let $G_{l+1},\dots, G_k\leftarrow \mcalg^N$ independently; (iii) for each $i\in[k]$ and every $j=(i-1)t/k+1,\dots, it/k$, we let $e_{j}=\mbbone_{(u_j,v_j)\in E(G_{i})}$. Intuitively, we divide the input stream into $k$ phases, where the samples of each phase is from an independently sampled graph $G$ from the YES or NO distribution. Our proof idea is as follows.

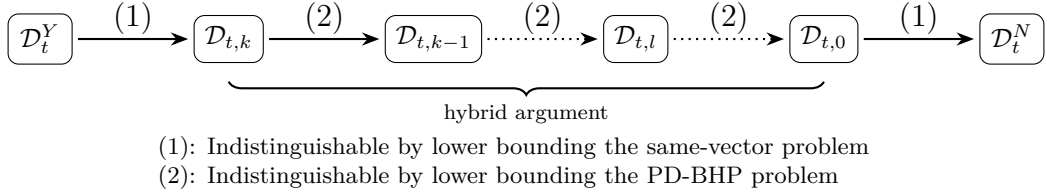
\begin{figure}[ht]
\centering
\begin{tikzpicture}[
  node distance=1.5cm and 1.6cm,
  every node/.style={font=\small},
  dist/.style={draw, rectangle, rounded corners, minimum height=1.2em, minimum width=2.2em, inner sep=4pt},
  arrow/.style={-{Stealth}, thick},
  labelstyle/.style={font=\large, midway, above=-2pt}
]

% Nodes
\node[dist] (Dy)      {\( \mcald_t^Y \)};
\node[dist, right=of Dy] (Dk) {\( \mcald_{t,k} \)};
\node[dist, right=of Dk] (Dk-1) {\( \mcald_{t,k-1} \)};
\node[dist, right=of Dk-1] (Di) {\( \mcald_{t,l} \)};
\node[dist, right=of Di] (D0) {\( \mcald_{t,0} \)};
\node[dist, right=of D0] (Dn) {\( \mcald^N_t \)};

% Arrows
\draw[arrow, shorten >=2pt, shorten <=2pt] (Dy) -- (Dk) node[labelstyle] {(1)};
\draw[arrow, shorten >=2pt, shorten <=2pt] (Dk) -- (Dk-1) node[labelstyle] {(2)};
\draw[arrow, dotted, thick, shorten >=2pt, shorten <=2pt] (Dk-1) -- (Di) node[labelstyle] {(2)};
\draw[arrow, dotted, thick, shorten >=2pt, shorten <=2pt] (Di) -- (D0)node[labelstyle] {(2)};
\draw[arrow, shorten >=2pt, shorten <=2pt] (D0) -- (Dn) node[labelstyle] {(1)};

% Brace below the hybrid chain
\draw [decorate, decoration={brace, amplitude=5pt, mirror}, thick]
  ([yshift=-7pt]Dk.south) -- ([yshift=-7pt]D0.south)
  node[midway, yshift=-12pt, font=\scriptsize] {hybrid argument};

\end{tikzpicture}

{\footnotesize
\begin{tabular}{@{}ll}
(1): Indistinguishable by lower bounding the same-vector problem \\
(2): Indistinguishable by lower bounding the PD-BHP problem
\end{tabular}
}

\caption{Structure of distributions used in the lower bound proof}
\label{fig:streaming_structure_of_proofs}
\end{figure}

Given the distributions, we show that the correlation clustering cost of graphs in $\mcald^Y$ and $\mcald^N$ are $2\varepsilon n^2$-far apart from each other with high probability.

\begin{lemma}
    \label{lem:streaming_cc_cost}
    Let $\varepsilon\in\left(\omega\left(\frac1{\sqrt n}\right),0.05\right)$. Then with probability $\ge 1-n^{-\omega(1)}$ a random graph $G$ drawn from $\mcalg^Y$ has correlation clustering cost at most $\frac{n^2}{4}-5\varepsilon n^2+ \tilde O(n)$; and a random graph $G$ drawn from $\mcalg^N$ has correlation clustering cost at least $\frac{n^2}{4}-\tilde O(n\sqrt n)$.
\end{lemma}

\begin{proof}
    We first lower bound the optimal clustering cost of a graph from $\mcalg^N=\mcalg(n,1/2)$. Observe that for every clustering, every pair of vertices $(u,v)$ contribute $1$ cost with independent probability $1/2$. Let $X$ denote the cost of an arbitrary clustering. By Chernoff bound,
    $$\Pr_{G\leftarrow \mcalg^N}[|X-\mbbe[X]|\ge n\sqrt n\log n]\le 2\cdot \exp\left(-\frac {n^3\log ^2n}{3\cdot \mbbe[X]}\right)\le 2\cdot \exp\left(-\frac43n\log^2 n\right)$$
    Since there are at most $n^n< \exp(n\log n)$ clusterings, by union bound, a random graph $G$ drawn from $\mcalg^N$ has correlation clustering cost at least $\frac{n^2}{4}-O(n^{3/2}\log n)$ with $\ge 1-\exp(-n\log n)$ probability.

    For graphs $G$ drawn from $\mcalg^Y$. By our construction to $\mcalg^Y$, each pair of vertices $(u,v)$ contribute 1 cost to $\costg(P)$ with independent probability $1/2-\rho$. Let $X$ denote $\costg(P)$. Then $\mbbe[X]=\frac{n^2}{4}-5\varepsilon n^2-O(n)$. By Chernoff bound,
    $$\Pr_{ G\leftarrow \mcalg^Y}[|X-\mbbe[X]|\ge n\log n]\le 2\cdot \exp\left(-\frac{n^2\log ^2 n}{3\cdot \mbbe[X]}\right)\le 2\cdot \exp\left(-\frac{4}{3}\log ^2 n\right).$$
\end{proof}

\subsection{The boolean hidden partition problem}
\label{sec:PD-BHP}

We analyze the 2-party one-way communication complexity of the perturbed  distributional Boolean hidden partition (PD-BHP) problem, a variant of the distributional Boolean hidden partition problem introduced in \cite{kapralov2014streaming}. In this variant, the input vector is noisy.

\paragraph{Perturbed Distributional Boolean Hidden Partition Problem (PD-BHP).} In this problem, Alice receives a vector $x \in \{0,1\}^n$, and Bob receives a graph $G=(V,E)$ and a vector $w \in \{0,1\}^r$, where $r = |E|$. Let $M \in \{0,1\}^{r \times n}$ be the edge-vertex incidence matrix of $G$. The problem is to distinguish between two cases based on the relationship between $x$, $w$, and $G$:
\begin{enumerate}
    \item \textbf{YES case:} $w = Mx + \Delta$, where $\Delta$ is a random noise vector with each entry independently set to 1 with probability $1/2 + \rho$.
    \item \textbf{NO case:} The vector $w$ is uniformly random in $\{0,1\}^r$, independent of $x$. This is equivalent to $w = Mx + \Delta'$, where $\Delta'$ is a uniformly random vector.
\end{enumerate}
Alice sends a message to Bob, who must distinguish between the two cases.

We analyze this problem under a specific distribution where Alice's input $x$ is uniformly random in $\{0,1\}^n$, and Bob's graph $G$ is sampled from $\mathcal{G}_{n,r'}$, a graph distribution defined by taking a union of $r'$ edges sampled over all pairs of vertices, allowing for repetition. The final answer is YES or NO with probability $1/2$ each. As a result, the number of edges $r$ in the resulting simple graph is exactly the number of distinct edges from the $r'$ samples, and there do not exist two identical rows in the edge-vertex incidence matrix $M$.

\sloppy
\begin{lemma}
    \label{lem:streaming_pdbhp}
    Fix three parameters $\varepsilon\in\left(\omega\left (\frac1{\sqrt n}\right),0.05\right)$, $\gamma>\omega\left(\frac1{\sqrt n}\log n\right)$ and $\frac{\varepsilon^2}{n}\le \alpha\le \min(10^{-7}\gamma^{-2},10^{-4})$. Consider an instance of the PD-BHP problem where Alice receives a uniformly random string $x\in\{0,1\}^n$, and Bob receives a graph $G\in \mathcal{G}_{n,\alpha n/\varepsilon^2}$ together with the corresponding vector $w$. Any deterministic protocol for the PD-BHP problem that uses at most $\gamma\sqrt n$ bits of communication can achieve an advantage of at most $O((\gamma+\alpha)\sqrt \alpha)$ over random guessing.
\end{lemma}
%  \sgnote{Should $G^wr$ be without replacement? Or mention if Bob has two same edges, the value of $w$ is the same.}

This lower bound is tight when both $\alpha$ and $\gamma$ are constants. In this scenario, Alice can send the first $\Theta(\sqrt n)$ bits of her input to Bob. In expectation, Bob can then learn $\Theta(1/\varepsilon^{2})$ entries of the vector $Mx$, and thus $\Theta(1/\varepsilon^{2})$ entries of the noise vector $\Delta$. Since the noise vector in the YES case and the NO case correspond to Bernoulli distributions that are $\rho$-far apart, $\Theta(1/\rho^{2})=\Theta(1/\varepsilon^{2})$ independent samples are sufficient to distinguish them with constant advantage.
While the input to Bob are exactly a block of $\Theta(n/\varepsilon^2)$ random queries, Lemma~\ref{lem:streaming_pdbhp} also directly implies an $\Omega(n/\varepsilon^2)$ random query lower bound \emph{without} memory constraints.

\begin{corollary}
    \label{cor:lb_randomquery_benchmark}
    Fix a parameter $\varepsilon\in\left(\omega\left(\frac1{\sqrt n}\right),0.05\right)$.
    Any deterministic algorithm in the random query model for distinguishing input graphs from $\mcalg^Y$ and $\mcalg^N$ with advantage $\ge 0.1$ requires $\Omega(n/\varepsilon^2)$ random queries.
\end{corollary}

\begin{proof}
    Fix $\gamma=1$. By Lemma~\ref{lem:streaming_pdbhp}, there exists a constant $\alpha\le 10^{-7}$ such that any communication protocol with $\sqrt n$ bits of communication for PD-BHP that distinguishes the YES and the NO cases can achieve an advantage of at most $\le 0.1$. The same lower bound also implies to protocols \emph{without} communication. Therefore, if there exists an algorithm (without memory constraint and) with advantage $>0.1$ that makes $\le \alpha n/\varepsilon^2$ random queries, Bob for PD-BHP could also distinguish YES and NO cases with advantage $>0.1$, which is a contradiction.
\end{proof}

This subsection focuses on proving Lemma~\ref{lem:streaming_pdbhp}. Our proof follows the same framework as \cite{gavinsky2007exponential, verbin2011streaming, kapralov2014streaming}. Compared to previous works, our problem involves a key trade-off: while the noise vector $\Delta$ makes distinguishing the distributions harder, the denser graphs make it easier. We establish a trade-off between the graph's density, the noise probability, and the communication complexity.

Specifically, Alice's message induces a partition $A_1,A_2,\dots ,A_{2^c}$ of $\{0,1\}^n$, where $c$ is the number of bits sent by Alice. Let the protocol's advantage over random guessing be at least $\tau$, where $\tau := (\gamma+\alpha)\sqrt \alpha$. Notice that at least a $1 - \tau$ fraction of all input strings $x \in \{0,1\}^n$ are contained within sets in the partition that have a size of at least $\tau2^{n-c}$. We can therefore focus our analysis on one such ``typical'' set, which we denote by $A$. Define a parameter $c' = c + \log(1/\tau)$, the size of $A$ is $|A| \ge 2^{n-c'}$.

The central idea is to show that if Alice's input $x$ is drawn uniformly from such a typical set $A$, the resulting vector $Mx$ is statistically close to uniform. If this holds, Bob is unable to distinguish the YES case (where $\Delta = Mx + w$) from the NO case (where $\Delta$ is uniformly random).

Our main technical contribution is an extension of Fourier analysis techniques from prior work \cite{gavinsky2007exponential, verbin2011streaming, kapralov2014streaming}. While previous work dealt with sparse, cycle-free graphs \cite{kapralov2014streaming}, our analysis is adapted to handle the presence of cycles.

We use the following normalization of the Fourier transform:
$$\hat f(v)=\frac1{2^n}\sum_{x\in\{0,1\}^n}f(x)(-1)^{x\cdot v}$$
We use the following bounds on the Fourier mass of $f$ contributed by coefficients of various weight:

\begin{lemma}[Lemma 6 in \cite{gavinsky2007exponential}]
\label{lem:streaming_coefficients_weight_bound}
Let $A\subseteq \{0,1\}^n$ of size at least $2^{n-c'}$ and $f$ its indicator function. Then for every $l\in[4c']$
$$\frac{2^{2n}}{|A|^2}\sum_{v:|v|=l}\hat f(v)^2\le \left(\frac{4\sqrt 2c'}{l}\right)^l$$    
\end{lemma}

Recall that $G$ is the graph Bob receives as input. The number of edges in $G$ is $r$ where $r\le \alpha n/\varepsilon^2$. $M\in\{0,1\}^{r\times n}$ denotes the incidence matrix of $G$ where $M_{eu}=1$ iff $u\in V$ is an endpoint of $e\in\binom{V}{2}$. We are interested in the distribution of $Mx+\Delta$ where $x$ is uniformly randomly selected from $A$ and $\Delta$ is a random noise vector where each entry is $1$ with probability $1/2+\rho$. For $z\in\{0,1\}^r$, let
$$p_M(z)=\frac{\sum_{x\in A}\Pr_{\Delta}[Mx+\Delta=z]}{|A|}$$
Note that $p_M$ is a function of the message $A$. We will prove that $p_M(z)$ is close to uniform by bounding the sum of squares of its Fourier coefficients for non-zero weight vectors. Specifically, let $U_r$ denote the uniform distribution over $r$-bit binary strings, then
$$\|p_M-U_r\|^2_{tvd}\le2^r\|p_M-U_r\|^2_2=2^{2r}\sum_{s\in\{0,1\}^r,s\ne 0}\widehat{p_M}(s)^2$$
By expanding the Fourier coefficients of $p_M$, we have
$$\begin{aligned}
\widehat{p_M}(s)=&\frac1{2^r}\sum_{z\in\{0,1\}^r}p_M(z)(-1)^{z\cdot s}\\
=&\frac{1}{|A|2^r}\sum_{\Delta'\in\{0,1\}^r}\Pr[\Delta=\Delta']\cdot (|\{x\in A:(Mx+\Delta')\cdot s=0\}|-|\{x\in A:(Mx+\Delta')\cdot s=1\}|)\\
=&\frac{1}{|A|2^r}\sum_{\Delta'\in\{0,1\}^r}\Pr[\Delta=\Delta']\cdot (|\{x\in A:x\cdot (M^Ts)=(\Delta')^Ts\}|-|\{x\in A:x\cdot (M^Ts)=(\Delta')^Ts+1\}|)\\
=&\frac{1}{|A|2^r}\sum_{\Delta'\in\{0,1\}^r}\Pr[\Delta=\Delta']\sum_{x}f(x)\cdot (-1)^{x\cdot (M^Ts)}\cdot (-1)^{(\Delta')^Ts}\\
=&\frac{2^n\hat f(M^Ts)}{|A|2^r}\sum_{\Delta'\in\{0,1\}^r}\Pr[\Delta=\Delta']\cdot (-1)^{(\Delta')^Ts}\\
=&\frac{2^n\hat f(M^Ts)}{|A|2^r}(-2\rho)^{|s|}
\end{aligned}$$
where the last equation is derived from the fact that the summation is exactly the expansion of the polynomial
$$((1/2+\rho)+(1/2-\rho))^{r-|s|}((1/2-\rho)-(1/2+\rho))^{|s|}=(-2\rho)^{|s|}$$

Therefore, we can rewrite the TV distance as

\begin{equation}
    \label{eq:streaming_statistical_distance}
    \begin{aligned}
    \|p_M-U_r\|^2_{tvd}\le&2^{2r}\sum_{s\in\{0,1\}^r,s\ne 0}\widehat{p_M}(s)^2\\
    =&\frac{2^{2n}}{|A|^2}\sum_{s\in\{0,1\}^r,s\ne 0}\hat f(M^Ts)^2\cdot (2\rho)^{2|s|}\\
    =&\frac{2^{2n}}{|A|^2}\sum_{v\in\{0,1\}^n}\hat f(v)^2\sum_{s\ne 0}(2\rho)^{2|s|}\mbbone_{v=M^Ts}\\
    =&\frac{2^{2n}}{|A|^2}\sum_{l\ge 0}\sum_{v\in \{0,1\}^n,|v|=l}\hat f(v)^2\sum_{s\ne 0} (2\rho)^{2|s|}\mbbone_{v=M^Ts}
    \end{aligned}
\end{equation}
where the first inequality is by Cauchy-Schwartz, the second is Parseval's equality. We will show that $\sum_{s\ne 0} (2\rho)^{2|s|}\mbbone_{v=M^Ts}$ is well-bounded for every $v$. Combined with Lemma~\ref{lem:streaming_coefficients_weight_bound}, we get our desired bound.

If we regard $s$ as an indicator of a subgraph of $G$, where each edge $e_i$ is selected if $s_i=1$, the summation represents the total weight of different vector $s$ such that the parity of degrees of vertices in the subgraph induced by $s$ is exactly the vector $v$.

\begin{lemma}
    \label{lem:streaming_counting}
    Let $M$ be the edge incidence matrix of a graph $G$ drawn from $\mathcal{G}_{n,\alpha n/\varepsilon^{2}}$. For any vector $v \in \{0,1\}^n$ with Hamming weight $l=|v|$, the following bound on the expected value holds for parameters $\varepsilon\in\left(\omega\left(\frac1{\sqrt n}\right),0.05\right)$ and $\alpha\in\left[\frac{\varepsilon^{2}}{n},10^{-4}\right]$:
    $$\mathbb{E}_{M}\left[\sum_{s\in\{0,1\}^r,s\ne 0} (2\rho)^{2|s|}\mathbbm{1}_{v=M^Ts}\right]\le \begin{cases}
        10^{10}\alpha^3&\textnormal{if }l=0\\
        \left({10^4\alpha}/{n}\right)^{l/2}&\textnormal{if }l>0
    \end{cases}$$
\end{lemma}

The above expectation bound holds because the weight $\rho=10\varepsilon$ is small enough, and because $r$ is upper bounded by $\alpha n/\varepsilon^{2}$.

\begin{proof}
    We will first bound the expectation for every fixed $r$. Notice that the graph distribution $\mcalg_{n,\alpha n/\varepsilon^{2}}$ given the number of edges $r$ fixed is uniform over the set of all $r$-edge simple graphs.

    Now fix an arbitrary $r\le \alpha n/\varepsilon^{2}$.
    We bound the expectation of the summation by counting the number of matrices $M$ such that $v=M^Ts$ for every fixed vector $s\in\{0,1\}^r$. Assume $l$ to be an even number as otherwise the summation is always $0$. Denote by $G_{M,s}$ the graph induced by $M$ and $s$. Then the set of all edges in $G_{M,s}$ can be decomposed into $\ge l/2$ walks, each walk either starts and finishes at the same (arbitrary) vertex, or starts and finishes at vertices where $v_i=1$. There are exactly $l/2$ walks of the second type, which forms the vector $v$.

We first bound the probability that $k$ random edges form a walk. Either the walk has determined start/end points, or it is a closed walk. For closed walks, this probability that $k$ selected (and distinct) edges from $M$ of a fixed order form a closed walk is at most
$$n(2.01/n)^{k-1}(2.01/n^2)=2.01^k/n^{k}$$
The factor of $n$ denotes the number of start vertices of the walk. At each step of the walk, with probability at most $\frac{n-1}{\binom{n}{2}-r}\le 2.01/n$ the next sampled edge shares an endpoint with the previous edge. The last step of the walk must go back to the start point, hence with probability $\frac1{\binom n2}\le 2.01/n^2$ one samples this edge.

Analogously, for walks with fixed start and end vertices, the probability is at most
$$(2.01/n)^{k-1}(2.01/n^2)=2.01^k/n^{k+1}$$

Therefore, for every fixed $r$, we have
%\sgnote{$w$ is input to Bob. }

    $$\begin{aligned}
    &\mbbe_{M:M\in\{0,1\}^{r\times n}}\left[\sum_{s\in \{0,1\}^r,s\ne 0}(2\rho)^{2|s|}\mbbone_{v=M^Ts}\right]\\
    \le& \sum_{t=\begin{cases}
        3&\textnormal{if }l=0\\
        l/2&\textnormal{if }l>0
    \end{cases}}^r\binom rt(2\rho)^{2t}2^tt!2.01^t/n^{t+l/2}\\
    =&\sum_{t=\begin{cases}
        3&\textnormal{if }l=0\\
        l/2&\textnormal{if }l>0
    \end{cases}}^r\binom rtt!1608^t\varepsilon^{2t}/n^{t+l/2}\\
    \le&\sum_{t=\begin{cases}
        3&\textnormal{if }l=0\\
        l/2&\textnormal{if }l>0
    \end{cases}}^rr^t1608^t\varepsilon^{2t}/n^{t+l/2}
    \end{aligned}$$
    where the last inequality is due to the fact that $\binom rtt!\le r^t$.
    The factor of $2^tt!$ accounts for ordering the $t$ edges arbitrarily ($t!$ ways) and independently partitioning it to multiple walks ($2^{t-1}$ ways).

    Notice that for $l=0$, one needs to select at least $3$ distinct edges to form $v=0^n$ since $s\ne 0$. The number of edges cannot be $2$ because two edges do not form a cycle in a simple graph. For $l>0$, one needs at least $l/2$ edges. Note that the desired expectation is a weighted sum of the expectations for every fixed $r$. We have
\allowdisplaybreaks
   \begin{align*}
        &\mbbe_{M:G_M\leftarrow \mcalg_{n,\alpha n/\varepsilon^{2}}}\left[\sum_{s\in\{0,1\}^r,s\ne 0} (2\rho)^{2|s|}\mbbone_{v=M^Ts}\right]\\
        \le&\max_{r\le \alpha n/\varepsilon^{2}}\mbbe_{M:M\in\{0,1\}^{r\times n}}\left[\sum_{s\in \{0,1\}^r,s\ne 0}(2\rho)^{2|s|}\mbbone_{v=M^Ts}\right]\\
        \le&\max_{r\le \alpha n/\varepsilon^{2}}\sum_{t=\begin{cases}
            3&\textnormal{if }l=0\\
            l/2&\textnormal{if }l>0
        \end{cases}}^rr^t1608^t\varepsilon^{2t}/n^{t+l/2}\\
        \le&\sum_{t=\begin{cases}
            3&\textnormal{if }l=0\\
            l/2&\textnormal{if }l>0
        \end{cases}}(1608\alpha)^t/n^{l/2}\\
        \le& \begin{cases}
            10^{10}\alpha^3&\textnormal{if }l=0\\
            (10^4\alpha/n)^{l/2}&\textnormal{if }l>0
        \end{cases}
    \end{align*}
    The sum is dominated by small $t$ because it is a decreasing geometric sum when $\alpha\le 10^{-4}$.

\end{proof}

With Lemma~\ref{lem:streaming_counting}, we are able to bound the statistical distance between $p_M$ and the uniform distribution.

\begin{lemma}
    \label{lem:streaming_tvd_dis}
    Fix three parameters $\varepsilon\in\left(\omega\left (\frac1{\sqrt n}\right),0.05\right)$, $\gamma>\omega\left(\frac1{\sqrt n}\log n\right)$ and $\frac{\varepsilon^2}{n}\le \alpha\le \min(10^{-7}\gamma^{-2},10^{-4})$. Let $\tau=(\gamma+\alpha)\sqrt \alpha$ and $c'= \gamma\sqrt n+\log (1/\tau)$.
    Let $A\subseteq \{0,1\}^n$ of size at least $2^{n-c'}$, and let $f:\{0,1\}^n\rightarrow \{0,1\}$ be the indicator of $A$. Let $G$ be a random graph sampled from $\mcalg_{n,\alpha n/\varepsilon^2}$. Then
    $$\mbbe_M[\|p_M-U_r\|^2_{tvd}]=O((\alpha^2+\gamma^2)\alpha)$$
\end{lemma}

\begin{proof}
    Recall that by inequality (\ref{eq:streaming_statistical_distance}),
    $$\|p_M-U_r\|^2_{tvd}\le \frac{2^{2n}}{|A|^2}\sum_{l\ge 0}\sum_{v\in\{0,1\}^n,|v|=l}\hat f(v)^2\sum_{s\ne 0}(2\rho)^{2|s|}\mbbone_{v=M^Ts}$$

    We break this summation into two parts: the part $l\in[0,4c')$ and the part $l\in[4c',n]$. For the first part, by Lemma~\ref{lem:streaming_coefficients_weight_bound} and Lemma~\ref{lem:streaming_counting}
    $$\begin{aligned}
        &\mbbe_{M}\left[\frac{2^{2n}}{|A|^2}\sum_{l=0,~l\textnormal{ is even}}^{4c'-1}\sum_{v\in\{0,1\}^n,|v|=l} \hat f(v)^2\sum_{s\ne 0}(2\rho)^{2|s|}\mbbone_{v=M^Ts}\right]\\
        \le& 10^{10}\alpha^3+\sum_{l=2,~l\textnormal{ is even}}^{4c'-1}\left(\frac{4\sqrt 2c'}{l}\right)^l\cdot (10^4\alpha /n)^{l/2}\\
        \le& O(\alpha^3)+\sum_{l=2,~l\textnormal{ is even}}(10^6\gamma^2 \alpha/l^2)^{l/2}\\
        =& O((\alpha^2+\gamma^2)\alpha)
    \end{aligned}$$
    where $\frac{2^{2n}}{|A|^2}\hat f(0^n)=1$ and $10^6\gamma^2\alpha<0.1$.

    For the part $l\in [4c',n]$, we have
    $$\begin{aligned}
        &\mbbe_M\left[\frac{2^{2n}}{|A|^2}\sum_{l\ge 4c'}\sum_{v\in\{0,1\}^n,|v|=l}\hat f(v)^2\sum_{s\ne 0}(2\rho)^{2|s|}\mbbone_{v=M^Ts}\right]\\
        \le&2^{c'}\cdot (10^4\alpha /n)^{2c'}\\
        =&(\sqrt 2\cdot10^4\alpha /n)^{2\gamma \sqrt n+O(\log n)}\\
        =&n^{-\omega(\log n)}
    \end{aligned}$$
    where we used the fact that
    $$\frac{2^{2n}}{|A|^2}\sum_{v\in\{0,1\}^n,|v|=l}\hat f(v)^2\le \frac{2^{2n}}{|A|^2}\sum_{v\in\{0,1\}^n}\hat f(v)^2=\frac{2^n}{|A|}\le 2^{c'}$$
    
    Together with the part $l\in [0,4c')$, we have
    $$\mbbe_M[\|p_M-U_r\|^2_{tvd}]=O((\alpha^2+\gamma^2)\alpha)$$
\end{proof}

We will need the following lemma, which is the Lemma 5.6 from \cite{kapralov2014streaming}.

\begin{lemma}[Lemma 5.6 of \cite{kapralov2014streaming}]
    \label{lem:streaming_lem56_kap}
    Let $(X,Y^1),(X,Y^2)$ be random variables taking values on finite sample space $\Omega=\Omega_1\times\Omega_2$. For any $x\in \Omega_1$ let $Y_x^i$, $i=1,2$ denote the conditional distribution of $Y^i$ given the event $\{X=x\}$. Then
    $$\|(X,Y^1)-(X,Y^2)\|_{tvd}=\mbbe_X[\|Y_X^1-Y_X^2\|_{tvd}]$$
\end{lemma}

We now prove Lemma~\ref{lem:streaming_pdbhp}.

\begin{proof}[Proof of Lemma~\ref{lem:streaming_pdbhp}]
Let $P(x)$ be Alice’s message function and let $Q(M, i, w)$ be Bob’s function, where $M$ is the edge incidence matrix of the input graph, $w$ is input vector of Bob, and $i$ is Alice’s message. Since we are analyzing the protocol under a fixed input distribution, we can assume $P$ and $Q$ are deterministic.

Let $D^1$ denote the distribution of $(M, P(x), w)$ under YES instances, and $D^2$ the same under NO instances. We aim to show that
$$
\|D^1 - D^2\|_{{tvd}} = O\left((\gamma + \alpha)\alpha^{1/2} \right),
$$
which implies that the protocol's distinguishing advantage is at most $O\left((\gamma + \alpha)\sqrt\alpha \right)$ on PD-BHP.

The function $P(x)$ partitions $\{0,1\}^n$ into sets $A_1, \dots, A_{2^c}$, where $c = \gamma\sqrt{n}$ is the number of bits communicated. Since there are $2^c$ such sets, at least a $1 - \tau$ fraction of $\{0,1\}^n$ is contained in those $A_i$ of size at least $\tau 2^{n - c}$, where we define $\tau := (\alpha + \gamma)\sqrt \alpha$. We call any message $i$ with $|A_i| \ge \tau 2^{n - c}$ a \emph{typical} message. Then the probability that $i = P(x)$ is not typical on a uniformly random input $x$ is at most $\tau$.

Let us write $D^1 = (M, i, p_{M,i})$, where $p_{M,i}$ is the conditional distribution of $w$ given $M$ and $i$ under the YES distribution, and write $D^2 = (M, i, U_r)$, where $U_r$ is the uniform distribution over $\{0,1\}^r$. For each $M$ and $i$, let $D^1_{(M,i)} := p_{M,i}$ and $D^2_{(M,i)} := U_r$ denote the respective conditional distributions of $w$.

Applying Lemma~\ref{lem:streaming_lem56_kap}, we get:
$$
\begin{aligned}
\|D^1 - D^2\|_{{tvd}} &= \mathbb{E}_i \left[ \mathbb{E}_M \left[ \|D^1_{(M,i)} - D^2_{(M,i)}\|_{{tvd}} \right] \right] \\
&\le \Pr[i \textnormal{ not typical}] + \mathbb{E}_M\left[ \|p_{M,i'} - U_r\|_{{tvd}} \right] \quad (i' \textnormal{ any typical message}) \\
&\le \tau + \mathbb{E}_M\left[ \|p_{M,i'} - U_r\|_{{tvd}} \right]
\end{aligned}
$$

Now fix any typical message $i'$. Then, by definition, $|A_{i'}| \ge 2^{n - c'}$ where $c' = \gamma \sqrt{n} + \log(1/\tau)$. Lemma~\ref{lem:streaming_tvd_dis} then implies
$$
\mathbb{E}_M\left[ \|p_{M,i'} - U_r\|_{\textnormal{tvd}}^2 \right] = O\left( (\gamma^2 + \alpha^2)\alpha  \right).
$$
Applying Jensen's inequality yields
$$
\mathbb{E}_M\left[ \|p_{M,i'} - U_r\|_{\textnormal{tvd}} \right] \le \sqrt{ \mathbb{E}_M\left[ \|p_{M,i'} - U_r\|_{\textnormal{tvd}}^2 \right] } = O\left((\gamma + \alpha)\sqrt\alpha \right).
$$

Combining this with the earlier bound, we conclude that
$$
\|D^1 - D^2\|_{\textnormal{tvd}} = O\left((\gamma + \alpha)\sqrt\alpha \right),
$$
as claimed.
\end{proof}

\subsection{Same vector problem}

We prove a query-space lower bound for the following problem in the random query model.

\paragraph{Same vector problem.} Let $N=\binom n2$, $p\in [0,1]$ and $k=o(N)$. Let $q< N$ divides $N$. Let $X,Y_1,\dots,Y_k$ be binary vectors drawn from $Bern(p)^N$ independently, i.e., they are length-$N$ vectors with each bit being $1$ with probability $p$. An algorithm for this problem has $k$ phases. At each phase $k'\in[k]$, the algorithm makes $q$ random queries. The return of each random query is a pair $(I_t, W_{I_t})$ where $I_t\in_R[N]$ independently. For $W_{I_t}$, (i) in the YES case, $W_{I_t}=X_{I_t}$; (ii) in the NO case, $W_{I_t}=Y_{k', I_t}$ at phase $k'$. The algorithm must distinguish between the two cases above.

When $k>1$, a trivial upper bound is to memorize $\tilde O(N/kq)$ queries from the first $kq/2$ queries, and check if there is any collision with the latter $kq/2$ samples. We show that the upper bound is tight using information theory. Formally, we obtain the following lower bound.

\begin{lemma}
    \label{lem:streaming_graph_main}
    Fix $k\ge 1$ and $q\ge 2n$. Fix $p\in[0,1]$ to be a function of $n$. For every algorithm that solves the same vector problem correctly with probability $\ge 2/3$, the algorithm must use $\Omega(N/kq)$ bits of memory. Specifically, this is true when the algorithm outputs correctly w.p. $\ge 2/3$ on an input distribution where the input is in the YES case and the NO case uniformly.
\end{lemma}

We will need to use Fano's inequality and Shearer's lemma in the proof.

\begin{proposition}[Fano's inequality]
    Let $X\rightarrow Y\rightarrow \tilde X$ be a Markov chain, and let $p_e=\Pr[\tilde X\ne X]$. Let $H(x)=-x\log x-(1-x)\log (1-x)$ denote the binary entropy function. Let $\mcalx$ be the support of $X$. Then 
    $$H(p_e)+p_e\log(|\mcalx|-1)\ge H(X|Y)$$
\end{proposition}

\begin{proposition}[Shearer's lemma \cite{chung1986some}]
    \label{prop:shearer}
    Let $X_1,\dots,X_n$ be random variables, and let $S_1,\dots,S_m\subseteq [n]$ be subsets such that each $i\in[n]$ belongs to at least $k$ sets. Then
    $$k\cdot H(X_1,\dots,X_n)\le \sum_{j=1}^m H(X_i:i\in S_j).$$
\end{proposition}

\begin{proof}[Proof to Lemma~\ref{lem:streaming_graph_main}]
    We will prove the lower bound to any algorithm over the distribution where the input is in the YES case and the NO case uniformly.

    Let $c$ denote the length of the memory.
    For every $j\in[kq]$, let $M_j\in\{0,1\}^c$ denote the memory after receiving the $j$-th random query. Specifically, $M_0=0^c$ denotes the initial memory.

    Without loss of generality assume that $M_{kq}$ is the output of the algorithm.
    Let $D$ be an indicator such that $D=1$ at the YES case, and $D=0$ at the NO case. Let $D\rightarrow M_{kq}\rightarrow M_{kq}$ be the Markov chain. By Fano's inequality and $p_e\le 1/3$, we get
    \begin{equation}
        \label{eq:streaming_graph_fanos}
        0.92\ge H(1/3)\ge H(p_e)\ge H(D|M_{kq})=1-I(D;M_{kq})
    \end{equation}

    We will decompose the mutual information $I(D;M_{kq})$ into the sum of increases of mutual information in each phase. Formally, we denote by $I_{(t)}$ the list of random variables $I_{(t-1)q+1},\dots, I_{(tq)}$ in phase $t$, and we use the same notation for other variables.
    \begin{equation}
    \label{eq:streaming_graph_potential}
    \begin{aligned}
        &I(D;M_{kq})\\
        =&\sum_{t=1}^{k}I(D;M_{tq})-I(D;M_{(t-1)q})\\
        \le &\sum_{t=1}^{k}I(D;M_{(t-1)q},I_{(t)},W_{I_{(t)}})-I(D;M_{(t-1)q})\\
        =&\sum_{t=1}^{k} I(D;I_{(t)},W_{I_{(t)}}|M_{(t-1)q})\\
        =&\sum_{t=1}^{k} I(D;W_{I_{(t)}}|I_{(t)},M_{(t-1)q})
    \end{aligned}
    \end{equation}
    by the data processing inequality.

    We will upper bound each $I(D;W_{I_{(t)}}|I_{(t)},M_{(t-1)q})$ by $I(M_{(t-1)q}; X_{I_{(t)}}|D,I_{(t)})$. For every $t,m$ and $i_{(t)}$, let $\gamma_{t,m,i_{(t)}}=\Pr[D=0|M_{t-1}=m,I_t=i_{(t)}]$. We can rewrite

    \begin{align*}
        &\sum_{t=1}^{k} I(D;W_{I_{(t)}}|I_{(t)},M_{(t-1)q})\\
        =&\sum_{t=1}^{k} H(W_{I_{(t)}}|I_{(t)},M_{(t-1)q})-H(W_{I_{(t)}}|D,I_{(t)},M_{(t-1)q})\\
        \le&\sum_{t=1}^{k} H(W_{I_{(t)}}|I_{(t)})-\mbbe_{m,i_{(t)}}[H(W_{I_{(t)}}|D,I_{(t)}=i_{(t)},M_{(t-1)q}=m)]\\
        =&\sum_{t=1}^{k} H(W_{I_{(t)}}|I_{(t)})-\mbbe_{m,i_{(t)}}[\gamma_{t,m,i}H(W_{I_{(t)}}|D=0,I_{(t)}=i_{(t)},M_{(t-1)q}=m)+\\
        &(1-\gamma_{t,m,i})H(W_{I_{(t)}}|D=1,I_{(t)}=i_{(t)},M_{(t-1)q}=m)]
    \end{align*}
    where the inequality follows from the fact that conditioning reduces entropy.
    \allowdisplaybreaks
    And we can rewrite the above as the summation of $I(M_{(t-1)q}; X_{I_{(t)}}|D,I_{(t)})$:
    \begin{align}
            \nonumber
        &\sum_{t=1}^{k}I(M_{(t-1)q}; X_{I_{(t)}}|D,I_{(t)})\\
                \nonumber
        =&\sum_{t=1}^{k} H(X_{I_{(t)}}|D,I_{(t)})-H(X_{I_{(t)}}|D,M_{(t-1)q},I_{(t)})\\
                \nonumber
        =&\sum_{t=1}^{k} H(X_{I_{(t)}}|I_{(t)})-\mbbe_{m,i_{(t)}}[H(X_{I_{(t)}}|D,M_{(t-1)q}=m,I_{(t)}=i_{(t)})]\\
                \nonumber
        =&\sum_{t=1}^{k} H(W_{I_{(t)}}|I_{(t)})-\mbbe_{m,i_{(t)}}[\gamma_{t,m,i}H(X_{I_{(t)}}|D=0,I_{(t)}=i_{(t)},M_{(t-1)q}=m)+\\
                \nonumber
        &(1-\gamma_{t,m,i})H(X_{I_{(t)}}|D=1,I_{(t)}=i_{(t)},M_{(t-1)q}=m)]\\
                \nonumber
        =&\sum_{t=1}^{k} H(W_{I_{(t)}}|I_{(t)})-\mbbe_{m,i_{(t)}}[\gamma_{t,m,i}H(W_{I_{(t)}}|D=0,I_{(t)}=i_{(t)},M_{(t-1)q}=m)+\\
        \nonumber
        &(1-\gamma_{t,m,i})H(W_{I_{(t)}}|D=1,I_{(t)}=i_{(t)},M_{(t-1)q}=m)]\\
\label{eq:streaming_graph_iden}
        \ge&\sum_{t=1}^{k} I(D;W_{I_{(t)}}|I_{(t)},M_{(t-1)q})
    \end{align}

    We are able to replace the conditional entropy of $X_{I_{(t)}}$ by $W_{I_{(t)}}$ regardless of whether the answer is YES or NO. When $D=1$ (the YES case), $W_{I_{(t)}}=X_{I_{(t)}}$ by our definition. When $D=0$ (the NO case), $W_{I_{(t)}}=Y_{t,I_{(t)}}$ is independent from previous samples, and follows the same (conditional) distribution as $X_{I_{(t)}}$.

    Let $S\leftarrow [N]^q$ denote a random subset $S\subseteq[N]$ to be the range of a random vector from $[N]^q$. We have
    \begin{align}
        \nonumber
        I(M_{(t-1)q}; X_{I_{(t)}}|D,I_{(t)})&=\frac12I(M_{(t-1)q}; X_{I_{(t)}}|D=1,I_{(t)})\\
        \nonumber
        &=\frac12\mbbe_{S\leftarrow[N]^q}\left[I(M_{(t-1)q}; X_{S}|D=1)\right]\\
        \nonumber
        &=\frac12\mbbe_{S\leftarrow[N]^q}\left[H(X_S|D=1)-H(X_S|D=1,M_{(t-1)q})\right]\\
        \nonumber
        &=\frac12\mbbe_{S\leftarrow[N]^q}\left[H(X_S|D=1)\right]-\frac12\mbbe_{S\leftarrow[N]^q}\left[H(X_S|D=1,M_{(t-1)q})\right]\\
        \label{eq:streaming_sl}
        &=\frac12\mbbe_{S\leftarrow[N]^q}\left[\sum_{i\in S}H(X_i|D=1)\right]-\frac12\mbbe_{S\leftarrow[N]^q}\left[H(X_S|D=1,M_{(t-1)q})\right]
    \end{align}
    Let $\mu=\Pr_{S\leftarrow [N]^q}[i\in S],\forall i\in[N]$.
    We have $\mu=1-(1-1/N)^q\le q/N$.
    Then, by Shearer's Lemma (Proposition~\ref{prop:shearer}),
    \begin{equation*}
        \mbbe_{S\sim I_{(t)}}\left[H(X_S|D=1,M_{(t-1)q})\right]\ge \mu\cdot H(X|D=1,M_{(t-1)q}).
    \end{equation*}
    And,
    \begin{align*}
        \mbbe_{S\sim I_{(t)}}\left[\sum_{i\in S}H(X_i|D=1)\right]&=\sum_{i\in[N]}\Pr_{S\leftarrow [N]^q}[i\in S]\cdot H(X_i|D=1)\\
        &=\mu\left(\sum_{i\in[N]}H(X_i|D=1)\right)\\
        &=\mu \cdot H(X|D=1).
    \end{align*}
    Substituting in Equation~\eqref{eq:streaming_sl}, we get 
    \begin{align*}
        I(M_{(t-1)q}; X_{I_{(t)}}|D=1,I_{(t)})&\le \frac\mu2 H(X|D=1) -\frac\mu2 H(X|D=1,M_{(t-1)q})\\
        &=\frac\mu 2 I(M_{(t-1)q},X|D=1)\\
        &\le \frac{cq}{2N}.
    \end{align*}

    Combined with inequalities (\ref{eq:streaming_graph_fanos}), (\ref{eq:streaming_graph_potential}), (\ref{eq:streaming_graph_iden}), we obtain 
    $$0.92\ge 1-I(D;M_{kq})\ge 1-\sum_{t=1}^{k} I(D;W_{I_{(t)}}|I_{(t)},M_{(t-1)q})\ge 1-\sum_{t=1}^{k}I(M_{(t-1)q}; X_{I_{(t)}}|D,I_{(t)})\ge 1-\frac{ckq}{2N}$$
    and we get a space lower bound $c\ge 0.16\frac{N}{kq}$.
\end{proof}

\subsection{Proof to the main theorem}

In this section, we prove the main theorem, which is a query-space lower bound to approximating the max cut in the random query model. We repeat the theorem here.

\begin{namedtheorem}[Theorem~\ref{thm:streaming_main}]
    Let $G=(V,E)$ be an undirected simple graph with $n$ vertices. Let $\Pi$ be any randomized algorithm that, in the random query model, approximates the correlation clustering cost of $G$ to within an additive error of $\varepsilon n^2$ with probability at least $99/100$. For this algorithm, if the worst-case query complexity is $q$ and the space used is at most $\gamma \sqrt n$ bits, then the following lower bound holds:
    $$q=\begin{cases}
        \Omega\left(\min\left(\frac{n}{\varepsilon^2\sqrt{\gamma}}, \frac{n\sqrt n}{\gamma}\right)\right)&\textnormal{if }\gamma<1\\
        \Omega\left(n/\varepsilon^2\right)&\textnormal{if }\gamma\ge 1
    \end{cases}$$
    for parameters $\varepsilon\in \left(\omega\left(\frac{1}{\sqrt n}\right),0.05\right)$ and $\gamma>\omega\left(\frac{\log n}{\sqrt n}\right)$.
\end{namedtheorem}

The lower bound is obtained through a reduction from both the same-vector problem and the PD-BHP problem to the max cut problem.
Our reduction follows a similar structure to \cite{kapralov2014streaming} despite adding a reduction to the same vector problem. This reduction is necessary in our case. We start from introducing the reduction used in \cite{kapralov2014streaming}. Then we outline how we adapt this reduction to obtain our result. After presenting a full map, we will show the formal reduction.

In \cite{kapralov2014streaming}, the main result is a $\tilde\Omega(\sqrt n)$ space lower bound for $(2-\varepsilon)$-approximating the max cut given a stream of edges of the input graph. To that end, they show that every algorithm with limited space cannot distinguish a bipartite graph from an Erd\"os-R\'enyi graph with the same expected number of edges. The lower bound is also achieved by a reduction to the D-BHP problem, which is similar to the PD-BHP problem we study in this work but (i) the noise vector is always $0^n$; (ii) the number of edges given to Bob is $\Theta(n/k)$ where $k=\Theta(\log n)$ is the number of phases. Besides, their total number of edges over all $k$ phases in the stream is only $\Theta(n)$.

Their reduction has two steps. First, they show that a stream of $k$ independently sampled random graphs of $\Theta(n/k)$ edges is close in total variation to a stream of the edges of a random graph of $\Theta(n)$ edge. This holds both when the graph is sampled from an Erd\"os-R\'enyi distribution or when the graph is a random bipartite graph given a fixed random partition. It is true because the total number of edges is $\Theta(n)$ and it is unlikely an edge will be sampled twice in two different phases. Given the first step, we only need to show that algorithms cannot distinguish between a stream of $k$ phases of random sparse bipartite graphs given and a stream of $k$ phases of random sparse Erdos-Renyi graphs.

The next step is the standard hybrid argument. For every $i\in\{0,\dots,k\}$, let $\mcald^Y_{(i)}$ denote the input distribution where the first $i$ phases are random sparse bipartite graphs given a fixed random partition, and the last $k-i$ phases are sparse Erd\"os-R\'enyi graphs. By the hybrid argument, if there is an algorithm that can distinguish the two cases with constant advantage, there also exists an index $i\in[k]$ such that the algorithm can distinguish the two distributions $\mcald^Y_{(i-1)}$ and $\mcald^Y_{(i)}$ with $\Omega(1/k)$ advantage. Such an algorithm $\Pi$ can be used to devise an algorithm for the D-BHP problem: Alice, given a partition as input, samples $i-1$ sparse graphs of $\Theta(n/k)$ edges and runs $\Pi$ on it locally. Then, Alice sends the content of the space of $\Pi$ to Bob. Bob sets the graph of the $i$-th phase to be its input graph, and the last $n-i$ phases to be Erd\"os-R\'enyi graphs. Bob runs $\Pi$ on the given space and the last $n-i+1$ phases, and outputs the output of $\Pi$. Finally, by the lower bound to D-BHP, the streaming algorithm also requires $\tilde\Omega(\sqrt n)$ space.

In our case, however, it is impossible to show that a stream of random samples of pairs of vertices to the graph is statistically indistinguishable from a stream of $k$ phases of independent random graphs of $1/k$ size. This is because each phase has $\omega(n)$ pairs of vertices in our case. Therefore, almost certainly many pairs of vertices will be sampled twice and cause a conflict (i.e., a pair of vertices is connected in one phase, but disconnected in another). To that end, we use our lower bound established above for the same-vector problem to show that every algorithm in the random query model with bounded time and space cannot distinguish a stream of consistent random queries of a fixed random graph from a stream of random queries that consists of $k$ phases of independent random graphs. Given that, we still use the hybrid argument to reduce the streaming lower bound to the communication lower bound for the PD-BHP problem studied in Section~\ref{sec:PD-BHP}. We refer readers to Figure~\ref{fig:streaming_structure_of_proofs} for an illustration to our reductions.

\begin{proof}[Proof to Theorem~\ref{thm:streaming_main}]
    The case $\gamma\ge 1$ directly follows Corollary~\ref{cor:lb_randomquery_benchmark}. Below, we focus on the case $\gamma<1$.
    
    Suppose for the sake of contradiction that there is a randomized algorithm in the random query model that makes $t$ queries to the input graph, uses $\gamma\sqrt n$ bits of memory, and approximates the correlation clustering cost to within additive error $\varepsilon n^2$ with probability at least $99/100$. Here we let
    $$t:=C\cdot \min\left(\frac{n}{\varepsilon^2\sqrt{\gamma}}, \frac{n\sqrt n}{\gamma}\right)$$
    where $C>0$ is a small enough constant. By Yao's minimax principle, there also exists a deterministic algorithm $\Pi$ with the same query complexity and space complexity such that $\Pi$ approximates the correlation clustering cost given a random input from $\mcald_t=\frac12\mcald_t^Y+\frac12\mcald_t^N$ with $\ge 99/100$ probability. Let $\Pi'$ be the algorithm that outputs $1$ if and only if the output of $\Pi$ is less than $\frac{n^2}{4}-2.5\varepsilon n^2$. By Lemma~\ref{lem:streaming_cc_cost}, $\Pi'$ is a distinguisher for input distributions $\mcald_t^Y$ and $\mcald_t^N$ with high advantage. Formally, we have
    \begin{equation}
        \label{eq:streaming_main_proof_1}
        \Pr_{S\leftarrow \mcald_t^Y}[\Pi'(S)=1]\ge 99/100-n^{-\omega(1)}
    \end{equation}
    and
    \begin{equation}
        \label{eq:streaming_main_proof_2}
        \Pr_{S\leftarrow \mcald_t^N}[\Pi'(S)=1]\le 1/100+n^{-\omega(1)}
    \end{equation}
    where $S$ denotes the input stream.

    Let $k=\frac1{\gamma\sqrt \gamma}$. Without loss of generality, we assume $k$ divides $t$.
    We obtain that
        $$\Pr_{S\leftarrow \mcald_{t,k}}[\Pi'(S)=1]\ge 0.6$$
    Otherwise, by (\ref{eq:streaming_main_proof_1}), $\Pi'$ solves the same-vector problem correctly with probability $\ge 0.695-n^{-\omega(1)}$ on the distribution $\frac12\mcald_t^Y+\frac12\mcald_{t,k}$, which contradicts  Lemma~\ref{lem:streaming_graph_main} since $\frac N t\le \gamma \sqrt n/C$. Analogously, by (\ref{eq:streaming_main_proof_2}) and Lemma~\ref{lem:streaming_graph_main}, we have
        $$\Pr_{S\leftarrow \mcald_{t,0}}[\Pi'(S)=1]\le 0.4$$

    By the standard hybrid argument, there exists a number $i\in[k]$ such that $\Pi'$ distinguishes between $\mcald_{t,i-1}$ and $\mcald_{t,i}$. Formally, if we let $p_i:=\Pr_{S\leftarrow \mcald_{t,i}}[\Pi'(S)=1]$ for every $i=0,\dots, k$, we have
    $$\max_{i\in[k]}|p_i-p_{i-1}|\ge \frac1k\sum_{i\in [k]}|p_i-p_{i-1}|\ge \frac1k|p_k-p_0|\ge (0.6-0.4)/k=\frac1{5k}$$

    Given such an algorithm $\Pi'$ and index $i$, we can obtain an algorithm for the PD-BHP problem with $\ge \frac{1}{10k}$ advantage in the communication model.

    Formally, the algorithm works as follows: when Alice receives a partition $P\in\{0,1\}^n$ as input, she independently samples a sequence of $i-1$ graphs $G_1,\dots, G_{i-1}$ drawn from $\mcalg_P^Y$ locally. Alice then runs the algorithm $\Pi'$ on $\frac{i-1}{k}t$ random queries to these graphs. After that, she sends the content of the space of $\Pi'$ to Bob. Bob receives the space of $\Pi'$, the input graph $G$ and the input vector $w$. He samples a sequence of $t/k$ triples $((u_1,v_1,e_1),\dots,(u_{t/k},v_{t/k},e_{t/k}))$ on the marginal distribution that a random graph from $\mcalg_{n,t/k}^{\textnormal{}}$ is exactly $G$, where each $e_i$ is the the corresponding bit in $w$. After this, Bob samples a sequence of $k-i$ graphs $G_{i+1},\dots,G_{k}$ from $\mcalg^N$ and correspondingly $\frac{k-i}{k}t$ random queries. Bob then simulates $\Pi'$ on these random queries with its initial content of the space set to the message of Alice. Finally, Bob outputs the output of $\Pi'$.

    To see its correctness, notice that the input distribution of Bob is exactly the distribution of subgraphs induced by the $i$-th phase of the random queries. On the YES case, Alice and Bob together simulates $\Pi'$ on the input distribution $\mcald_{t,k}$, whereas on the NO case, they simulates $\Pi'$ on $\mcald_{t,k-1}$. Hence the above algorithm distinguishes YES cases from NO cases with advantage $\ge \frac1{10k}$. However, given $\gamma<1$, by applying Lemma~\ref{lem:streaming_pdbhp} and setting $\alpha=\frac{t}{kn/\varepsilon^2} \le C\cdot \gamma$, the advantage should be $\le (\gamma+\alpha)\cdot \sqrt{\alpha}\le C'\cdot \gamma\sqrt{\gamma}$, which is smaller than $\frac{1}{10k}=0.1\gamma\sqrt\gamma$ when $C'$ is small enough. Therefore, we get a contradiction.

\end{proof}

\section{Tight query lower bound for the correlation clustering}

In this section, we prove an $\Omega(n/\varepsilon^2)$ query lower bound for the correlation clustering partition to within $\varepsilon n^2$ additive error, where the output is the partition itself, for $\varepsilon\in\left(\omega\left(\sqrt{\frac{\log n}{n}}\right),0.001\right)$. Here, algorithms have adaptive adjacency-matrix query access to the input graph. A matching upper bound was given in \cite{bressan2019correlation}.

In addition, our lower bound and a similar proof structure also applies to the max cut partition and the minimum bisection partition problems. We defer the formal statement and the proof to Appendix~\ref{appendix:strong_mcut_mbis}.

\begin{theorem}[tight query lower bound]
    \label{thm:strong_tight}
    Let $\varepsilon\in\left(\omega\left(\sqrt{\frac{\log n}{n}}\right),0.001\right)$, for every randomized adaptive algorithm $\Pi$ in the adjacency-matrix query model, if the algorithm outputs a clustering with additive error $\le \varepsilon n^2$ compared to the optimal correlation clustering cost with $> 1/100$ probability, then the worst-case query complexity of $\Pi$ is $\Omega(n/\varepsilon^2)$.
\end{theorem}

Observe that when $\varepsilon\le O\left(\sqrt{\frac1n}\right)$, $n/\varepsilon^2=\Omega(n^2)$. No sublinear algorithm is possible for approximating the correlation clustering to within an additive error smaller than $O(n\sqrt n)$.

First, we introduce our input distribution.

\begin{definition}
    \label{def:strong_input_distribution}
    Let $\rho=100\varepsilon$. Let $\mu$ be a distribution over $\{(P,G)\}$ where $P\in\{0,1\}^n$ denotes a partition, and $G=(V,E)$ denotes a graph where $|V|=n$. We call $P$ the \emph{underlying partition}.
    
    In $\mu$, $P$ follows a uniform distribution over $\{0,1\}^n$. $G$ is constructed in a way that, for each pair of vertices $v_i,v_j\in V$, $(v_i,v_j)$ is included in $E$ with independent probability
    $$\Pr[(v_i,v_j)\in E]=\begin{cases}1/2-\rho& \textnormal{ if }P_i\ne P_j\\1/2+\rho&\textnormal{ if }P_i= P_j\end{cases}$$
\end{definition}

We will prove a distributional query lower bound using $\mu$. Intuitively, the optimal solution for a random graph from $\mu$ will be dominated by the underlying partition $P$.

Our proof has two steps. First, we show that, given our input distribution, with high probability without learning $\Omega(n)$ bits of the underlying partition one cannot approximate the correlation clustering to within $\varepsilon n^2$ error. This is formalized by showing that only clusterings that are close to the underlying partition have low additive error. (See Lemma~\ref{lem:strong_hamming_dis_cc} for the formal statement.) We capture the closeness of two clusterings by comparing the largest two clusters of them. Our input distribution ensures that we can ignore the rest of small clusters.

Given the above reduction, we show that every single query can at most reveal $O(\varepsilon^2)$ information of the optimal answer. This is because the KL-divergence between $Bern(1/2-\rho)$ and $Bern(1/2+\rho)$ is $O(\rho^2)=O(\varepsilon^2)$.

The above proof steps formalize as the following lemmas.
We will use a generalized version of Fano's inequality, which transforms the error probability in the approximation to mutual information bounds.

\begin{lemma}[Generalized Fano's inequality]
    \label{lem:strong_fano}
    Let $X\rightarrow Y\rightarrow \tilde X$ be a Markov chain, where $X,\tilde X\in \mathcal X$. Let $\mcala_X\subseteq \mcalx$ be the set of approximations of $X$ where (i) $X\in \mcala_X$, for each $X\in \mcalx$; and (ii) $\tilde X\in \mcala_X$ iff $X\in \mcala_{\tilde X}$, for every $X,\tilde X\in \mcalx$. Let $p_e=\Pr[\tilde X\not\in \mcala_X]$, and $H(p_e)$ be the binary entropy function evaluated at $p_e$. Then
    \begin{equation}
        \label{eq:strong_fano_1}
        p_e\cdot \log(|\mcalx|)+(1-p_e)\cdot \sup_{X}\log(|\mcala_X|)+H(p_e)\ge H(X|Y)
    \end{equation}

    Specifically, let $\mcalx'\subseteq \mcalx$ denote the support of $X$. Then
    \begin{equation}
        \label{eq:strong_fano_2}
        p_e\cdot \log(|\mcalx'|)+(1-p_e)\cdot \sup_{X\in \mcalx}\log(|\mcala_X\cap \mcalx'|)+H(p_e)\ge H(X|Y)
    \end{equation}
\end{lemma}

In our case, $X$ represents the underlying partition, $Y$ represents the history of queries the algorithm makes to the input graph, and $\tilde X$ is the output of the algorithm. Given the above inequality, lower bounds to $H(X|Y)$ yield lower bounds to the error probability $p_e$ of the algorithm.
Notice that there are $\ge \exp(\Theta(n\log n))$ many different clusterings that approximates the underlying clustering, and inequality (\ref{eq:strong_fano_1}) does not obviously lead to a non-trivial bound to $p_e$. We resolve this problem by strengthening Fano's inequality by intersecting $\mcala_X$ with the support $\mcalx'=\{0,1\}^n$ of underlying partitions (inequality (\ref{eq:strong_fano_2})). We defer the proof to Lemma~\ref{lem:strong_fano} to Appendix~\ref{appendix:strong_fano}.

Next, we will show that only clusterings close to the underlying partition $P$ have low errors. We use $\costg(\mcalc)$ to denote the cost of a clustering $\mcalc$.

\begin{lemma}
    \label{lem:strong_hamming_dis_cc}
    Let $(P,G)\leftarrow \mu$. With probability $\ge1-n^{-\Omega(1)}$, for every clustering $\mcalc=(C_1,\dots, C_k)$, either of the following two events happen
    $$\max_{i,j\in [k]:i\ne j}|C_i \cap \{t:P_{t}=0\}|+|C_j\cap \{t:P_t=1\}|>0.9n$$
    or
    $$\costg(\mcalc)\ge \costg(P)+\varepsilon n^2$$
    where $\costg(P)$ is the cost of the clustering partitioning the vertex set into two parts by $P$.
\end{lemma}

This lemma ensures that an approximately good correlation clustering must have 2 clusters that are close to $P$.

\begin{proof}
    Without loss of generality, we assume that $C_1$ and $C_2$ are the two largest clusters among the $k$ clusters.
    
    For every $a\in \{0,1\},b\in[k]$, we still use $S_{a,b}$ to denote the set of vertices $\{v_i:P_i=a\textnormal{ and }v_i\in C_b\}$. We have
    \begin{equation}
    \label{eq:strong_hamming_dis_cc_1}
    \begin{aligned}
        &\costg(\mcalc)-\costg(P)\\
        =&\sum_{\substack{a\in\{0,1\},\\1\le b_1<b_2\le k}} \left(|\{(v_i,v_j)\in E:v_i\in S_{a,b_1},v_{j}\in S_{a,b_2}\}|-|\{(v_i,v_j)\not\in E:v_i\in S_{a,b_1},v_{j}\in S_{a,b_2}\}|\right)\\
        &-\sum_{b\in[k]} (|\{(v_i,v_j)\in E:v_i\in S_{0,b},v_{j}\in S_{1,b}\}|-|\{(v_i,v_j)\not\in E:v_i\in S_{0,b},v_{j}\in S_{1,b}\}|)
    \end{aligned}
    \end{equation}
    where the costs of pairs $(v_i,v_j)$ are dismissed if they belong to the same cluster in $\mcalc$ and $P$, or if they belong to different clusters in $\mcalc$ and $P$.

    By the Chernoff bound, with high probability the Hamming weight of $P$ falls in $[n/2-\sqrt n\log n,n/2+\sqrt n\log n]$. We fix such a partition $P$ and a clustering $\mcalc$ such that 
    \begin{equation}
        \label{eq:strong_hamming_dis_cc}
        |C_1\cap \{v_i:P_i=a\}|+|C_2\cap \{v_i:P_i=1-a\}|=|S_{a,1}|+|S_{1-a,2}|\le 0.9n
    \end{equation}
    for every $a\in\{0,1\}$. We are going to show that $\costg(\mcalc)-\costg (P)<\varepsilon n^2$ with exponentially small probability.

    For simplicity, let $s=|P|$ denote the Hamming weight of $P$. We have
    $$\begin{aligned}
        &\mbbe_{G}[\costg(\mcalc)-\costg(P)]\\
        =&\left(\binom s 2+\binom{n-s}{2}-\sum_{b=1}^k\sum_{a=0}^1 \binom {|S_{a,b}|}2\right)\cdot 2\rho+\left(\sum_{b=1}^k|S_{0,b}||S_{1,b}|\right)\cdot 2\rho\\
        =&\rho\left((n-s)^2+s^2-\sum_{b=1}^k (|S_{0,b}|-|S_{1,b}|)^2\right)
    \end{aligned}$$

    Here $(n-s)^2+s^2\ge n^2/2$. We will show that $\sum_{b=1}^k (|S_{0,b}|-|S_{1,b}|)^2\le n^2/2-0.02n^2$ to give a $0.02\rho n^2=2\varepsilon n^2$ lower bound to the expectation. 

    Suppose for the sake of contradiction that $\sum_{b=1}^k (|S_{0,b}|-|S_{1,b}|)^2> n^2/2-0.02n^2$. Then $\max_{b\in [k]}||S_{0,b}|-|S_{1,b}||> 0.48n$. Otherwise
    $$\sum_{b=1}^k (|S_{0,b}|-|S_{1,b}|)^2\le \left(\sum_{b=1}^k ||S_{0,b}|-|S_{1,b}||\right)\cdot \max_b ||S_{0,b}|-|S_{1,b}||\le n\cdot \max_b ||S_{0,b}|-|S_{1,b}||\le 0.48n^2$$

    Besides, the second-largest $||S_{0,b}|-|S_{1,b}||$ is larger than $0.46n$. This is obtained by minimizing $\frac{0.48n^2-x^2}{n-x}$, where $x\in[0,0.5n]$ denotes the largest $||S_{0,b}|-|S_{1,b}||$. Therefore, the largest two $||S_{0,b}|-|S_{1,b}||$ comes from $C_1,C_2$. Because the size of the two clusters are respectively larger than $0.48n$ and $0.46n$. In addition, $|S_{0,b}|-|S_{1,b}|$ is positive in exactly one of $C_1,C_2$. Or otherwise the Hamming weight of $P$ is larger than $0.9n$ or smaller than $0.1n$. This implies that one of the following two cases must occur
    $$|S_{0,1}|+|S_{1,2}|> 0.94n$$
    or
    $$|S_{1,1}|+|S_{0,2}|> 0.94n$$
    However, this contradicts inequality (\ref{eq:strong_hamming_dis_cc}).

    Therefore, $\sum_{b=1}^k (|S_{0,b}|-|S_{1,b}|)^2\le n^2/2-0.02n^2$. We get
    $$\mbbe_G[\costg(\mcalc)-\costg(P)]\ge 2\varepsilon n^2$$
    Observe from equation (\ref{eq:strong_hamming_dis_cc_1}) that $\costg(\mcalc)-\costg(P)$ is a sum of at most $\binom n2$ independent random variables ranging from $\{-1,0,1\}$. By applying Hoeffding's inequality, we get

    $$\begin{aligned}
        &\Pr_G[\costg(\mcalc)-\costg(P)\le \varepsilon n^2]\\
        \le&\Pr_G[|\costg(\mcalc)-\costg(P)-\mbbe[\costg(\mcalc)-\costg(P)]|\ge \varepsilon n^2]\\
        \le&2\exp(-\varepsilon^2 n^2)=\exp(-\omega(n\log n))
    \end{aligned}$$
    where we used the fact that $\varepsilon=\omega\left(\sqrt{\frac{\log n}{ n}}\right)$.
    
    Applying the union bound over all $\le n^n$ clusterings $\mcalc$, the probability that all the clusterings with large distance to $P$ has high cost is exponentially close to $1$ for $1-n^{-\Omega(1)}$ proportion of partitions $P$.
\end{proof}

Next, we prove the mutual information bound.

\begin{lemma}
    \label{lem:strong_mutual_information}
    Fix $\varepsilon\in\left(\omega\left(\sqrt{\frac{\log n}{n}}\right),0.001\right)$ and $q=10^{-8} n/\varepsilon^2$. Let $(P,G=(V,E))\leftarrow \mu$.
    Let~$\sigma_q=((v_{i_1},v_{j_1},a_1),\dots,(v_{i_q},v_{j_q},a_q))$ be a list of random variables, where each $v_{i_t},v_{j_t}\in V$ are deterministic given $v_{i_{<t}},v_{j_{<t}},a_{<t}$, for every $t\in[q]$. Each $a_t$ is the indicator of whether $(v_{i_t},v_{j_t})\in E$. Then
    $$I(P;\sigma_q)< 0.0025n.$$
\end{lemma}

\begin{proof}

    By chain rule, we have
    $$\begin{aligned}
        I(P;\sigma_q)=&\sum_{t=1}^q I(P;v_{i_t},v_{j_t},a_t|v_{i_{<t}},v_{j_{<t}},a_{<t})\\
        =&\sum_{t=1}^q I(P;a_t|v_{i_{\le t}},v_{j_{\le t}},a_{<t})
    \end{aligned}$$

    Let $\mcale_{t}$ be the event ``for every $t'<t$, $(v_{i_{t'}},v_{j_{t'}})\ne (v_{i_t},v_{j_t})$''. Then
    $$I(P;a_t|v_{i_{\le t}},v_{j_{\le t}},a_{<t})=\Pr[\mcale_t]I(P;a_t|v_{i_{\le t}},v_{j_{\le t}},a_{<t},\mcale_t)$$
    since when $\mcale_t$ is false, $a_t$ is determined.

    Let $\mcale_t'$ be the event ``$P_{i_t}=P_{j_t}$''. By the fact that $I(X;Y)=\mbbe_Y[\DKL(P_{X|Y}\|P_X)]$, we have 
    $$\begin{aligned}
        &I(P;a_t|v_{i_{\le t}},v_{j_{\le t}},a_{<t},\mcale_t)\\
        =&\Pr[\mcale_t'|v_{i_{\le t}},v_{j_{\le t}},a_{<t},\mcale_t]\mbbe_{P|v_{i_{\le t}},v_{j_{\le t}},a_{<t},\mcale_t,\mcale_t'}[\DKL(Bern(1/2-\rho)\|Bern(1/2+\eta))]+\\
        &(1-\Pr[\mcale_t'|v_{i_{\le t}},v_{j_{\le t}},a_{<t},\mcale_t])\mbbe_{P|v_{i_{\le t}},v_{j_{\le t}},a_{<t},\mcale_t,\neg\mcale_t'}[\DKL(Bern(1/2+\rho)\|Bern(1/2+\eta))]
    \end{aligned}$$
    for some $\eta\in[-\rho,\rho]$. This is because $a_t$ follows a mixed distribution of $Bern(1/2+\rho)$ and $Bern(1/2-\rho)$ conditioned on $v_{i_{\le t}},v_{j_{\le t}},a_{<t}$.
    $$\begin{aligned}
        &\DKL(Bern(1/2\pm\rho)\|Bern(1/2+\eta))\\
        =&(1/2\pm\rho)\log(\frac{1/2\pm \rho}{1/2+\eta})+(1/2\mp \rho)\log (\frac{1/2\mp \rho}{1/2-\eta})\\
        <&(1/2\pm \rho)\frac{\pm \rho-\eta}{(1/2+\eta)\ln 2}+(1/2\mp \rho)\frac{\mp\rho+\eta}{(1/2-\eta)\ln 2}\\
        =&\frac1{\ln 2}((1+\frac{-\eta\pm \rho}{1/2+\eta})(-\eta\pm \rho)+(1+\frac{\eta\mp \rho}{1/2-\eta})(\eta\mp \rho))\\
        = &\frac{1}{\ln 2}\frac{(\eta\mp\rho)^2}{1/4-\eta^2}\\
        < &25\rho^2
    \end{aligned}$$
    since $\rho=100\varepsilon<0.1$.
    Therefore,
    $$I(P;\sigma_q)= \sum_{t=1}^q\Pr[\mcale_t]I(P;a_t|v_{i_{\le t}},v_{j_{\le t}},a_{<t},\mcale_t)\le 25q\rho^2=2.5\cdot 10^{-3}n$$
\end{proof}

Now, we are ready to prove the main result of this section.

\begin{proof}[Proof to Theorem~\ref{thm:strong_tight}]
    We will prove an $\Omega(n/\varepsilon^2)$ bound for the correlation clustering partition given the distribution $\mu$. Let $(P,G)\leftarrow \mu$.
    Let $\mcala\subseteq \{0,1\}^n$ be the set of partitions with correlation clustering cost at most $\costg(P)-\varepsilon n^2$.
    For every $P\in\{-1,0,1\}^n$, let     
    $$\mcala_P:=\{P'\in \{-1,0,1\}^n:|\{i:P_i=P'_i\in\{0,1\}\}|>0.9n\textnormal{ or }|\{i:P_i=1-P'_i\}|>0.9n\}$$
    We claim that the problem of approximating the correlation clustering can be reduced from the problem of outputting an element in $\mcala_P$. Formally, given an algorithm $\Pi$ for the correlation clustering and the input $(P,G)$. We simulate $\Pi$ on $(P,G)$ and obtain a clustering $\mcalc$. Then we output the two largest clusters of $\mcalc$. This output can be written as a vector from $\{-1,0,1\}^n$, where elements from the two largest clusters are respectively marked as $0,1$, other elements are marked as $-1$. By Lemma~\ref{lem:strong_hamming_dis_cc}, with high probability every clustering of cost smaller than $\costg(P)+\varepsilon n^2$ has a large intersection with $P$, which is captured by our definition to $\mcala_P$.

    We lower bound the correlation clustering by lower bounding the problem of outputting an element from $\mcala_P$. By Yao's minimax principle, it suffices to prove a $10^{-11}n/\varepsilon^2$ average query lower bound given $\mu$ for deterministic algorithms. Fix a deterministic algorithm $\Pi$ of average query complexity $\le 10^{-11}n/\varepsilon^2$. By Markov's inequality, with probability $\ge 0.999$ the number of queries used by $\Pi$ is at most $q:=10^{-8}n/\varepsilon^2$. Thus, we can instead prove that every deterministic algorithm $\Pi'$ of worst-case query complexity $\le10^{-8}n/\varepsilon^2$ has error probability $\ge 0.991$.
    
    Let $(P,G)\leftarrow \mu$. To apply Fano's inequality, observe that for every $P',P''\in\{-1,0,1\}^n$,
    $$P'\in \mcala_{P''}\Longleftrightarrow P''\in \mcala_{P'}$$
    Besides, for every $P'\in\{-1,0,1\}^n$, let $k:=|\{i: P'_i=-1\}|$ where $k<0.1n$. We have
    $$|\mcala_{P'}\cap \{0,1\}^n|\le \sum_{i=0}^{0.1n-k}2\cdot \binom{n-k}{i}\cdot 2^k\le2^{k+1+\log n+(0.1n-k)\log\frac{(n-k)e}{0.1n-k}}\le  2^{0.477n}$$
    where the exponent takes its maximum at $k=0$, by calculating its partial derivative.
    
    So we get
    $$p_e\cdot n+(1-p_e)\cdot 0.477n+H(p_e)\ge n-I(P;\sigma_q)\ge 0.9975n$$
    where $p_e$ is the probability the algorithm outputs a wrong answer.
    Solving this inequality, we get
    $$p_e\ge 0.995>0.991$$

    Therefore, every algorithm for approximating the correlation clustering with high probability requires $\Omega(n/\varepsilon^{2})$ queries.
\end{proof}

\section{Lower bound in the general graph model}
\label{sec:general_lb}

In this section, we prove a lower bound for the correlation clustering partition problem in the general graph model, which enables both adjacency-list query access and adjacency-matrix query access to the input graph. The study of property testing on the general graph model was initiated by \cite{kaufman2004tight}. As both types of query access are allowed, proving lower bounds in the general graph model becomes hard.

\paragraph{The general graph model.} Fix an input graph $G=(V,E)$, an algorithm in the general graph model is given unit-cost query access to the following information:

\begin{itemize}
    \item Degree queries: Given a vertex $v$, return its degree $\deg(v)$.
    \item Neighbor queries: Given a vertex $v$ and an index $i$, return the $i$-th neighbor of $v$ if $i\le \deg(v)$, and return $\perp$ otherwise.
    \item Pair queries: Given two vertices $u$ and $v$, return whether they are connected by an edge in $G$.
\end{itemize}

It is important to note that we assume the neighbors of each vertex to be presented in a uniformly random order, independent of other vertices. This is a standard assumption in query lower bound analyses (e.g., see \cite{behnezhad2024approximating}), as it simplifies the analysis and removes any bias due to the order in which queries are made.

We give the following lower bound, which is the main result of this section.

\begin{theorem}[Main theorem]
    \label{thm:loose_bound_main}
    Fix a parameter $\varepsilon\in\left(\omega\left(\frac{\log n}{n}\right),10^{-6}\right)$, any (randomized) algorithm that approximates the correlation clustering partition to within additive error $\varepsilon n^2$ with probability $\ge 1/3$ costs $\Omega(n/\varepsilon)$ worst-case query complexity.
\end{theorem}

This gives a weaker lower bound compared to the $\Omega(n/\varepsilon^2)$ lower bound given in Theorem~\ref{thm:strong_tight} for the adjacency matrix model, yet it is proven in a stronger computational model. The key challenge is that, while pair queries only reveal local information about the connectivity of given pairs of vertices, degree queries and neighbor queries reveal global information of each vertex. And it becomes hard to apply similar proof ideas as in Theorem~\ref{thm:strong_tight}. In addition, we are forced to use regular graphs instead of Erd\H{o}s-R\'enyi graphs as the hard distribution unless one can show that the degree queries will not reveal too much information about the underlying partition.

We construct a well-structured input distribution on regular graphs that is hard for the correlation clustering problem.

\paragraph{Input distribution.} Let $\mu$ be the following distribution on the input graph $G$ and an \emph{underlying clustering} $\mcalc$: 
\begin{enumerate}
    \item Let $k:= 0.01/\varepsilon$ be the number of clusters. Let $\mcalc=(C_1,\dots,C_k)$ be a uniformly random partition of $V$ such that each cluster $C_\alpha$ has size exactly $100\varepsilon n$. (We ignore rounding issues.)
    \item Connect each pair of vertices from the same cluster. So the induced subgraph of each cluster is a clique of size $100\varepsilon n$.
    \item For every pair of clusters $C_\alpha,C_\beta$ in $\mcalc$, their induced bipartite graph $G[C_\alpha,C_\beta]$ is an independent and uniformly random $(\varepsilon n)$-regular bipartite graph. That is, each vertex in $C_\alpha$ is connected to exactly $0.01$ proportion of vertices in $C_\beta$, and vice versa.
\end{enumerate}

Specifically, we use $\mu_\mcalc$ to denote the distribution of $G$ when $\mcalc$ is the underlying clustering.

Intuitively, the optimal clustering is dominated by the underlying clustering $\mcalc$ with high probability, as each cluster in $\mcalc$ is a clique, and the edges between different clusters are sparse. Nevertheless, the algorithm is expected to pay $\Omega(1/\varepsilon)$ queries to find the cluster each vertex belongs to. Because the neighbors of a vertex is dominated by inter-cluster edges. Only $\Theta(\varepsilon n)$ out of its $\Theta(n)$ neighbors are belonging to the same cluster. In addition, the randomness of our input graph construction is ``local'' enough, so queries from different clusters do not have statistical dependencies. This largely simplifies our analysis.

\paragraph{Relaxation to the query model}

We assume a relaxed model of computation where the algorithm is granted extra information about the underlying clustering $\mcalc$ with each query. Even in this more powerful model, we prove the same lower bound, which simplifies our analysis while making the result more robust. Specifically, we assume the following information is given with each query:

\label{pg:loose_bound_relaxation}
\begin{itemize}
    \item Degree queries (unchanged): Given a vertex $v$, return its degree $\deg(v)$.
    \item Neighbor queries: Given a vertex $v$ and an index $i$, return the $i$-th neighbor $u$ of $v$ if $i\le \deg(v)$, and return $\perp$ otherwise. \emph{In addition, if $u\simc v$, it returns the label $\alpha$ of the cluster $C_\alpha\ni u,v$.}
    \item Pair queries: Given two vertices $u$ and $v$, return whether they are connected by an edge in $G$. \emph{In addition, if $u\simc v$, it returns the label $\alpha$ of the cluster $C_\alpha\ni u,v$.}
    \item \emph{By the end of each neighbor/pair query, if more than $0.001k$ queries have involved $u$ (including neighbor queries whose $i$-th neighbors are $u$), also return the label $\alpha$ of the cluster $C_\alpha\ni u$. Check the same for $v$ as well;}
    \item \emph{By the end of each neighbor/pair query, if there are more than $0.005n$ queries involving vertices in a cluster $C_\alpha$, also return the whole set of vertices in $C_\alpha$ and its label $\alpha$.}
\end{itemize}

We also define some notions that we will use in the proof.

\begin{definition}
    A \emph{query history} $\sigma=((Q_1,A_1),\dots,(Q_t,A_t))$ is a list of queries and answers from the input $(G,\mathcal{C})$. It includes the direct return values of degree, neighbor, and pair queries, as well as any extra information provided by the relaxed model.

    A vertex $u$ is \emph{revealed} (resp., \emph{unrevealed}) given a query history $\sigma$ if it is (resp. is not) explicit from the query history which cluster $C_\alpha$ $u$ belongs to.
    
    A \emph{direct query answer reveals a vertex $u$ given a query history $\sigma$} if $u$ is unrevealed given $\sigma$, but the query answer directly provides the cluster label of $u$. This occurs if a neighbor or pair query returns an edge $(u,v)$ where both vertices are in the same cluster ($u \sim_\mathcal{C} v$), and the query answer directly provides their cluster label. In such a case, both $u$ and $v$ are revealed by the query answer.
\end{definition}

With the above relaxation, it becomes possible to show that, given a ``good'' query history of bounded length, the probability of $u\simc v$ for every pair of vertices $u,v$ where $v$ is unrevealed is always bounded by $O(1/k)=O(\varepsilon)$.

\paragraph{Main lemmas and the proof to the main result.}

In high level, our proof has three steps. First, we show that for every fixed algorithm, it cannot reveal more than $0.001n$ vertices with high probability. Then, we show that with high probability over the input distribution, every clustering with low cost must be close to the underlying clustering $\mcalc$. We call it the \emph{closeness condition}. Lastly, we show that when the closeness condition holds, for every query history of bounded length that reveals no more than $0.001n$ vertices, every possible output clustering has a high additive error with high probability. We formalize the above proof steps in the following three lemmas.

\begin{definition}
    A query history $\sigma$ is \emph{good} if no more than $0.001n$ vertices are revealed, and is \emph{bad} otherwise.
\end{definition}

We use $t:= 10^{-9}\cdot n/\varepsilon$ to denote the desired worst-case query lower bound.

\begin{lemma}
    \label{lem:loose_bound_prune}
    Fix a parameter $\varepsilon\in\left(\omega\left(\frac{\log n}n\right),10^{-6}\right)$.
    For every deterministic algorithm with worst-case query complexity $\le t$, its query history is good with probability $\ge 1-o(1)$ given input from distribution $\mu$.
\end{lemma}

\begin{definition}
    For every clustering $\mcalc'=(C_1,\dots, C_{k'})$, we use $\mcalp(\mcalc'):=\bigcup_{i=1}^{k'} \binom {C_i}{2}$ to denote the set of unordered pairs of vertices $(u,v)$ such that $u\simcp v$.

    For every two clusterings $\mcalc$ and $\mcalc'$, their symmetric difference is defined as the following subset of pairs of vertices
    $$\Delta(\mcalp(\mcalc),\mcalp(\mcalc')):= \{(u,v): u\simc v,u\nsimcp v\}\cup \{(u,v):u\nsimc v,u\simcp v\}$$
    For simplicity, we also use $\Delta(\mcalc,\mcalc')$ to denote the symmetric difference.
    
    We say \emph{$\mcalc'$ is close to $\mcalc$} if $|\Delta(\mcalc,\mcalc')|\le 10\varepsilon n^2$. We use $\mcale_{\textnormal{close}}$ to denote the event that ``for every clustering $\mcalc'$, either $\costg(\mcalc')>\costg(\mcalc)+\varepsilon n^2$ or $\mcalc'$ is close to $\mcalc$''.
\end{definition}

\begin{lemma}
    \label{lem:loose_bound_case_i}
    $$\Pr_{(G,\mcalc)\leftarrow \mu}[\mcale_{\textnormal{close}}]\ge 1-o(1).$$
\end{lemma}

\begin{lemma}
\label{lem:loose_bound_conditionalbd}
    Fix a parameter $\varepsilon\in\left(\omega\left(\frac{\log n}n\right),10^{-6}\right)$. Let $\mcale_\sigma$ denote the event that ``$(\mcalc,G)$ is consistent with the query history $\sigma$''.
    For every good query history $\sigma$ of length $\le t$, and for every clustering $\mcalc'$, we have
    $$\Pr_{(G,\mcalc)\leftarrow \mu}[\costg(\mcalc') \le \costg(\mcalc)+\varepsilon n^2~|~\mcale_\sigma\textnormal{ and }\mcale_{\textnormal{close}}]<1/4$$
\end{lemma}

The query lower bound directly follows the above lemmas.

\begin{proof}[Proof of Theorem~\ref{thm:loose_bound_main}]
    By Yao's minimax principle, if there is a randomized algorithm that approximates correlation clustering partition with probability $\ge 1/3$ over the input distribution $\mu$, there is also a deterministic algorithm of the same query complexity that approximates it with probability $\ge 1/3$ given $\mu$. Therefore, we only need to lower bound deterministic algorithms.

    We are going to show a $0.01t$ average query lower bound to deterministic algorithms. Fix an algorithm $\Pi$ with $\le 0.01t$ average query complexity. We will show that $\Pi$ cannot approximate the correlation clustering with bounded additive error with $\ge 1/3$ probability. By Markov's inequality, the probability that $\Pi$ makes more than $t$ queries is less than $0.01$. Therefore, if we can prove that every deterministic algorithm $\Pi'$ with worst-case complexity $t$ cannot compute correctly with probability $\ge 1/3-0.01$, we also obtain a lower bound to algorithms with bounded average query complexity.

    Given Lemma~\ref{lem:loose_bound_prune}, Lemma~\ref{lem:loose_bound_case_i} and Lemma~\ref{lem:loose_bound_conditionalbd}, we bound the probability that $\Pi'$ approximates the correlation clustering. We use $\mcale_{\Pi'}$ to denote the event ``$\costg(\Pi'(G,\mcalc))\le \costg(\mcalc)+\varepsilon n^2$''. And we use $\sigma_{\Pi'}$ to denote the query history of $\Pi'$.
    $$\begin{aligned}
        &\Pr_{(G,\mcalc)\leftarrow\mu}[\mcale_{\Pi'}]\\
        \le&\Pr_{(G,\mcalc)\leftarrow\mu}[\mcale_{\Pi'}~|~\mcale_{\textnormal{close}}\wedge\sigma_{\Pi'}\textnormal{ is good}]+o(1)\\
        \le&o(1)+\sum_{\sigma}\Pr_{(G,\mcalc)\leftarrow\mu}[\sigma_{\Pi'}=\sigma~|~\mcale_{\textnormal{close}}\wedge\sigma\textnormal{ is good}]\times\\
        &\Pr_{(G,\mcalc)\leftarrow\mu}[\mcale_{\Pi'}~|~\mcale_{\textnormal{close}}\wedge (G,\mcalc)\textnormal{ is consistent with }\sigma\wedge \sigma\textnormal{ is good}]\\
        \le&o(1)+0.25\times \sum_{\sigma}\Pr_{(G,\mcalc)\leftarrow\mu}[\sigma_{\Pi'}=\sigma~|~\mcale_{\textnormal{close}}\wedge\sigma\textnormal{ is good}]\\
        \le&o(1)+0.25\\
        <&1/3-0.01
    \end{aligned}$$

    Lemma~\ref{lem:loose_bound_prune} and Lemma~\ref{lem:loose_bound_case_i} are used in the first inequality. And Lemma~\ref{lem:loose_bound_conditionalbd} is used in the third inequality.

    For the second inequality, the condition ``$(G,\mcalc)$ is consistent with $\sigma$'' is equivalent to the condition ``$\sigma_{\Pi'}=\sigma$'' because for every fixed deterministic algorithm $\Pi'$ in the query model (which is a decision tree), there does not exist two different query history $\sigma,\sigma'$ (which corresponds to two root-to-leaf execution paths of the decision tree) such that $\sigma'$ is a subsequence of $\sigma$. There always exist a query that $\sigma$ and $\sigma'$ disagree on (e.g., the query on the LCA of the two leaves in the decision tree).
\end{proof}

\subsection{Proof to Lemma~\ref{lem:loose_bound_prune}}

\paragraph{Outline of the proof.}

To show that every algorithm cannot reveal a large proportion of vertices with high probability, our idea is to show that ``at any point of the execution of the algorithm, if the query history $\sigma$ has length $\le t$ and is good, with $O(1/k)$ probability the next pair/neighbor query finds a pair of vertices $u,v$ such that $u\simc v$ and at least one of $u,v$ is unrevealed in $\sigma$''. Thus, the number of vertices revealed by direct query answers is a sum of $t=O(nk)$ random bits each with $O(1/k)$ probability to be $1$. And we apply a Chernoff-style bound to show their sum is bounded with high probability. The number of other revealed vertices is also bounded by our relaxation. Below, we provide a more detailed breakdown of this approach and address potential problems.

The probability bound on each query holds because we have ruled out the extreme cases by revealing information about the input $(G,\mcalc)$ before the probability of $u\simc v$ gets large (see the last two terms of our relaxation to the computational model at page~\pageref{pg:loose_bound_relaxation}). For example, if there is an algorithm that knows from the query history that $n-s+1$ vertices are not in the same cluster as a vertex $u$, without querying the remaining vertices one can know that they are in the same cluster. But with the extra information provided (more precisely, the fifth term of the relaxation), the whole cluster $C_\alpha\ni u$ will be revealed before making queries between $u$ and all these $n-s+1$ vertices.

To show the above probability bound, we do not need to deal with pair queries and neighbor queries separately. Observe that the probability that a neighbor query reveals a new vertex can be expressed as a weighted average of the probability that each fixed pair of vertices are belonging to the same cluster. The probability bound to neighbor queries reduces to the probability bound to pair queries.

Lastly, to bound the probability, we use a standard counting argument. Specifically, let $\mcals_{\sigma,u,v}$ be the set of pairs of graphs and clusterings $(G,\mcalc)$ such that (i) $(G,\mcalc)$ is consistent with the query history $\sigma$; and (ii) $u\simc v$. Let $\mcals_{\sigma,u,v}'$ be the set of pairs $(G,\mcalc)$ where $(u,v)\in E$ but $u\nsimc v$. We will show that $\frac{|\ssuv|}{|\ssuvp|}\le O(1/k)$ by constructing a many-to-many relation between $\ssuv$ and $\ssuvp$ such that (i) each element of $\ssuv$ is related to $\Omega(n)$ elements of $\ssuvp$; and (ii) each element of $\ssuvp$ is related to $O(n/k)$ elements from $\ssuv$.

Formally, we have the following key lemma, which will be used in proving both Lemma~\ref{lem:loose_bound_prune} and Lemma~\ref{lem:loose_bound_conditionalbd}.

\begin{lemma}
    \label{lem:loose_bound_core_prob}
    For every good query history $\sigma$ of length $\le t=10^{-9}n/\varepsilon$, and for every $u,v\in V$ such that $v$ is not revealed in $\sigma$,
    $$\Pr_{(G,\mcalc)\leftarrow\mu}[u\simc v ~|~(G,\mcalc)\textnormal{ is consistent with }\sigma]\le 334/k$$
    In addition, this is true even conditioned on knowing that $u$ is connected to $v$ in $G$:
    $$\Pr_{(G,\mcalc)\leftarrow\mu}[u\simc v ~|~(G,\mcalc)\textnormal{ is consistent with }\sigma \wedge (u,v)\in E]\le 334/k$$
\end{lemma}

The above lemma directly implies a proportional probability bound for neighbor queries.

\begin{lemma}
    \label{lem:loose_bound_core_prob_neighbor}
    Let $\mcale_{\sigma}$ denote the event that ``$(G,\mcalc)$ is consistent to the query history $\sigma$''. Let $\mcale_{u,i,\sigma}$ denote the event that ``the direct query answer of the neighbor query $N(u,i)$ reveals any vertex (either $u$ or the $i$-th neighbor of $u$) given query history $\sigma$''.
    For every good query history $\sigma$ of length $\le t=10^{-9}n/\varepsilon$, for every $u\in V$ and $i\le \deg(u)$ such that it is yet unknown from $\sigma$ which vertex $v$ the $i$-th neighbor of $u$ is,
    $$\Pr_{(G,\mcalc)\leftarrow \mu}[\mcale_{u,i,\sigma}~|~\mcale_{\sigma}]\le 444/k$$
\end{lemma}

\begin{proof}
    Let $\alpha$ denote the label of the cluster $C_\alpha\ni u$.
    There are two different cases: $u$ is revealed or unrevealed given $\sigma$.
    
    When $u$ is revealed. We may assume that at most $0.005n$ queries in $\sigma$ have involved $u$. Otherwise by the fifth term of our relaxation in page~\pageref{pg:loose_bound_relaxation}, the whole $C_\alpha$ would have been revealed, and new queries will not reveal either $u$ or the neighbor of $u$. In addition, $u$ has $\ge 0.01n$ neighbors, in which $100\varepsilon n-1$ of them belong to $C_\alpha$. Since the order of neighbors of $u$ is uniformly random, the $i$-th neighbor of $u$ belong to $C_\alpha$ with probability
    $$\Pr_{(G,\mcalc)\leftarrow \mu}[\mcale_{u,i,\sigma}~|~\mcale_{\sigma}]\le \frac{100\varepsilon n}{0.005n}=2\cdot 10^4\varepsilon=200/k $$

    When $u$ is unrevealed, at most $0.001k=10^{-5}/\varepsilon$ queries in $\sigma$ have involved $u$, by the the fourth term of our relaxation in page~\pageref{pg:loose_bound_relaxation}. By a similar analysis, the probability the neighbor query reveals the $i$-th neighbor of $u$ is at most $1.1\cdot 10^4\varepsilon=110/k$. Next, we show that the probability $u$ is revealed by this neighbor query is also small.
    $$\begin{aligned}
        &\Pr_{(G,\mcalc)\leftarrow \mu}[\textnormal{the direct query anwer of the neighbor query }N(u,i)\textnormal{ reveals }u\textnormal{ given }\sigma~|~\mcale_\sigma]\\
        \le&\sum_{v\in V}\Pr_{(G,\mcalc)\leftarrow \mu}[\textnormal{the }i\textnormal{-th neighbor of }u\textnormal{ is }v~|~\mcale_\sigma]\times\Pr_{(G,\mcalc)\leftarrow \mu}[u\simc v~|~\mcale_\sigma\wedge(u,v)\in E]\\
        \le &334/k
    \end{aligned}$$
    The first inequality is obtained by the chain rule.
    Since the order of the adjacency list of $u$ is independent of the randomness of the graph, we are able to replace the condition of the second probability by ``$(u,v)\in E$''. Therefore,
    $$\Pr_{(G,\mcalc)\leftarrow \mu}[\mcale_{u,i,\sigma}~|~\mcale_{\sigma}]\le 110/k+334/k<444/k$$
\end{proof}

Next, we prove Lemma~\ref{lem:loose_bound_core_prob}, which is the key lemma for obtaining our lower bound.

\begin{proof}[Proof to Lemma~\ref{lem:loose_bound_core_prob}]

    We only need to prove the second inequality since it implies the first inequality. Because $u\simc v$ implies $(u,v)\in E$.

    By the fifth term of our relaxation (page~\pageref{pg:loose_bound_relaxation}), we can assume that there are no more than $0.005n$ queries involving the cluster $u$ belong to, since otherwise the whole cluster will be revealed. Given that $v$ is unrevealed in $\sigma$, $v$ must belong to a different cluster to $u$ in $\mcalc$.
    
    Denote by $\mcals_{\sigma,u,v}=\{(G,\mcalc)\}$ the set of pairs of graphs and underlying clusterings such that $(G,\mcalc)$ is consistent with $\sigma$ and $u\simc v$. Denote by $\mcals_{\sigma,u,v}'$ the set of pairs where $(G,\mcalc)$ is consistent with $\sigma$, $(u,v)\in E$ but $u\nsimc v$. We have
    $$\Pr_{(G,\mcalc)\leftarrow\mu}[u\simc v ~|~(G,\mcalc)\textnormal{ is consistent with }\sigma \wedge (u,v)\in E]=\frac{|\mcals_{\sigma,u,v}|}{|\mcals'_{\sigma,u,v}|+|\ssuv|}\le \frac{|\ssuv|}{|\ssuvp|}$$
    Because $\mu$ is a uniform distribution over the pairs.

    To show that $|\mcals_{\sigma,u,v}|$ is at most $O(1/k)$ of $|\mcals'_{\sigma,u,v}|$, 
    we will construct a relation between $\mcals_{\sigma,u,v}$ and $\mcals'_{\sigma,u,v}$ such that (i) each element in $\mcals_{\sigma,u,v}$ relates to at least $\Omega(n)$ elements in $\mcals'_{\sigma,u,v}$; (ii) each element in $\mcals'_{\sigma,u,v}$ relates to at most $O(n/k)$ elements in $\mcals_{\sigma,u,v}$.

    The relation is constructed by relating every pair $(G,\mcalc)$ to another pair $(G',\mcalc')$ obtained by locally adjusting two clusters. For every $(G,\mcalc)\in \mcals_{\sigma,u,v}$, we will show that there are $\Omega(n)$ many vertices $w\nsimc v$, such that exchanging the clusters $v$ and $w$ belong to will give another pair $(G',\mcalc')\in\mcals'_{\sigma,u,v}$.

\begin{figure}[t]
  \centering
  \begin{subfigure}[t]{0.48\textwidth}
    \centering
    \includegraphics[width=\linewidth]{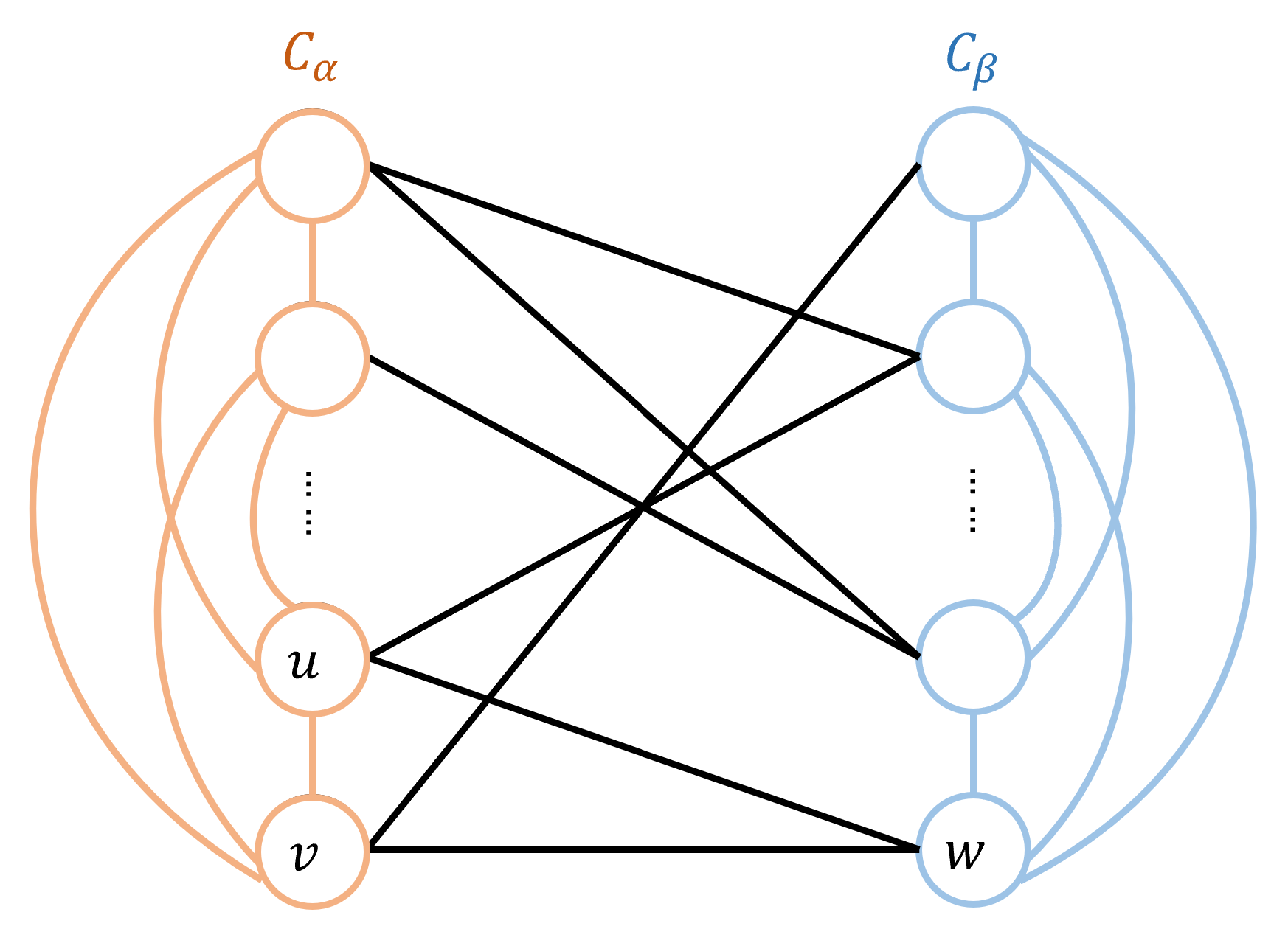} % Replace with your filename
    \caption{The subgraph over $C_\alpha$ and $C_\beta$ of $(G,\mcalc)$} % Optional subcaption
  \end{subfigure}
  \hfill
  \begin{subfigure}[t]{0.48\textwidth}
    \centering
    \includegraphics[width=\linewidth]{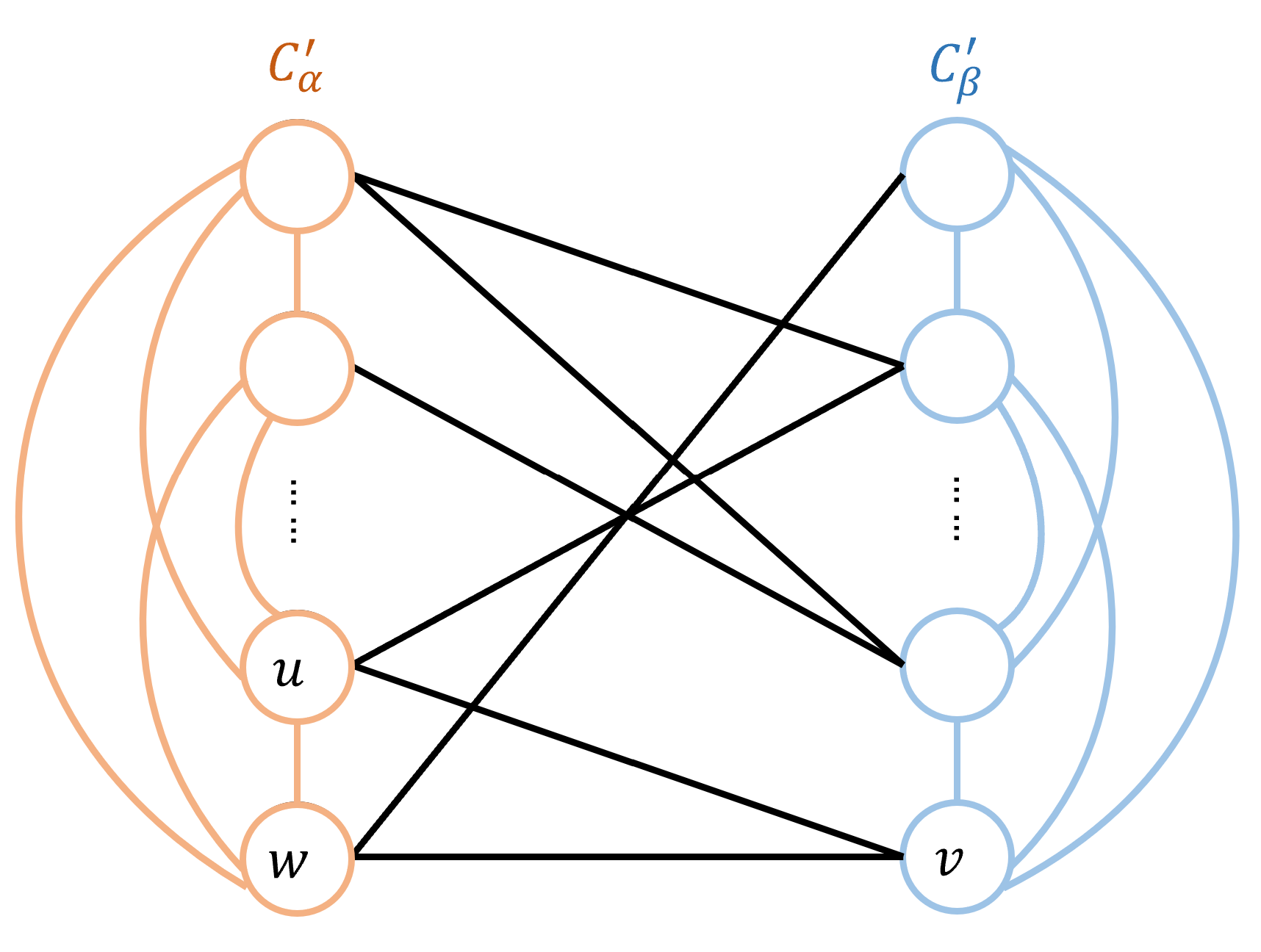} % Replace with your filename
    \caption{The subgraph over $C_\alpha'$ and $C_\beta'$ of $(G',\mcalc')$} % Optional subcaption
  \end{subfigure}
  \caption{To obtain $(G',\mcalc')$ from $(G,\mcalc)$, $v$ and $w$ exchange their clusters, and swap their edges connected to the two clusters $C_\alpha$ and $C_\beta$. The induced bipartite subgraph $G'[C_\alpha',C_\beta']$ remains a regular bipartite graph.}
    \label{figure:coupling_argument}
\end{figure}

    The following is the construction.
    Suppose $u$ and $v$ are at cluster $C_\alpha\in \mcalc$ for $\alpha\in[k]$, and $w$ is at cluster $C_\beta\in\mcalc$ for $\beta\ne \alpha\in[k]$. Such a vertex $w$ should satisfy the following conditions:
    \begin{enumerate}[label=(\roman*)]
        \item In the query history $\sigma$, there are no queries between $w$ and vertices in cluster $C_\alpha$;
        \item In the query history $\sigma$, there are no queries between $v$ and vertices in cluster $C_\beta$;
        \item $w$ is unrevealed in $\sigma$;
        \item $(u,w)\in E(G)$.
    \end{enumerate}
    For each of such vertex $w$, we construct $(G',\mcalc')$ by
    \begin{enumerate}[label=(\roman*)]
        \item Move $v$ to cluster $C_\beta$, and move $w$ to cluster $C_\alpha$;
        \item For every vertex $x\ne v\in C_\alpha$, $(x,w)\in E(G')$. Besides, $(x,v)\in E(G')$ if and only if $(x,w)\in E(G)$;
        \item For every vertex $y\ne w\in C_\beta$, $(y,v)\in E(G')$. Besides, $(y,w)\in E(G')$ if and only if $(y,v)\in E(G)$.
    \end{enumerate}
    We relate every pair $(G,\mcalc)$ with every pair $(G',\mcalc')$ constructed above. See Figure~\ref{figure:coupling_argument} for the illustrated construction.
    
    If $w$ follows the above conditions, the pair $(G',\mcalc')$, compared to $(G,\mcalc)$, will only swap the two vertices $v$ with $w$ in $C_\alpha$ and $C_\beta$, as well as the edges between $w$ and clusters $C_\alpha,C_\beta$, with the edges between $v$ and clusters $C_\alpha,C_\beta$. Since the number of edges $v$ connects to cluster $C_\beta$ is the same as the number of edges $w$ connects to cluster $C_\alpha$ (recall that the subgraphs induced by every pair of clusters $G[C_\alpha,C_\beta]$ are regular bipartite graphs), the degrees of the vertices remain unchanged. In addition, by the conditions, none of the pairs between $v,w$ and clusters $C_\alpha,C_\beta$ are queried in $\sigma$. Therefore, $(G',\mcalc')$ is consistent with the query history $\sigma$.

    Given the relation above, we show that every $(G,\mcalc)\in \ssuv$ is related to $\Omega(n)$ pairs $(G',\mcalc')\in\ssuvp$. Equivalently, we show for every $(G,\mcalc)$, one can select as many as $\Omega(n)$ vertices $w$ satisfying the above conditions. We achieve this by showing that the sum of the total number of vertices that respectively do not satisfy the above four conditions is at most $0.997n$. 

    \begin{enumerate}[label=(\roman*)]
        \item As assumed earlier, there are at most $0.005n$ queries made on cluster $C_\alpha$ in $\sigma$. Condition (i) rules out $\le 0.005n$ vertices from being $w$.
        \item Since $v$ is unrevealed, there are at most $0.001k$ queries involving $v$ in $\sigma$. It means that there are at most $0.001k$ clusters that cannot contain $w$. Condition (ii) rules out $\le 0.001k\cdot \frac nk=0.001n$ vertices from being $w$.
        \item Given that $\sigma$ is good, there are at most $0.001n$ vertices that are revealed. Condition (iii) rules out $\le 0.001n$ vertices.
        \item $u$ has more than $0.01n$ neighbors in $G$. Condition (iv) rules out $\le 0.99n$ vertices.
    \end{enumerate}

    The above sum up together to $0.997n$, which means that there are at least $0.003n$ vertices that can be selected as $w$.

    Now, we just need to show that each pair $(G',\mcalc')\in\mcals_{\sigma,u,v}'$ relates to at most $O(\frac nk)$ elements $(G,\mcalc)$ in $\mcals_{\sigma,u,v}$. Note that, by our construction, given $u,v$ and $(G',\mcalc')$, $(G,\mcalc)$ is obtained by exchanging the vertex $v$ with another vertex $w$ where $u\simcp w$. Since there are at most $\frac nk$ vertices in the cluster $C'_\alpha \ni u$, there are at most $\frac nk$ many such vertices $w$.

    Therefore, the probability is at most
    $$\frac{|\ssuv|}{|\ssuvp|}\le \frac{n/k}{0.003n}<334/k$$
\end{proof}

Given Lemma~\ref{lem:loose_bound_core_prob} and Lemma~\ref{lem:loose_bound_core_prob_neighbor}, we know that the direct query answer of each query reveals vertices with at most $444/k$ probability whenever $\sigma$ is good. We apply concentration bounds on this part of revealed vertices. And we upper bound the rest of revealed vertices by $2.4\cdot 10^{-4}n$ using our relaxation.

\begin{proof}[Proof to Lemma~\ref{lem:loose_bound_prune}]

Let random variables $X_1,\dots, X_t$ respectively denote the number of vertices revealed by the direct query answer of each query; and $X$ to be their sum.

Recall from our relaxation on page~\pageref{pg:loose_bound_relaxation} that:
\begin{itemize}
    \item We will reveal a vertex $u$ if more than $0.001k$ queries involves $u$;
    \item We will reveal the whole cluster $C_\alpha$ in $\mcalc$ if more than $0.005n$ queries were made on vertices in $C_\alpha$.
\end{itemize}
Since there are at most $t=10^{-9}\cdot n/\varepsilon$ queries; each involves at most $2$ vertices. The first case will only occur $\le 2t/0.001k=2\cdot 10^{-4}\cdot n$ times. The second case will only occur $\le 2t/0.005n=4\cdot 10^{-7}/\varepsilon$ times, revealing $2t\frac nk/0.005n=4\cdot 10^{-5}\cdot n$ vertices in total. Hence we can assume that the above cases reveal $\le 2.4\cdot 10^{-4}\cdot n$ vertices. The total number of revealed vertices given $\sigma$ is at most $X+2.4\cdot 10^{-4}n$.

$X_1,\dots,X_t$ are not independent random bits. Nevertheless, from Lemma~\ref{lem:loose_bound_core_prob} and Lemma~\ref{lem:loose_bound_core_prob_neighbor} we know that for each $X_i$, whenever the query history $\sigma_{i-1}$ of the first $i-1$ queries is good,
$$\Pr_{(G,\mcalc)\leftarrow\mu}[X_i=1\vee X_i=2~|~(G,\mcalc)\textnormal{ is consistent with }\sigma_{i-1}]\le 444/k$$
Recall that we want to upper bound the probability that the total number of revealed vertices is larger than $10^{-3}\cdot n$. By applying the Chernoff bound for the sum of adaptive but conditionally bounded random bits (Proposition~\ref{prop:adaptive}), we have
$$\Pr_{(G,\mcalc)\leftarrow\mu}[X\ge 10^{-3}n-2.4\cdot 10^{-4}n]\le 2\cdot \exp\left(-6.3\cdot 10^{-4}n\right)=o(1)$$

\end{proof}

\subsection{Proof to Lemma~\ref{lem:loose_bound_case_i}}

We prove in Lemma~\ref{lem:loose_bound_case_i} that with high probability every clustering $\mcalc'$ of low cost must be close to the underlying clustering $\mcalc$, where we capture the closeness by their symmetric difference: $|\Delta(\mcalc,\mcalc')|\le 10\varepsilon n^2$. Our proof is based on a careful discussion on different parts of their symmetric difference, and uses different proof strategies and concentration bounds on each part.

\begin{proof}[Proof to Lemma~\ref{lem:loose_bound_case_i}]
    We show that with high probability every clustering $\mcalc'$ that is not close to $\mcalc$ also has a high cost.
    
    The symmetric difference $\Delta(\mcalc,\mcalc')$ consists of three types of pairs of vertices.
    \begin{enumerate}[label=(\alph*)]
        \item $u\simc v$, $u\nsimcp v$ 
        \item $u\nsimc v$, $u\simcp v$, $(u,v)\in E$ 
        \item $u\nsimc v$, $u\simcp v$, $(u,v)\not\in E$ 
    \end{enumerate}
    We denote by $\Delta_a(\mcalc,\mcalc'), \Delta_b(\mcalc,\mcalc'), \Delta_c(\mcalc,\mcalc')$ the set of pairs of vertices of the three cases.
    
    Among these cases, only case (b) contributes $-1$ to $\costg(\mcalc')-\costg(\mcalc)$, and both case (a) and case (c) contributes $1$. Note that $(u,v)\in E$ whenever $u\simc v$, by our construction. We rewrite the difference of their costs as
    $$\costg(\mcalc')-\costg(\mcalc)=|\Delta(\mcalc,\mcalc' )|-2|\Delta_b(\mcalc,\mcalc')|$$

    Fix the underlying clustering $\mcalc$ and an arbitrary clustering $\mcalc'$.
    We will prove a concentration bound on $|\Delta_b(\mcalc,\mcalc')|$, which is a sum of edge indicators $X:=|\Delta_b(\mcalc,\mcalc')|=\sum_{(u,v):u\nsimc v, u\simcp v}X_{(u,v)}$ where $X_{(u,v)}:=\mbbone_{(u,v)\in E}$. That is, with high probability only a small proportion of pairs in the symmetric difference are connected. 

    We further decompose the symmetric difference into independent parts.
    Let $\Delta_{bc}(\mcalc,\mcalc')=\Delta_b(\mcalc,\mcalc')\cup \Delta_c(\mcalc,\mcalc')$ the set of pairs where $u\nsimc v, u\simcp v$. For every $1\le \alpha<\beta\le k$, its intersection with $C_\alpha\times C_\beta$ is a list of vertex-disjoint complete bipartite subgraphs. We use $(a_{\alpha,\beta,1},b_{\alpha,\beta,1})\dots, (a_{\alpha,\beta,k'},b_{\alpha,\beta, k'})$ to denote the number of vertices on each side of each bipartite subgraph, where $k'$ is the number of clusters in $\mcalc'$. Note that $\sum_{i=1}^{k'}a_{\alpha,\beta,i}=\sum_{i=1}^{k'}b_{\alpha,\beta,i}=100\varepsilon n$. And we suppose that the list $(a_{\alpha,\beta,i},b_{\alpha,\beta,i})$ is sorted in increasing order of $\max(a_{\alpha,\beta,i},b_{\alpha,\beta,i})$.
    
    We discuss two types of bipartite subgraphs $(a_{\alpha,\beta,i},b_{\alpha,\beta,i})$ by their size. For the part $\max(a_{\alpha,\beta,i},b_{\alpha,\beta,i})$ is small, we show that when the sum of previous indicators in the first $i$ subgraphs are concentrated, the probability that the next indicator is $1$ is always small; for the part $\max(a_{\alpha,\beta,i},b_{\alpha,\beta,i})$ is large, by the $(\varepsilon n)$-uniformity of the bipartite subgraph $G[\mcalc_{\alpha},\mcalc_\beta]$, only a small proportion of its indicators are $1$ since every vertex is connected to at most $\varepsilon n$ vertices.

    More formally, when $\max(a_{\alpha,\beta,i},b_{\alpha,\beta,i})\ge 5\varepsilon n$, by the fact that the subgraph $G[C_\alpha,C_\beta]$ is an $(\varepsilon n)$-regular bipartite graph, at most $1/5$ of its edges indicators are $1$. Let $r_{\alpha,\beta}$ denote the maximum index $i$ such that $\max(a_{\alpha,\beta,i},b_{\alpha,\beta,i})< 5\varepsilon n$. We use $s_0$ to denote the total number of indicators in $\Delta_{bc}(\mcalc,\mcalc')$ that belong to large bipartite subgraphs, and $s_{\alpha,\beta}=\sum_{i=1}^{r_{\alpha,\beta}}a_{\alpha,\beta,i}b_{\alpha,\beta,i}$ the total number of indicators in $G[C_\alpha,C_\beta]$ in small bipartite subgraphs. In addition, we use $X_{\alpha,\beta}:=\sum_{i=1}^{r_{\alpha,\beta}}\sum_{(u,v)\in \binom{C_i'}{2}\cap \mcalc_\alpha\times \mcalc_\beta}X_{(u,v)}$ to denote the sum of indicators in small bipartite subgraphs in $G[C_\alpha,C_\beta]$. Then, we have $|\Delta_{bc}(\mcalc,\mcalc')|=s_0+\sum_{\alpha,\beta} s_{\alpha,\beta}$, $X\le 0.2s_0+\sum_{\alpha,\beta}X_{\alpha,\beta}$. We show below that each $X_{\alpha,\beta}$ follows a concentration bound.

    \begin{claim}
        \label{claim:loose_bound_concentration_xab}
        For every $1\le \alpha<\beta\le k$,
        $$\Pr_{G\leftarrow \mu_\mcalc}\left[X_{\alpha,\beta}\ge 0.1 s_{\alpha,\beta}\right]\le \exp\left(-0.044s_{\alpha,\beta}\right)$$
    \end{claim}

    \begin{proof}
        For brevity, given fixed $\alpha,\beta$, we use $r$ to denote $r_{\alpha,\beta}$, $s$ to denote $s_{\alpha,\beta}$, and $(a_1,b_1),\dots,(a_{r},b_{r})$ to denote $(a_{\alpha,\beta,1},b_{\alpha,\beta,1}),\dots,(a_{\alpha,\beta,r},a_{\alpha,\beta,r})$. Since $s=\sum_{i=1}^{r} a_ib_i$ and $\max(a_i,b_i)<5\varepsilon n$ for every $i\in[r]$, $s<5\varepsilon n\cdot 100\varepsilon n=500\varepsilon^2 n^2$.

        To bound the tail probability of $X_{\alpha,\beta}$, we will use our Chernoff bound for conditionally bounded random bits (Proposition~\ref{prop:adaptive}). To that end, we bound the probability each $(u,v)\in E$ given any possible assignment to its previous edge indicators where their summation is $<0.1s$. Let $\sigma_{i-1}$ denote an assignment to all the indicators of the first $i-1$ bipartite subgraphs, and $|\sigma_{i-1}|$ the number of indicators that evaluate to $1$. Let $Y_1,\dots, Y_{a_ib_i}$ respectively denote the edge indicators of the $i$-th bipartite subgraph. For every $j\in[a_ib_i]$, we will upper bound
        $$\max_{y_1,\dots,y_{j-1}\in \{0,1\}: |\sigma_{i-1}|+\sum_{t=1}^{j-1}y_{t}<0.1s}\Pr_{G\leftarrow \mu_\mcalc}[Y_j=1|\sigma_{i-1},Y_1=y_1,\dots,Y_{j-1}=y_{j-1}]$$

        We use a similar exchange argument as in Lemma~\ref{lem:loose_bound_core_prob}. Formally, let $\sigma$ be an assignment to the first $i-1$ bipartite subgraphs and $Y_{1\dots j-1}$ such that $|\sigma|<0.1s$. Let $\mcals_{\sigma,u,v}$ be set of bipartite subgraphs $G[C_\alpha,C_\beta]$ that is consistent with $\sigma$ and $(u,v)\in E$. And let $\ssuvp$ be the set of bipartite subgraphs consistent with $\sigma$ but $(u,v)\not\in E$. We bound the probability by constructing a relation between $\ssuv$ and $\ssuvp$ where each graph in $\ssuv$ is related to at least $>38\varepsilon^2n^2$ graphs in $\ssuvp$, and each graph in $\ssuvp$ is related to at most $\le \varepsilon^2n^2$ graphs in $\ssuv$.

        The construction is simple. Given a graph $G[C_\alpha,C_\beta]$ in $\ssuv$. For every pair of vertices $(u',v')$ such that $u\ne u'$, $v\ne v'$, $(u',v')\in E$ and $(u,v'),(u',v)\not\in E$ and none of the four edges are assigned in $\sigma$, construct a graph $G'$ obtained by flipping the four edges. We relate $G$ to every such $G'$. Since none of the edges are fixed by $\sigma$ and the degrees of vertices are unchanged, $G'[C_\alpha,C_\beta]\in \ssuvp$.

        For every graph $G$, we lower bound the number of vertices $(u',v')$. There are exactly $100\varepsilon^2 n^2$ edges in $G[C_\alpha,C_\beta]$. Among them, at most $0.1s<50\varepsilon^2 n^2$ edges are fixed in $\sigma$; at most $10\varepsilon^2 n^2$ edges are incident to the $i$-th bipartite subgraph $(a_i,b_i)$, where $(u,v'),(u',v)$ may have been fixed if either of $u',v'$ belongs to the $i$-th bipartite subgraph; and at most $2\varepsilon^2n^2$ edges have an endpoint $u'/v'$ that is connected to $v/u$. The remaining $>38\varepsilon^2n^2$ edges $(u',v')$ are valid edges for our exchange argument.

        For every graph $G'$, since $(u',v),(u,v')\in E(G')$, there are at most $\varepsilon^2n^2$ such pairs of vertices $(u',v')$ given $u,v$ fixed.

        Therefore, for every $\sigma$ where $|\sigma|<0.1s$,
        $$\Pr_{G\leftarrow \mu_\mcalc}[Y_j=1|\sigma]= \frac{|\ssuv|}{|\ssuvp|+|\ssuv|}<1/39$$

        By our Chernoff bound for conditionally bounded random variables (Proposition~\ref{prop:adaptive}) and by setting $\delta =2.9$
        $$\Pr\left[\sum_{i=1}^{r}\sum_{(u,v)\in \binom{C_i'}{2}\cap \mcalc_\alpha\times \mcalc_\beta}X_{(u,v)}\ge 0.1s\right]\le 2\cdot \exp\left(-\frac{\delta^2}{2+\delta}\cdot \frac s{39}\right)<\exp(-0.044s)$$
        
    \end{proof}

    Define $B_{\alpha,\beta}=\mbbone_{X_{\alpha,\beta}\ge 0.1s_{\alpha,\beta}}$ independent random variables where each $\Pr[B_{\alpha,\beta}=1]\le \exp(-0.044s_{\alpha,\beta})$. Then $X\le 0.2s_0+\sum_{\alpha,\beta}(0.1+0.9B_{\alpha,\beta})s_{\alpha,\beta}$.
    To combine the concentration bounds on each $X_{\alpha,\beta}$, we apply the generic Chernoff bound to obtain a concentration bound on the weighted sum of $B_{\alpha,\beta}$.

    By the generic Chernoff bound (Proposition~\ref{prop:generic_chernoff}) and by setting $t=0.01$
    $$\begin{aligned}
        &\Pr\left[\sum_{\alpha,\beta}s_{\alpha,\beta}B_{\alpha,\beta}\ge 0.1\sum_{\alpha,\beta}s_{\alpha,\beta}\right]\\
        \le&\exp\left(-0.1t\sum_{\alpha,\beta}s_{\alpha,\beta}\right)\cdot \mbbe[\exp(t\sum_{\alpha,\beta}s_{\alpha,\beta}B_{\alpha,\beta})]\\
        \le&\exp\left(-0.1t\sum_{\alpha,\beta}s_{\alpha,\beta}\right)\cdot \prod_{\alpha,\beta}\mbbe[\exp(ts_{\alpha,\beta}B_{\alpha,\beta})]\\
        \le&\exp\left(-0.1t\sum_{\alpha,\beta}s_{\alpha,\beta}\right)\cdot \prod_{\alpha,\beta}(1+\exp(-0.034s_{\alpha,\beta})-\exp(-0.044s_{\alpha,\beta}))\\
        \le&\exp\left(-0.001\sum_{\alpha,\beta}s_{\alpha,\beta}+\sum_{\alpha,\beta}\exp(-0.034s_{\alpha,\beta})\right)
    \end{aligned}$$
    where we used the fact that $1+x\le e^x$.

    Notice that $\exp(-0.034s_{\alpha,\beta})<10^{-4}s_{\alpha,\beta}$ for every large enough constant $s_{\alpha,\beta}$. We consider small $s_{\alpha,\beta}$ later and obtain
    \begin{equation}
        \label{eq:loose_bound_bigsab}
        \begin{aligned}
        \Pr\left[\sum_{\alpha,\beta:s_{\alpha,\beta}\ge\log n}s_{\alpha,\beta}B_{\alpha,\beta}\ge 0.1\sum_{\alpha,\beta:s_{\alpha,\beta}\ge\log n}s_{\alpha,\beta}\right]\le \exp\left(-0.0009\sum_{\alpha,\beta: s_{\alpha,\beta}\ge\log n}s_{\alpha,\beta}\right)
    \end{aligned}
    \end{equation}

    For small $s_{\alpha,\beta}$, we reuse the proof framework from Claim~\ref{claim:loose_bound_concentration_xab} and show the following

    \begin{claim}
        
    \begin{equation}
        \label{eq:loose_bound_smallsab}
        \Pr_{G\leftarrow \mu_{\mcalc}}\left[\sum_{\alpha,\beta:s_{\alpha,\beta}<\log n}X_{\alpha,\beta}\ge 0.1\sum_{\alpha,\beta: s_{\alpha,\beta}<\log n}s_{\alpha,\beta}\right]\le \exp \left(-0.044\sum_{\alpha,\beta: s_{\alpha,\beta}<\log n}s_{\alpha,\beta}\right)
    \end{equation}
    \end{claim}

    \begin{proof}
        Following almost the same analysis as in Claim~\ref{claim:loose_bound_concentration_xab}, for every $\alpha,\beta$ where $s_{\alpha,\beta}<\log n$ and every $j\in[s_{\alpha,\beta}]$, given every $\sigma$ an assignment to the previous $j-1$ edge indicators of $X_{\alpha,\beta}$, the $j$-th edge indicator always have a $<1/39$ probability to be $1$ because much less than $50\varepsilon^2n^2$ edges are fixed in $\sigma$. By applying the same Chernoff bound for conditionally bounded random bits (Proposition~\ref{prop:adaptive}), we obtain the same concentration bound. Besides, for each edge indicator, its probability is always bounded without constraints on $\sigma$, hence the lower bound also applies to the sum of all the indicators for every $\alpha,\beta$ where $s_{\alpha,\beta}<\log n$.
    \end{proof}

    Given the above, we are ready to give a concentration bound on $X$:
    $$\Pr_{G\leftarrow \mu_\mcalc}[X\ge \max(0.4|\Delta_{bc}(\mcalc,\mcalc')|, 4.5\varepsilon n^2)]\le \exp(-\Omega(\varepsilon n^2))$$
    Recall that $$X\le 0.2s_0+\left(\sum_{\alpha,\beta:s_{\alpha,\beta}<\log n}X_{\alpha,\beta}\right)+\left(\sum_{\alpha,\beta:s_{\alpha,\beta}\ge \log n}(0.1+0.9B_{\alpha,\beta})s_{\alpha,\beta}\right)$$
    For the latter two terms, by inequalities (\ref{eq:loose_bound_bigsab}) and (\ref{eq:loose_bound_smallsab}), both exceeds $0.4\sum_{\alpha,\beta}s_{\alpha,\beta}$ with probability $\le \exp(-\Omega(\sum_{\alpha,\beta}s_{\alpha,\beta}))$.
    When any of $\sum_{\alpha,\beta:s_{\alpha,\beta}<\log n}s_{\alpha,\beta}$, $\sum_{\alpha,\beta:s_{\alpha,\beta}\ge \log n}s_{\alpha,\beta}$ is $o(\varepsilon n^2)$, the corresponding summation cannot be $\Omega(\varepsilon n^2)$ since it is upper bounded by $o(\varepsilon n^2)$ unconditionally. Therefore, we have
    $$\Pr\left[\sum_{\alpha,\beta:s_{\alpha,\beta}\ge\log n}s_{\alpha,\beta}B_{\alpha,\beta}\ge \max\left(0.01\varepsilon n^2, 0.1\sum_{\alpha,\beta:s_{\alpha,\beta}\ge\log n}s_{\alpha,\beta}\right)\right]\le \exp\left(-\Omega(\varepsilon n^2)\right)$$
    and
    $$\Pr_{G\leftarrow \mu_{\mcalc}}\left[\sum_{\alpha,\beta:s_{\alpha,\beta}<\log n}X_{\alpha,\beta}\ge \max\left(0.01\varepsilon n^2,0.1\sum_{\alpha,\beta: s_{\alpha,\beta}<\log n}s_{\alpha,\beta}\right)\right]\le \exp (-\Omega(\varepsilon n^2))$$
    In addition, $$0.2s_0+\sum_{\alpha,\beta:s_{\alpha,\beta}\ge \log n}0.1s_{\alpha,\beta}\le 0.2|\Delta_{bc}(\mcalc,\mcalc')|\le \max(0.2|\Delta_{bc}(\mcalc,\mcalc')|,2.25\varepsilon n^2)$$ Combining the above inequalities together, we obtain
    $$\begin{aligned}
        &\Pr_{G\leftarrow \mu_\mcalc}[X\ge \max(0.4|\Delta_{bc}(\mcalc,\mcalc')|, 4.5\varepsilon n^2)]\\
        \le&\Pr\left[0.2s_0+\left(\sum_{\alpha,\beta:s_{\alpha,\beta}<\log n}X_{\alpha,\beta}\right)+\left(\sum_{\alpha,\beta:s_{\alpha,\beta}\ge \log n}(0.1+0.9B_{\alpha,\beta})s_{\alpha,\beta}\right)\ge \max(0.4|\Delta_{bc}(\mcalc,\mcalc')|, 4.5\varepsilon n^2)\right]\\
        \le&\Pr\left[\left(\sum_{\alpha,\beta:s_{\alpha,\beta}<\log n}X_{\alpha,\beta}\right)+0.9\left(\sum_{\alpha,\beta:s_{\alpha,\beta}\ge \log n}B_{\alpha,\beta}s_{\alpha,\beta}\right)\ge \max(0.2|\Delta_{bc}(\mcalc,\mcalc')|, 2.25\varepsilon n^2)\right]\\
        \le&\exp(-\Omega(\varepsilon n^2))
    \end{aligned}$$

    Therefore, for every $\mcalc'$ where $|\Delta(\mcalc,\mcalc')|\ge 10\varepsilon n^2$
    $$\begin{aligned}
    &\Pr_{G\leftarrow \mu_{\mcalc}}[\costg(\mcalc')-\costg(\mcalc)\le \varepsilon n^2]\\
    =&\Pr_{G\leftarrow \mu_{\mcalc}}[|\Delta(\mcalc,\mcalc')|-2X\le \varepsilon n^2]\\
    =&\Pr_{G\leftarrow \mu_{\mcalc}}[X\ge \frac{|\Delta(\mcalc,\mcalc')|-\varepsilon n^2}2]\\
    \le&\Pr_{G\leftarrow \mu_\mcalc}[X\ge \max(0.4|\Delta(\mcalc,\mcalc')|, 4.5\varepsilon n^2)]\\
    \le&\Pr_{G\leftarrow \mu_\mcalc}[X\ge \max(0.4|\Delta_{bc}(\mcalc,\mcalc')|, 4.5\varepsilon n^2)]\\
    \le&\exp(-\Omega(\varepsilon n^2))\\
    \le&\exp(-\omega(n\log n))
    \end{aligned}$$
    since $\varepsilon =\omega(\frac{\log n}{n})$. By a union bound over all possible clusterings $\mcalc'$, with $\ge 1-o(1)$ probability every clustering that is not close to $\mcalc$ must have a high cost.

\end{proof}

\subsection{Proof to Lemma~\ref{lem:loose_bound_conditionalbd}}

Lemma~\ref{lem:loose_bound_conditionalbd} states that, given the closeness condition, given any fixed parameter $\varepsilon\in\left(\omega\left(\frac{\log n}{n}\right), 10^{-6}\right)$, fixed clustering $\mcalc'$ and any query history $\sigma$ of length $\le t=10^{-9}n/\varepsilon$ that is good, the cost of $\mcalc'$ is much larger than the cost of the underlying clustering $\mcalc$ with high probability over the input distribution $\mu$:
$$\Pr_{(G,\mcalc)\leftarrow\mu}[\costg(\mcalc')\le \costg(\mcalc)+\varepsilon n^2~|~(G,\mcalc)\textnormal{ is consistent with }\sigma]<1/4$$

The key idea is to use Lemma~\ref{lem:loose_bound_core_prob} to show that for every possible clustering $\mcalc'$, its symmetric difference to $\mcalc$ is expected to be high. Therefore it has a high cost with high probability.

\begin{proof}[Proof to Lemma~\ref{lem:loose_bound_conditionalbd}]
    We show that given a good query history of length $\le t$, every clustering $\mcalc'$ will have a high symmetric difference with $\mcalc$ with high probability.

    By the definition of a good query history, at most $0.001n$ vertices are revealed. Hence $\le \frac120.001n\cdot 100\varepsilon n=0.05\varepsilon n^2$ unordered pairs of vertices $(u,v)$ are known to have $u\simc v$.

    Recall that we use $\mcalp(\mcalc)=\bigcup_{i=1}^k\binom{C_i}{2}$ to denote the set of all unordered pairs that belong to the same cluster in $\mcalc$. And we use $\Delta(\mcalc,\mcalc')$ to denote the symmetric difference of $\mcalp(\mcalc)$ and $\mcalp(\mcalc')$.
    
    Notice that $|\mcalp(\mcalc)|= \frac{n\cdot (100\varepsilon n -1)}{2}$ by the construction. This implies $|\mcalp(\mcalc')|\le 60\varepsilon n^2$, or otherwise $\Delta(\mcalc, \mcalc')>10\varepsilon n^2$. Similarly, we also have $|\mcalp(\mcalc')|\ge 39.9\varepsilon n^2$.

    Let $X=110\varepsilon n^2-\Delta(\mcalc,\mcalc')$. By Lemma~\ref{lem:loose_bound_core_prob}, we upper bound the expectation of $X$
    $$\begin{aligned}
        &\mbbe_{(G,\mcalc)\leftarrow \mu}[X~|~(G,\mcalc)\textnormal{ is consistent with }\sigma\wedge \mcale_{\textnormal{close}}]\\
        \le &o(1)\cdot 110\varepsilon n^2+\mbbe_{(G,\mcalc)\leftarrow \mu}[X~|~(G,\mcalc)\textnormal{ is consistent with }\sigma]\\
        =&(110+o(1))\varepsilon n^2-|\mcalp(\mcalc)|-|\mcalp(\mcalc')|+2\mbbe[\mcalp(\mcalc)\cap \mcalp(\mcalc')~|~(G,\mcalc)\textnormal{ is consistent with }\sigma]\\
        \le &(60.01+o(1))\varepsilon n^2-|\mcalp(\mcalc')|+2\cdot (0.05\varepsilon n^2+\frac{334}{k}\cdot 60\varepsilon n^2)\\
        \le&(64.118+o(1))\varepsilon n^2-|\mcalp(\mcalc')|\\
        \le&(24.218+o(1))\varepsilon n^2
    \end{aligned}$$
    where we used the fact that $\varepsilon <10^{-6}$.

    By Markov's inequality,
    $$\Pr[\Delta(\mcalc,\mcalc')\le 10\varepsilon n^2]=\Pr[X\ge 100\varepsilon n^2]\le 1/4$$
    Therefore, given the closeness condition and a good query history, any fixed clustering $\mcalc'$ has a cost $\le \costg(\mcalc)+\varepsilon n^2$ with probability $\le 1/4$.
\end{proof}

\section{Acknowledgement}

The authors would like to thank anonymous reviewers for their valuable suggestions. We thank David Garc{\'\i}a Soriano for helpful comments on the query complexity literature of correlation clustering. Songhua He would also like to thank Karthik C. S. for generously providing partial research fellowship.

\newpage
\bibliographystyle{alpha}
\bibliography{ref}

\newpage

\appendix

\section{Memory-query tradoffs for the max cut and minimum bisection in the random query model}
\label{appendix:streaming_distinguishability}

For completeness, we give memory-query lower bounds to the max cut and the minimum bisection in the random query model. We show that the max cut size and the minimum bisection size of graphs from the pair of hard distributions are also $2\varepsilon n^2$ apart with high probability. Consequently, our indistinguishability result for the pair of distributions directly implies the lower bounds for both problems.

\begin{theorem}
    \label{thm:streaming_mcut_mbis}
    Let $G=(V,E)$ be an undirected simple graph with $n$ vertices. Let $\mathcal{P}$ be either of the following problems: maximum cut size, or minimal bisection size. Let $\Pi$ be any randomized algorithm that, in the random query model, approximates the size of $\mathcal{P}$ on $G$ to within an additive error of $\varepsilon n^2$ with probability at least $99/100$. For this algorithm, if the worst-case query complexity is $q$ and the space used is at most $\gamma \sqrt n$ bits, then the following lower bound holds:
    $$q=\begin{cases}
        \Omega\left(\min\left(\frac{n}{\varepsilon^2\sqrt{\gamma}}, \frac{n\sqrt n}{\gamma}\right)\right)&\textnormal{if }\gamma<1\\
        \Omega\left(\frac{n}{\varepsilon^2}\right)&\textnormal{if }\gamma\ge 1
    \end{cases}$$
    for parameters $\varepsilon\in \left(\omega\left(\frac{1}{\sqrt n}\right),0.05\right)$ and $\gamma>\omega\left(\frac{\log n}{\sqrt n}\right)$.
\end{theorem}

The lower bound statement and parameters above are the same as in Theorem~\ref{thm:streaming_main}. For the max cut and minimum bisection, the seminal work \cite{goldreich1998property} showed property testing algorithms of query complexity $O(\poly(1/\varepsilon))$ that test the max cut size and the minimum bisection size of a graph. We note that our work does not contradict their results, as we only allow random queries in this model.

We first show that the cost of (the complement of) graphs from $\mcalg^Y$ and $\mcalg^N$ are $2\varepsilon n^2$ apart with high probability.

\begin{lemma}
\label{lem:streaming_mcut_size}
    Let $\varepsilon\in \left(\omega\left(\frac{1}{\sqrt n}\right),0.05\right)$. Then with probability $\ge 1-n^{-\omega(1)}$ a random graph $G$ whose complement graph $\overline G$ is drawn from $\mcalg^Y$ has a max cut size at least $\frac{n^2}{8}+2.5\varepsilon n^2-\tilde O(n)$; and a random graph $G$ whose complement graph is drawn from $\mcalg^N$ has a max cut size at most $\frac{n^2}{8}+O(n\sqrt n)$.
\end{lemma}

\begin{proof}
    We bound $\mcalg^N$ first. Fix an arbitrary cut $P\in\{0,1\}^n$. Let $X$ denote the cut size of $P$ where $\mbbe[X]=|P|\cdot (n-|P|)/2$. By Chernoff bound, when $\mbbe[X]\ge n\sqrt n$,
    $$\Pr_{\overline G\leftarrow \mcalg^N}[|X-\mbbe[X]|\ge n\sqrt n]\le 2\cdot \exp\left(-\frac{n^3}{3\cdot \mbbe[X]}\right)\le 2\cdot \exp(-8n/3)$$
    We may ignore cuts with $\mbbe[X]<n^{3/2}$ since the cut size of those cuts will not exceed $2n\sqrt n$.
    
    By union bound over all cuts $P$, the probability that the cut size of every cut is at most $\le \mbbe[X]+n\sqrt n\le n^2/8+n\sqrt n$ is $\ge 1-\exp(-\Omega(n))$.

    Now we turn to the case $\mcalg^Y$. We show that the underlying cut $P$ itself has a large cut size with high probability. First of all, by Chernoff bound, the probability for a random partition $P$ that $$n/2-\sqrt{n}\log n\le |P|\le n/2+\sqrt n\log n$$ is at least $1-2\exp(\frac23\log^2n)=1-n^{-\omega(1)}$.
    Still let $X$ denote the cut size of $P$, where 
    $$\mbbe[X]=|P|\cdot (n-|P|)\cdot (\frac12+\rho)\ge (\frac{n^2}{4}-n\log^2n)\cdot (\frac12+\rho)=(\frac18+\frac 52\varepsilon)n^2-\tilde O(n)$$
    and
    $$\mbbe[X]\le(\frac18+\frac 52\varepsilon)n^2\le \frac14 n^2$$
    Again, by Chernoff bound,
    $$\Pr_{\overline G\leftarrow \mcalg^Y}[|X-\mbbe[X]|\ge n\log n]\le 2\cdot \exp\left(-\frac{n^2\log ^2 n}{3\cdot \mbbe[X]}\right)\le 2\cdot \exp(-\frac 43\log ^2 n)$$
    The max cut has cut size at least $n^2/8+2.5\varepsilon n^2- \tilde O(n)$ with probability $\ge 1-n^{-\omega(1)}$.
\end{proof}

\begin{lemma}
\label{lem:streaming_mbisec_size}
    Let $\varepsilon\in \left(\omega\left(\frac{1}{\sqrt n}\right),0.05\right)$. Then with probability $\ge 1-n^{-\omega(1)}$ a random graph $G$ drawn from $\mcalg^Y$ has minimum bisection size at most $\frac{n^2}{8}-2.5\varepsilon n^2+\tilde O(n\sqrt n)$; and a random graph $G$ drawn from $\mcalg^N$ has minimum bisection size at least $\frac{n^2}{8}-O(n\sqrt n)$.
\end{lemma}

\begin{proof}
    Bounding the probability that random graphs $G$ drawn from $\mcalg^N$ has minimum bisection size at least $\frac{n^2}{8}-O(n\sqrt n)$ has the same proof as in Lemma~\ref{lem:streaming_mcut_size}.

    Given a graph $G$ drawn from $\mcalg^Y$. By Chernoff bound, the probability that $$n/2-\sqrt{n}\log n\le |P|\le n/2+\sqrt n\log n$$ is at least $1-2\exp(\frac23\log^2n)=1-n^{-\omega(1)}$. Assume $P$ to be such a balanced partition, there exists a bisection $P'$ where $|P'|=n/2$ and $d_H(P,P')\le \sqrt n\log n$. Let $X$ denote the bisection size of $P'$. Then
    $$\mbbe[X]\le \frac n2\cdot \frac n2\cdot (1/2-\rho)+2\rho\cdot d_H(P,P')\cdot \frac n2\le \frac{n^2}{8}-\frac 52\varepsilon n^2+20\varepsilon n\sqrt n\log n$$
    where $d_H(P,P')\cdot \frac n2$ is the maximum number of pairs of vertices across the bisection $P'$ that connects with probability $1/2+\rho$.
    
    By Chernoff bound,
    $$\Pr_{G\leftarrow \mcalg^Y}[|X-\mbbe[X]|\ge n\log n]\le 2\cdot \exp\left(-\frac{n^2\log ^2 n}{3\cdot \mbbe[X]}\right)\le 2\cdot \exp\left(-\frac{8}{3}\log ^2 n\right)$$
\end{proof}

\begin{proof}[Proof to Theorem~\ref{thm:streaming_mcut_mbis}]
    The proof is almost the same as in Theorem~\ref{thm:streaming_main}. Given an approximation algorithm for the max cut or the minimum bisection with additive error $\le \varepsilon n^2$, by Lemma~\ref{lem:streaming_mcut_size} and Lemma~\ref{lem:streaming_mbisec_size}, one can construct a distinguisher for $\mcalg^Y$ and $\mcalg^N$ with the same space complexity and query complexity. Following the lower bound to the distinguishability between $\mcalg^Y$ and $\mcalg^N$ as in Theorem~\ref{thm:streaming_main}, we obtain the same memory-query lower bound. Analogously, the case of $\gamma\ge 1$ directly follows from Corollary~\ref{cor:lb_randomquery_benchmark}.
\end{proof}

\section{Proof to the generalized Fano's inequality}
\label{appendix:strong_fano}

\begin{namedtheorem}[Lemma~\ref{lem:strong_fano}]
    Let $X\rightarrow Y\rightarrow \tilde X$ be a Markov chain, where $X,\tilde X\in \mathcal X$. Let $\mcala_X\subseteq \mcalx$ be the set of approximations of $X$ where (i) $X\in \mcala_X$, for each $X\in \mcalx$; and (ii) $\tilde X\in \mcala_X$ iff $X\in \mcala_{\tilde X}$, for every $X,\tilde X\in \mcalx$. Let $p_e=\Pr[\tilde X\not\in \mcala_X]$, and $H(p_e)$ be the binary entropy function evaluated at $p_e$. Then
    $$p_e\cdot \log(|\mcalx|)+(1-p_e)\cdot \sup_{X}\log(|\mcala_X|)+H(p_e)\ge H(X|Y)$$

    Specifically, let $\mcalx'\subseteq \mcalx$ denote the support of $\mcalx$. Then
    $$p_e\cdot \log(|\mcalx'|)+(1-p_e)\cdot \sup_{X\in \mcalx}\log(|\mcala_X\cap \mcalx'|)+H(p_e)\ge H(X|Y)$$
\end{namedtheorem}

\begin{proof}[Proof to Lemma~\ref{lem:strong_fano}]
    Define an indicator random variable $E$ as
    $$E:=\begin{cases}
    1 & \text{if }\tilde X\not\in \mcala_X \\
    0 & \text{if }\tilde X\in \mcala_X
    \end{cases}$$

    Consider the conditional entropy $H(X,E|\tilde X)$ and expand it using the chain rule.
    $$\begin{aligned}
    &H(X,E|\tilde X)\\
    =&H(E|X,\tilde X)+H(X|\tilde X)\\
    =&0+H(X|\tilde X)\\
    =&H(Y,X,\tilde X)-H(\tilde X)-H(Y|X,\tilde X)\\
    =&(H(Y)+H(X|Y)+H(\tilde X|X,Y))-H(\tilde X)-H(Y|X,\tilde X)\\
    =&(H(Y)+H(X|Y)+H(\tilde X|Y))-H(\tilde X)-H(Y|X,\tilde X)\\
    =&H(X|Y)+H(Y|\tilde X)-H(Y|X,\tilde X)\\
    \ge &H(X|Y)
    \end{aligned}$$
    The inequality is because conditioning reduces entropy.
    
    Using the chain rule in a different way
    $$\begin{aligned}
    &H(X|Y)\\
    \le&H(X,E|\tilde X)\\=&H(X|E,\tilde X)+H(E|\tilde X)\\
    =&H(X|E=1,\tilde X)\cdot p_e+H(X|E=0,\tilde X)\cdot (1-p_e)+H(E|\tilde X)\\
    \le&\log(|\mcalx|)\cdot p_e+(1-p_e)\cdot \sup_X\log(|\mcala_X|)+H(p_e)
    \end{aligned}$$
    In the special case where $\textnormal{supp}(X)=\mcalx'$, we have $H(X|E=1,\tilde X)\le \log(|\mcalx'|)$ and $H(X|E=0,\tilde X)\le \sup_X \log(|\mcala_X\cap \mcalx'|)$. The above inequality can be rewritten as
    $$H(X|Y)\le p_e\cdot \log(|\mcalx'|)+(1-p_e)\cdot \sup_{X\in \mcalx}\log(|\mcala_X\cap \mcalx'|)+H(p_e)$$
\end{proof}

\section{Tight query lower bound to the max cut partition and the minimum bisection partition}
\label{appendix:strong_mcut_mbis}

We also give an $\Omega(n/\varepsilon^{2})$ query lower bound to the max cut partition and the minimum bisection partition problems. The proof structure is similar to Theorem~\ref{thm:strong_tight}. This lower bound is tight up to polylogarithmic factors for these two problems. By Chernoff bound, it is known that $\tilde O(n/\varepsilon^{2})$ random edges are sufficient to estimate the cut size of \emph{every} cut to within $\varepsilon n^2$ additive error.

%\sgnote{Multiplicative approximation seems to strong of a guarantee for cuts with constant number of edges. Does sparsifier through random sampling give a multiplicative + additive  error?} 

\begin{theorem}[tight query lower bounds for the max cut and the minimum bisection]
    \label{thm:strong_tight_mcut_mbis}
    Let $\varepsilon\in\left(\omega\left(\sqrt{\frac{\log n}{n}}\right),0.001\right)$, for every randomized algorithm $\Pi$ in the adjacency-matrix query model, if the algorithm outputs a cut (resp., a bisection) with additive error $\le \varepsilon n^2$ compared to the max cut size (resp., the minimum bisection size) with $> 1/100$ probability, then the worst-case query complexity of $\Pi$ is $\Omega(n/\varepsilon^2)$.
\end{theorem}

We reuse the same input distribution defined in Definition~\ref{def:strong_input_distribution}. By Lemma~\ref{lem:strong_mutual_information}, we know that without paying $\Omega(n/\varepsilon^{2})$ queries, one cannot learn $\ge 0.0025n$ information about the underlying partition $P$. We prove below analogs to Lemma~\ref{lem:strong_hamming_dis_cc}, that only cuts/bisections close to the underlying partition $P$ have low additive error.

\begin{lemma}
    \label{lem:strong_hamming_dis_mcut}
    Let $(P,G)\leftarrow \mu$. With probability $\ge 1-n^{-\Omega(1)}$ every partition $P'\in\{0,1\}^n$ such that $d_H(P,P')\in[0.1n,0.9n]$ has
    $$\cutg(P')> \cutg(P)+1.01\varepsilon n^2$$
\end{lemma}

We will use $\mu$ as the input distribution of the minimum bisection, and the complement of graphs from $\mu$ as the input distribution of the max cut. Therefore, this lemma works for both problems.

\begin{proof}
    By Chernoff bound, the Hamming weight of $P$ falls in $[n/2-\sqrt n\log n,n/2+\sqrt n\log n]$ with probability $\ge 1-n^{-\Omega(1)}$. We fix such a $P$ with bounded Hamming weight.

    Fix an arbitrary partition $P'$ where $d_H(P,P')\in[0.5n,0.9n]$. We may ignore those partitions $P'$ for which $d_H(P,P')< 0.5n$ because $P'$ and $P'\oplus 1^n$ denote the same partition, and $d_H(P,P')=n-d_H(P,P'\oplus 1^n)$.

    For every $a,b\in\{0,1\}$, we denote $S_{a,b}$ as the subset of vertices $\{v_i:P_i=a,P'_i=b\}$. Let $s=|S_{0,0}|+|S_{0,1}|$ which is exactly $n$ minus the Hamming weight of $P$. We have
    $$\begin{aligned}
        &\mbbe_{G}[\cutg(P')-\cutg(P)]\\
        =&\mbbe_{G}\left[\sum_{a=0}^1|\{(v_i,v_j)\in E: v_i\in S_{a,0}, v_j\in S_{a,1}\}|-\sum_{b=0}^1|\{(v_i,v_j)\in E:v_i\in S_{0,b},v_j\in S_{1,b}\}|\right]\\
        =&(|S_{0,0}||S_{0,1}|+|S_{1,0}||S_{1,1}|)(1/2+\rho)-(|S_{0,0}||S_{1,0}|+|S_{0,1}||S_{1,1}|)(1/2-\rho)\\
        =&\frac12(|S_{1,1}|-|S_{0,0}|)(|S_{1,0}|-|S_{0,1}|)+\rho(|S_{0,0}|+|S_{1,1}|)(|S_{1,0}|+|S_{0,1}|)\\
        =&\frac12(|S_{1,1}|-|S_{0,0}|)(n-2s+|S_{0,0}|-|S_{1,1}|)+\rho(|S_{0,0}|+|S_{1,1}|)(n-|S_{0,0}|-|S_{1,1}|)\\
        \ge& -n\sqrt n\log n+0.09\rho n^2\\
        =&9\varepsilon n^2+\tilde O(n\sqrt n)
    \end{aligned}$$
    The inequality is derived from the fact that $n-2s\in[-2\sqrt n\log n, 2\sqrt n\log n]$ and $|S_{0,0}|+|S_{1,1}|=n-d_H(P,P')\in [0.1n,0.5n]$.

    Since the edge set of the graph $G$, given a fixed $P$, follows a product Bernoulli distribution on its edges, we are able to apply Hoeffding's inequality on the sum of edges:
    $$\Pr_{G}[\cutg(P')-\cutg(P)\le  1.01\varepsilon n^2]\le2\exp\left(-\frac{2\cdot 7.99^2\varepsilon^2 n^4}{\binom n2}\right)=  \exp(-\omega(n))$$
    where we used the fact that $\varepsilon=\omega(1/\sqrt n)$.
    
    Applying union bound over all $\le 2^n$ partitions $P$, the probability that a random graph $G$ from $\mu$ has a partition $P'$ such that $d_H(P,P')\in[ 0.1n,0.9n]$ and $\cutg(P')\le \cutg(P)-1.01\varepsilon n^2$ is at most
    $$n^{-\Omega(1)}+\exp(n)\cdot \exp(-\omega(n))=n^{-\Omega(1)}$$
\end{proof}

Regarding the minimum bisection, the previous lemma only compares the cut size of the underlying partition $P$ with other cuts. However, the partition $P$ is not necessarily a bisection. We address this by showing that the minimum bisection is close to the cut size of $P$ with high probability.

\begin{fact}
    \label{fact:strong_mbis_gap}
    Let $(P,G)\leftarrow \mu$. With probability $\ge 1-n^{-\Omega(1)}$ the minimum bisection of $G$ is at most $\cutg(P)+0.001\varepsilon n^{2}$
\end{fact}

\begin{proof}
    As mentioned above, with high probability the Hamming weight of $P$ falls in $[n/2-\sqrt n\log n,n/2+\sqrt n\log n]$. Without loss of generality suppose the Hamming weight of $P$ is $\ge n/2$. Let $s:=|P|-n/2$. We let $P'$ be the bisection obtained by flipping exactly $s$ $1$s in $P$. Then we have
    $$\begin{aligned}
        \mbbe_G[\cutg(P')-\cutg(P)]= &s\cdot \frac n2\cdot (1/2+\rho)-s\cdot (\frac n2-s)(1/2-\rho)\\
        =&ns\rho+s^2/2-\rho s^2\\
        =&\Theta(\varepsilon n\sqrt n\log n+n\log^2 n)\\
        =&\tilde \Theta(\varepsilon n\sqrt n)
    \end{aligned}$$

    Notice that $\cutg(P')-\cutg(P)$ is a sum of independent random variables from $\{-1,0,1\}$.
    By Hoeffding's inequality, we get that the probability $\cutg(P')>\cutg(P)+0.001\varepsilon n^2$ is at most $\exp(-\Omega(\varepsilon^2 n^{5/2}\log^{-1}n))=\exp(-\omega(n))$.
    Applying union bound over all $2^n$ partitions $P$, the probability $\cutg(P')$ is much larger than $\cutg(P)$ is at most $n^{-\Omega(1)}$. Hence the minimum bisection is also small with high probability.
\end{proof}

\begin{proof}[Proof to Theorem~\ref{thm:strong_tight_mcut_mbis}]
    We respectively prove the lower bounds for the max cut and the minimum bisection.

    \paragraph{Lower bound to max cut.} We will prove an $\Omega(n/\varepsilon^{2})$ bound for the max cut when the complement of the input graph is drawn from $\mu$. Let $(P,\overline G)\leftarrow \mu$.
    Let $\mcala\subseteq \{0,1\}^n$ be the set of partitions with cut size at least $\cutg(P)-\varepsilon n^2$.
    Let $\mcala_P=\{P'\in\{0,1\}^n:d_H(P,P')\in[0,0.1n)\cup (0.9n,n]\}$. Then for every $P'$ and $P''$, we have
    $$P'\in \mcala_{P''}\Longleftrightarrow P''\in \mcala_{P'}$$
    In addition, by Lemma~\ref{lem:strong_hamming_dis_mcut}, with high probability $\mcala\subseteq \mcala_P$. If one cannot output an element of $\mcala$ with high probability, it also cannot output an element of $\mcala_P$ with high probability.

    By Yao's minimax principle and by Markov's inequality, we instead prove a $q:=10^{-8}\cdot n/\varepsilon^{2}$ worst-case query lower bound for deterministically approximating the max cut to within error $\varepsilon n^2$ with error probability $\le 0.991$.

    We apply our generalized Fano's inequality. Specifically, we let $\sigma_q$ to denote the query history made by $\Pi'$. We use $Y$ to denote the output of $\Pi'$. Then $P\rightarrow \sigma_q\rightarrow Y$ is a Markov chain. By Lemma~\ref{lem:strong_fano} and the fact that $\binom ab\le (\frac {ea}b)^b$, we have
    $$
    \begin{aligned}
        &p_e\cdot n+(1-p_e)(0.1n\log (10e)+\log n)+H(p_e)\\
        \ge&p_e\cdot n+(1-p_e)\sup_{P'}\log(|\mcala_{P'}|)+H(p_e)\\
        \ge& H(P|\sigma_q)=H(P)-I(P;\sigma_q)=n-I(P;\sigma_q)\\
        \ge& 0.9975n
    \end{aligned}
    $$
    where $p_e$ denotes the error probability of $\Pi$. By solving the inequality, we get 
    $$p_e\ge 0.995>0.991$$
    Hence the query complexity of approximating the max cut to within error $\varepsilon n^2$ is $\Omega(n/\varepsilon^{2})$.

    \paragraph{Lower bound to minimum bisection.} The lower bound to minimum bisection follows a similar proof outline. We let $(P,G)\leftarrow \mu$. Let $\mcala$ be the set of bisections with cut size at most $\varepsilon n^2$ far from the minimum bisection. Still let $\mcala_P=\{P'\in \{0,1\}^n:d_H(P,P')\in [0,0.1n)\cup (0.9n,n]\}$. By Lemma~\ref{lem:strong_hamming_dis_mcut} and Fact~\ref{fact:strong_mbis_gap}, $\mcala\subseteq \mcala_P$ with high probability. The following part of the proof is the same as the proof for max cut. We omit it here.
\end{proof}

\section{Chernoff bound for conditionally bounded bits}
\label{appendix:chernoff_bound}

\begin{lemma}[Chernoff bound for conditionally bounded bits]
    Fix $p\in(0,1)$ and $\delta>0$. Let $X_1,\dots, X_n\in \{0,1\}$ be random variables such that for every $i\in[n]$ and every $a_1,\dots, a_{i-1}\in\{0,1\}$ such that $\sum_{t=1}^{i-1}a_{t}<(1+\delta)\cdot pn$,
    $$\Pr[X_i=1|X_1=a_1,\dots,X_{i-1}=a_{i-1}]\le p$$
    Let $X=\sum_{i=1}^n X_i$. Then
    $$\Pr[X\ge (1+\delta)\cdot pn]\le 2\cdot \exp\left(-\frac{\delta^2}{2+\delta}\cdot pn\right)$$
\end{lemma}

\begin{proof}
    Let $k=(1+\delta)\cdot pn$. We use a coupling argument. Let $U_1, \dots, U_n$ be a sequence of i.i.d. random variables from the uniform distribution on $[0,1]$. We define three corresponding sequences of binary random variables, $\{X'_i\}$, $\{Y'_i\}$, and $\{Y_i\}$, for $i \in [n]$.

    The sequence $\{X'_i\}$ is constructed to have the same joint distribution as $\{X_i\}$:
    $$
        X'_i = 1 \iff U_i \le \Pr[X_i=1|X'_1,\dots,X'_{i-1}].
    $$
    The sequence $\{Y'_i\}$ is defined as:
    $$
        Y'_i = 1 \iff (U_i \le p) \lor \left(\sum_{t=1}^{i-1} Y'_t \ge k\right).
    $$
    Finally, $\{Y_i\}$ is a sequence of i.i.d. Bernoulli bits:
    $$
        Y_i = 1 \iff U_i \le p.
    $$
    Let $X' = \sum_{i=1}^n X'_i$, $Y' = \sum_{i=1}^n Y'_i$, and $Y = \sum_{i=1}^n Y_i$. By construction, $\Pr[X \ge k] = \Pr[X' \ge k]$.

    We first prove by induction that $X'_i \le Y'_i$ for all $i \in [n]$. The base case $i=1$ holds as $\Pr[X_1=1] \le p$. For the inductive step, assume $X'_t \le Y'_t$ for all $t < i$. If $\sum_{t<i} Y'_t \ge k$, then by definition $Y'_i=1$, so the inequality $X'_i \le Y'_i$ holds trivially. If $\sum_{t<i} Y'_t < k$, then by the inductive hypothesis, $\sum_{t<i} X'_t < k$. The given condition guarantees that $\Pr[X_i=1|X'_1,\dots,X'_{i-1}] \le p$. In this case, if $X'_i=1$, it implies $U_i \le p$, which in turn implies $Y'_i=1$. Consequently, $X'$ is stochastically dominated by $Y'$, which gives
    $$
        \Pr[X' \ge k] \le \Pr[Y' \ge k].
    $$
    
    Next, we show that $\Pr[Y' \ge k] = \Pr[Y \ge k]$. The inequality $\Pr[Y' \ge k] \ge \Pr[Y \ge k]$ follows from $Y'_i \ge Y_i$ for all $i$. For the other direction, fix any assignment to $\{U_i\}$ where $Y' \ge k$. Let $\tau = \min\{i \mid \sum_{t=1}^i Y'_t \ge k\}$. For the sum to first reach $k$ at time $\tau$, we must have $Y'_\tau=1$ and $\sum_{t=1}^{\tau-1}Y_t=\sum_{t=1}^{\tau-1} Y'_t = k-1 < k$ by our construction. This implies that $U_\tau \le p$. Thus, $Y_i=1$ and $\sum_{t=1}^i Y_t=k$. Therefore, $\Pr[Y' \ge k] \le \Pr[Y \ge k]$.

    Combining these results, we have $\Pr[X \ge k] \le \Pr[Y \ge k]$. Since $Y$ is a sum of $n$ i.i.d. $\text{Bernoulli}(p)$ random variables, the inequality follows from a standard multiplicative Chernoff bound.
\end{proof}

\end{document}